\documentclass[aps,prd,nofootinbib]{revtex4}
\usepackage[utf8]{inputenc}
\usepackage{lmodern}
\usepackage{amsmath}
\usepackage{amsthm}
\usepackage{amsfonts}
\usepackage{graphicx}
\usepackage[caption=false]{subfig}
\usepackage[english]{babel}

\def\SPZ{SP$\zeta$}

\graphicspath{{./figures/}}

\newcommand{\pd}[2]{\frac{\partial{#1}}{\partial{#2}}}
\newcommand{\der}[2]{\frac{\md{#1}}{\md{#2}}}
\newcommand{\md}{\mathrm{d}}
\newcommand{\D}{D}  
\newcommand{\av}[1]{\left\langle{#1}\right\rangle}
\newcommand{\R}{\mathbb{R}}   
\newcommand{\C}{\mathbb{C}}   
\newcommand{\Z}{\mathbb{Z}}   
\DeclareMathOperator\tr{Tr}
\DeclareMathOperator\Res{Res}
\DeclareMathOperator\sinc{sinc}
\newcommand{\Vol}{\mathrm{Vol}} 
\newcommand{\dd}{{\mathrm d}}      

\newtheorem{theorem}{Theorem}
\newtheorem{lemma}[theorem]{Lemma}

\begin{document}

\title{Spectral estimators for finite non-commutative geometries}
\date{3 June 2019}                          
\author{John W. Barrett$^1$, Paul Druce$^1$, Lisa Glaser$^2$}
\affiliation{$^1$ University of Nottingham, $^2$ Radboud University Nijmegen}

\begin{abstract} A finite non-commutative geometry consists of a fuzzy space together with a Dirac operator satisfying the axioms of a real spectral triple. This paper addresses the question of how to extract information about these geometries from the spectrum of the Dirac operator. Since the Dirac operator is a finite-dimensional matrix, the usual asymptotics of the eigenvalues makes no sense and is replaced by measurements of the spectrum at a finite energy scale. The spectral dimension of the square of the Dirac operator is improved to provide a new spectral measure of the dimension of a space called the spectral variance. Similarly, the volume of a space can be computed from the spectrum once the dimension is known. Two methods of doing this are investigated: the well-known Dixmier trace and a recent improvement due to Abel Stern. Finally, the distance between two geometries is investigated by comparing the spectral zeta functions using the method of Cornelissen and Kontogeorgis. All of these techniques are tested on the explicit examples of the fuzzy spheres and fuzzy tori, which can be regarded as approximations of the usual Riemannian sphere and flat tori. Then they are applied to characterise some random fuzzy spaces using data generated by a Monte Carlo simulation.
\end{abstract}
\maketitle

\section*{Introduction}

A Riemannian geometry can be specified by giving a Dirac operator on the bundle of spinors~\cite{connes2013spectral}. According to this point of view, the Dirac operator is regarded as the fundamental data and the metric tensor field can be extracted from it. Non-commutative geometry then extends the idea of a geometry by generalising the algebra of functions on the manifold to a non-commutative algebra. There are no longer any points in this non-commutative `space', yet various aspects of geometry survive in a more abstract form. In particular, there is still a Dirac operator and its spectrum provides useful information about the geometry. This notion of non-commutative geometry has proven to be particularly useful in particle physics, where the addition of non-commutative extra dimensions to the usual (commutative) space-time allows for an elegant description of the standard model~\cite{barrett_lorentzian_2007,connes_noncommutative_2006}.

This work focusses on the idea that space-time itself might be non-commutative, using a simplified setting in which space-time is both Euclidean signature and compact. It investigates non-commutative geometries as analogues of compact Riemannian manifolds and asks how one can measure, or compare, geometric quantities derived from the spectrum of the Dirac operator.

The particular non-commutative geometries of interest are called fuzzy spaces and have as the algebra of `functions' the algebra of $N\times N$ matrices. The most famous example is the fuzzy sphere~\cite{podles_quantum_1987,madore_fuzzy_1992}. This is the non-commutative analogue of a 2-sphere in which the fields consist of linear combinations of the spherical harmonics on the sphere up to a certain maximum order.
A fuzzy space has a high-energy cut-off while keeping analogues of the coordinate symmetries.
This makes these spaces attractive as backgrounds for quantum field theory simulations~\cite{oconnor_monte_2006} but currently the limited number of well-understood backgrounds restricts the use of this.

A Dirac operator is defined for a fuzzy space using the mathematical framework of a real spectral triple~\cite{Connes:1995tu}. This spectral triple is finite, which means that the vector spaces are all finite-dimensional. The Dirac operator is also a finite-dimensional matrix and so its spectrum is a finite set.  Dirac operators with a high degree of symmetry are known for the fuzzy sphere~\cite{grosse_dirac_1995} and, in a more recent discovery, for fuzzy tori of various shapes~\cite{barrett_gaunt}.

Random non-commutative geometries allow the Dirac operator $D$ to vary, subject to the constraints imposed by the axioms. In the context of fuzzy spaces, one can integrate over the finite-dimensional space of Dirac operators in a well-defined way, giving a very particular type of random matrix model.
This allows an implementation of a Monte Carlo algorithm on these geometries.
Recent work using a simple action showed that this approach can give rise to a phase transition, close to which indications of 2-dimensional manifold-like behaviour were observed~\cite{barrett_monte_2015}. This was done by comparing the spectral density of the random fuzzy space to the spectral density of the fuzzy sphere by eye, but clearly quantitative tools are necessary. More generally, it raises the question of how to recover geometric information about a random fuzzy space.

This paper investigates ways to extract the dimension and the volume of a fuzzy space from the spectrum of the Dirac operator by applying methods that are adapted from the case of the Dirac operator on a Riemannian manifold. The guiding idea is that  if the fuzzy space approximates a certain Riemannian manifold, then one expects the dimension and volume of the fuzzy space to approximate the Riemannian ones. The paper also investigates a method to measure how close two geometries are by comparing their spectral zeta functions.

The various methods are tested on the fuzzy sphere and on some fuzzy tori, where the spectrum is known exactly and the relation to the continuum sphere and tori is clear. Then the methods are applied to random fuzzy spaces, obtaining some interesting new insights. A new difficulty in the random case is that the variance of geometric quantities can be significant, leading to a question of how to  carry out the averaging and also to the tricky issue of what an averaged geometry actually means.

In Section~\ref{sec:fuzzy} the fuzzy spaces that will be explored are introduced. Section~\ref{sec:dim} discusses the problem of measuring the dimension of a fuzzy space.  Weyl's law for a manifold relates the asymptotics of the eigenvalues of the Dirac operator to the dimension and volume of the manifold. As a warm-up exercise, this is adapted to a fuzzy space by looking directly at the eigenvalues rather than their asymptotics. It gives some insight into the behaviour of the dimension but the results depend on which part of the spectrum one looks at.

A more reliable method is to use a new quantity introduced in this paper that is called the spectral variance. This is a measure of dimension that depends on a parameter that determines the energy scale at which the dimension is measured. It is an adaptation of the notion of spectral dimension, which was first used in~\cite{ambjorn_spectral_2005} to study the model of causal dynamical triangulations, and has since been modified to explore many other models (see~\cite{carlip_small_2009} for a fairly comprehensive review). The spectral dimension is usually used with the Laplace-Beltrami operator on scalar fields; it is calculated from the heat kernel and exploits Weyl's law of eigenvalue scaling. In this paper, the spectral dimension is defined using the fermion fields and the operator $D^2$. The crucial difference is that although $D^2$ is a Laplace-type operator, it does not necessarily have zero as an eigenvalue and as a consequence the spectral dimension behaves badly.

The spectral variance is defined to remedy this defect. It is a variance of the squared eigenvalues, and hence only differences of squared eigenvalues matter. Another way of thinking of the spectral variance is that it is the heat capacity of an ideal gas `in' the geometric space. The spectral variance is applied to analysing properties of the averages of the random fuzzy spaces introduced in~\cite{barrett_monte_2015} and provides further evidence about the nature of the phase transition. It confirms that the transitions are from a phase that is approximately 1-dimensional to a phase that has significantly higher dimension over a large range of energies.

Measures of the volume of a fuzzy space are considered in section~\ref{sec:zeta}. Both the Dixmier trace and a new volume measure defined by Stern~\cite{stern_finite-rank_2017} are adapted to the finite setting and their properties compared. These measures provide sensible results as long as the dimension of the space is known or estimated as a single number. However, the scale-dependent dimension provided by the spectral variance does lead to difficulties in interpreting the results, since the notion of volume is dimension-dependent. Also, it is not clear how one should apply a volume measure to an ensemble where the dimension varies significantly. These difficulties are discussed in this section.

Section~\ref{sec:zetadist} considers the distance between two spectra determined by comparing  their spectral zeta functions using a method due to Cornelissen and Kontogeorgis~\cite{cornelissen_distances_2017}. This method applies equally well to finite and infinite spectra and so it is possible to compute the distance between a fuzzy sphere and a Riemannian sphere, for example. After testing the formalism on spheres and tori, the distances between the average spectra of some random geometries are computed. Also, the distance to the fuzzy sphere reveals that the average geometries are indeed close to the fuzzy sphere near the phase transition, as first suggested in~\cite{barrett_monte_2015}. The conclusions of the paper are presented in section~\ref{sec:conclusion}.

\section{Fuzzy spaces}\label{sec:fuzzy}

A fuzzy space is the generic term for a non-commutative geometry where the algebra is an algebra of matrices. Here it is assumed that this is the algebra $M_N(\C)$ of all $N\times N$ matrices over $\C$. The geometry on a fuzzy space is specified by a differential operator, originally a Laplace operator, but in this paper it is always the Dirac operator. The Dirac operator is defined by making the fuzzy space into a finite non-commutative spectral triple.
This consists of a finite dimensional Hilbert space $\mathcal{H}$, a non-commutative algebra $\mathcal{A}$ represented on $\mathcal{H}$ and a self-adjoint operator $D$ acting on $\mathcal{H}$ which is the Dirac operator. The fuzzy geometries in question all follow the framework outlined in~\cite{barrett_matrix_2015}, which is that $\mathcal{A} = M_N(\C)$ and $\mathcal{H}=V\otimes M_N(\C)$, with $V=\C^k$ carrying the action of a type $(p,q)$ Clifford module generated by the $k\times k$ matrices $\{\gamma^i\}$, $p$ of which square to $1$ and the remaining $q$ of which square to $-1$.

The simplest example of a finite non-commutative geometry is the fuzzy $2$-sphere. For the continuum (Riemannian) $2$-sphere  $S^2$ of unit radius and round metric, the eigenvalues of the Dirac operator are all the integers except $0$, with a multiplicity of $2 |j|$ for eigenvalue $j$~\cite{GraciaBondia:2001gta}.

The fuzzy sphere is studied as a real spectral triple in \cite{dandrea2013metric} and \cite{barrett_matrix_2015}. The former gives a spectral triple for the Grosse-Presnajder Dirac operator \cite{grosse_dirac_1995}. The problem with this is that the Clifford type is $(0,3)$, for which numerical computations of random Dirac operators are difficult due to the asymmetry in the spectrum.
For this reason, the fuzzy sphere used in this paper is the one studied in~\cite{barrett_matrix_2015}, for which  the spinor space is $V=\C^4$ and the Clifford type is $(1,3)$. It is not necessary to give the explicit formula for the Dirac operator, since all that is needed here is the spectrum. The eigenvalues are the integers $-(N-1)\le j\le (N-1)$ excluding zero, with the multiplicity $ 4|j|$,
and the values $j=-N,N$ with multiplicity $2N$. Thus, with the exception of the largest eigenvalues, the spectrum has double the multiplicity of the corresponding eigenvalues for the continuum sphere. It is worth noting that the spectrum of $|D|$ coincides with that of the Grosse-Presnajder operator (with doubled multiplicities), so for practical purposes either Dirac operator could be used.

Another simple geometry for which a good description in terms of a finite spectral triple exists is the fuzzy torus~\cite{barrett_gaunt}.
In the continuum, a 2-dimensional torus $T^2$ is created by identifying a square with corner points $(0,0),(0,2\pi),(2\pi,0),(2\pi,2\pi)$ along opposing edges. It is given the flat metric
\begin{equation}\dd s^2=(a^2+c^2)\,\dd\theta^2+2(ab+cd)\,\dd\theta\dd\phi+(b^2+d^2)\,\dd\phi^2.\label{metric}
\end{equation}
The torus has four spin structures, which are labelled $\Sigma=(0,0)$, $(1,0)$, $(0,1)$ or $(1,1)$.
The Dirac eigenvalues are
\begin{align}
\lambda_{k,l}=\pm\frac1{ad-bc} \sqrt{ \left({ak-bl}\right)^2 + \left({dl-ck} \right)^2}\label{eq:contTorus}
\end{align}
with $(k,l) \in \mathbb{Z}^2+\Sigma/2$.
Thus the values of $k$ and $l$ are either integers or integers plus one half, depending on the spin structure.

The fuzzy torus (in its simplest form) is described by a finite spectral triple with $\mathcal{A} = M_N(\C)$ and four integers $a$, $b$, $c$, $d$ which determine an invertible matrix $\begin{pmatrix}a&b\\c&d\end{pmatrix}$. The spinor space is again $V=\C^4$ but the Clifford type is $(0,4)$. It has Dirac eigenvalues given by
\begin{equation}
\lambda_{k,l}=\pm\frac1{[ad-bc]_q} \sqrt{ \left[ak-bl\right]_q^2 + \left[dl-ck\right]_q^2}\label{eq:fuzzyTorus}
\end{equation}
where the square brackets denote the `q-numbers'
\begin{equation}[l]_q=\frac{\sin \pi l/N}{\sin\pi/N},\end{equation}
using the parameter $q = e^{2\pi i/N}$ and ignoring an irrelevant sign in~\cite{barrett_gaunt}.
However, this time each geometry has only one spin structure, $\Sigma_c=(a+c,b+d) \mod 2$, which is called the canonical spin structure. Thus $(k,l) \in \mathbb{Z}^2+\Sigma_c/2$. The eigenvalues are periodic in $k$ and $l$, with period $N$. Thus one can index them uniquely by $-N/2\le k,l <N/2$. For each $k$ and $l$, the multiplicity of the eigenvalue is two.

The simplest example is the (unit) square torus, which has $a=d=1$, $b=c=0$, and hence spin structure $\Sigma_c=(1,1)$. Then using integers $k'$ and $l'$, the eigenvalues are
\begin{equation}
\lambda=\pm \sqrt{ \left[k'+\frac12\right]_q^2 + \left[l'+\frac12\right]_q^2}.\label{eq:squareTorus}
\end{equation}
The other examples used in this paper all have $N$ as a multiple of $ad-bc$. In this case, the eigenvalues repeat themselves $ad-bc$ times as $k$ and $l$ vary. For other values of $N$ the behaviour is more complicated, so these cases are not explored here.

The last class of fuzzy spaces examined is the class of random fuzzy spaces defined in~\cite{barrett_monte_2015}. The Dirac operator has an explicit expression in~\cite{barrett_matrix_2015} using several $N\times N$ Hermitian matrices $H_i$ and anti-Hermitian traceless matrices $L_j$. The Dirac operator is allowed to vary, defining an ensemble of geometries governed by a partition function.
The independent entries of the matrices, $H_i,L_j$, are then the free data for the Dirac operator and the ensemble can be explored using Monte Carlo moves that change them.  For more detail see~\cite{barrett_monte_2015}.
To decide which Monte Carlo moves to accept, the action
\begin{align}
\mathcal{S}= g_4 \tr( \D^4) +g_2 \tr(\D^2)\label{eq:action}
\end{align}
was used.
This is the simplest possible non-trivial action, using the lowest powers of $D$ (the trace of odd powers of $D$ is zero in the examples explored here).
For particular values of $g_4$ and $g_2$ it can be derived from the lowest order expansion of the heat kernel.

With $g_4=1$ and varying $g_2<0$, the random fuzzy spaces show a phase transition, which was described in more detail in~\cite{barrett_monte_2015,glaser_scaling_2016}.
The location of the phase transition depends on the Clifford module type $(p,q)$. It was suggested that close to the phase transition the density of the eigenvalue spectrum has some similarity to that of the fuzzy sphere, leading to speculation that the geometries at the phase transition might be $2$-dimensional.

In~\cite{glaser_scaling_2016}, the geometries of type $(1,1),(2,0),(1,3)$ were explored in more detail.
A particular focus was to determine the phase transition points and try to better understand the behaviour of the geometries at the phase transition.
The datasets generated for~\cite{glaser_scaling_2016} are the ones used in this paper.
The phase transition was determined to be at $g_2=-2.4,-2.8,-3.7$ for the types $(1,1),(2,0),(1,3)$ respectively.
There was also good indication, including strong correlations between different eigenvalues, that these phase transitions are higher than first order.

Throughout, $g_2$ values less (more negative) than the phase transition value $g_c$  are referred to as \emph{after} the phase transition and the system is then in a gapped phase.
Likewise, $g_2$ values greater than $g_c$ are referred to as \emph{before} the phase transition, and the system is in an ungapped phase. So the system undergoes the phase transition as $g_2$ gets more negative.

\section{Dimension measures}\label{sec:dim}

In ordinary geometry there are many different definitions of the dimension of a space that are equivalent.
The simplest is the minimum number of coordinates on a space needed to specify any point within it.
With the discovery of fractal spaces, new notions of dimension that need not take integer values were defined.
In non-commutative geometry, the dimension spectrum is defined as the set of singularities of the zeta function of the Dirac operator, and is a subset of the complex plane.

These definitions tend to be trivial when applied to finite spectral triples.
However, finite non-commutative approximations to the 2-sphere and 2-torus exist and a dimensional measure that captures the 2-dimensionality of these non-commutative geometries is desirable.
As more examples of fuzzy spaces are uncovered it will be useful to have a notion of dimension that approximates the
dimension of any possible continuum limit.

\subsection{Weyl's law}\label{sec:weyl}
The classic result of Weyl demonstrates that the asymptotic behaviour of the eigenvalues of the scalar Laplacian on subsets of $\R^d$ encodes the dimension and volume of the space.
This result was extended to operators of Laplace type, such as the square of the Dirac operator on a spin manifold~\cite{Roe:1998vp}.
Let $D$ be the Dirac operator on a compact oriented Riemannian manifold of dimension $d$ and volume $V$, with spin bundle of rank $k$. Let $\lambda_n$ be the $n$-th eigenvalue, ordered so that $|\lambda_n|$ is non-decreasing.
Weyl's asymptotic formula is
\begin{equation}\frac n{|\lambda_n|^d} \to \frac{k V}{(4\pi)^{d/2}\Gamma(1+d/2)}
\end{equation}
as $n\to\infty$.

Taking the logarithm gives
\begin{equation}\log n- d\log|\lambda_n|\to \text{constant}
\end{equation}
which suggests the use of $\log n/\log|\lambda_n|$ to approximate $d$.
Unfortunately this expression contains two unknowns, the dimension $d$ and the volume-dependent constant.
There is currently no tool to estimate the volume of a finite spectral triple without knowing its dimension, as will be discussed in section~\ref{sec:zeta}, hence leaving the constant unknown.

To remove the unknown constant, consider two different eigenvalues labelled by $n$ and $m$.
Then in a limit in which $n,m\to\infty$
\begin{equation}\log(n/m)-d\log(|\lambda_n|/|\lambda_{m}|)\to 0.
\end{equation}
Assuming $\log(n/m)$ remains bounded in this limit,
\begin{equation}W_{nm}=\frac{\log(n/m)}{\log(|\lambda_n|/|\lambda_m|)}\to d \label{eq:modDWeyl}
\end{equation}
Therefore the left-hand side can be used to estimate $d$ in the finite case, and this estimate is independent of any rescaling of the eigenvalues.

In Figure~\ref{fig:modDsphereWeyl}(a) this is plotted for the fuzzy sphere.
The large multiplicities in the spectra lead to a range of values for $n$ associated to the same eigenvalue $\lambda_n$.
This leads to $W_{nm}$ varying within the rectangular multiplicity blocks, which creates a psychedelic pattern in our plot.
To suppress this phenomenon and instead illustrate the overall trend of the spectra one can choose the middle of the multiplicity block as the value for $n$, denoted $\tilde{n}$, for each value $\lambda_n$, as shown in Figure~\ref{fig:modDsphereWeyl}(b).
With the second definition the fuzzy sphere has a stable notion of dimension which is close to $2$.
This is not surprising as the spectrum of the Dirac operator for the fuzzy sphere is a truncation of the spectrum of the continuum sphere; with increasing matrix size it includes more of the spectrum of the continuum sphere.

This however is not the case for the fuzzy torus, shown in figure~\ref{fig:WeylLawTori} for the square fuzzy torus ($a=d=1$ and $b=c=0$).
The value of $W_{nm}$ depends on which part of the spectrum is examined. This arises because the eigenvalues of the fuzzy torus are in general different from those of the continuum torus.
The reason for this is that the $q$-number $[l]_q$ approximates the integer $l$ for $\pi l /N \ll 1$, but deviates from this, and thus from the classical spectrum, as $l$ increases. In fact, the `correct' continuum value of $2$ for the dimension can only be seen in the lower left corner of the plot, where $m$ and $n$ are small.
\begin{figure}
  \subfloat[$W_{nm}$]{\includegraphics[width=0.45\textwidth]{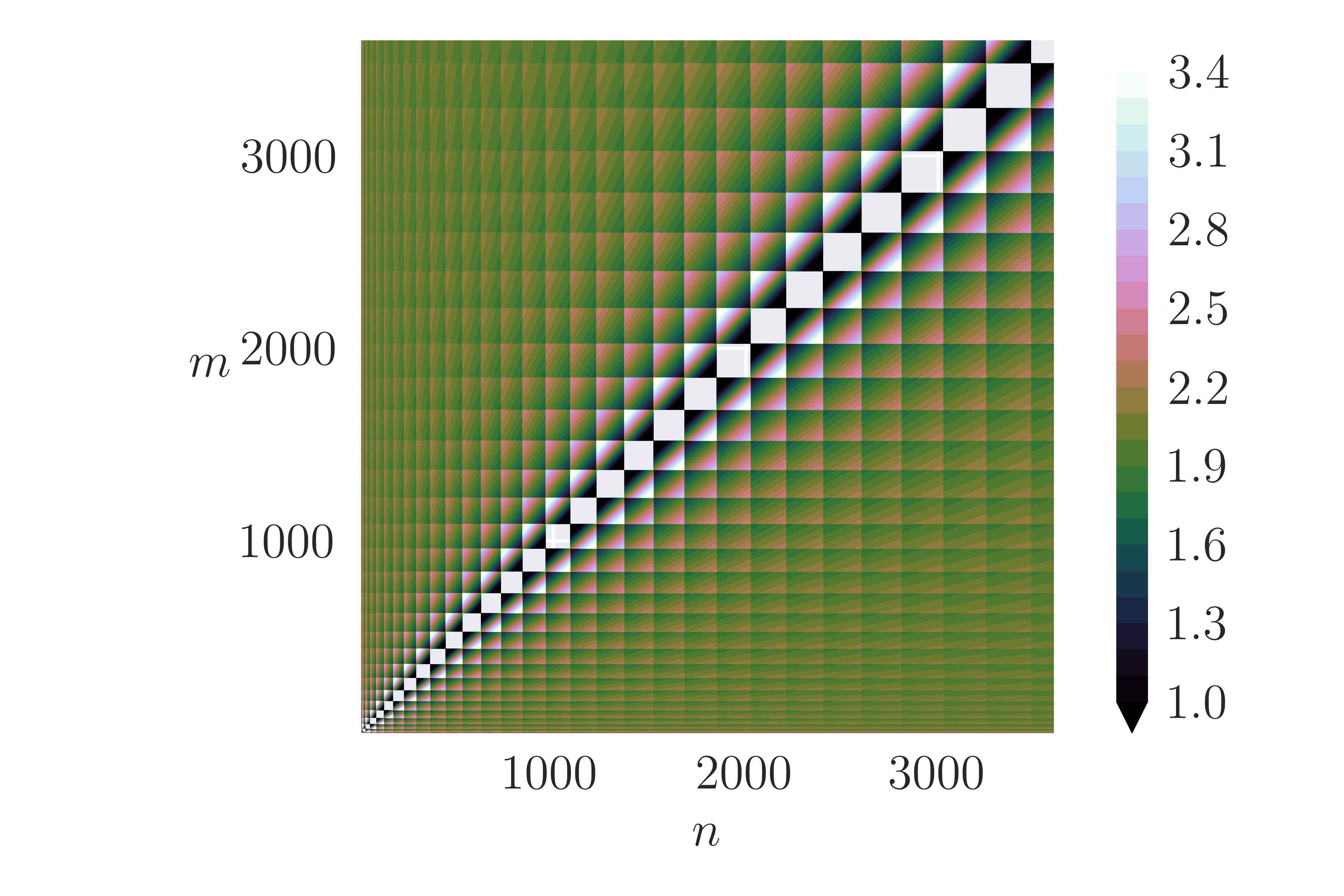}}
  \subfloat[$W_{\tilde{n}\tilde{m}}$ where $\tilde{n}$ is the mid point of the multiplicity of the eigenvalue $\lambda_n$]{
  \includegraphics[width=0.45\textwidth]{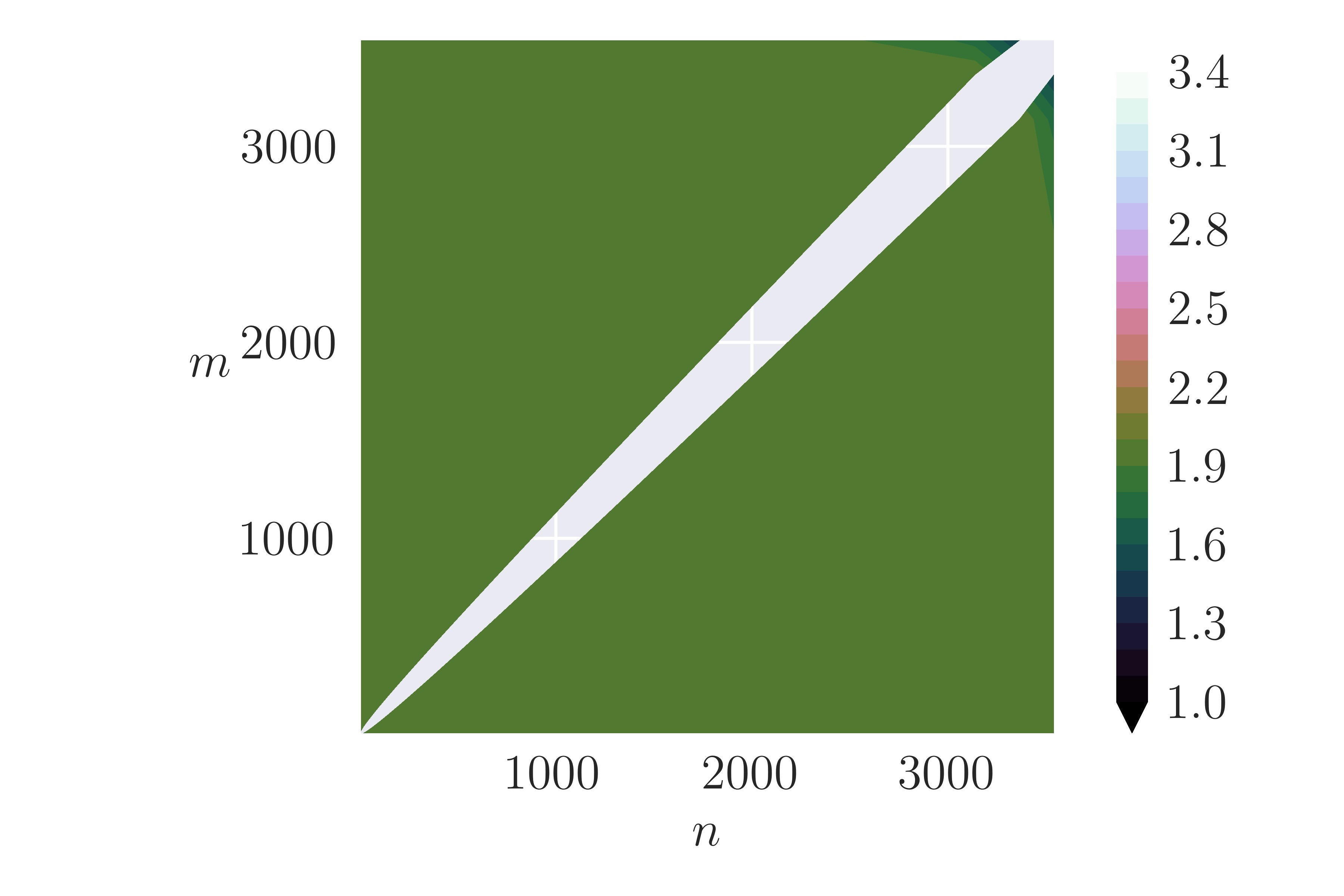}}
  \caption{Weyl's law for the fuzzy sphere for $N=30$: A 2-parameter plot of the estimate of the dimension $W_{nm}$.}
  \label{fig:modDsphereWeyl}
\end{figure}

\begin{figure}
    \centering
\includegraphics[width=0.5\textwidth]{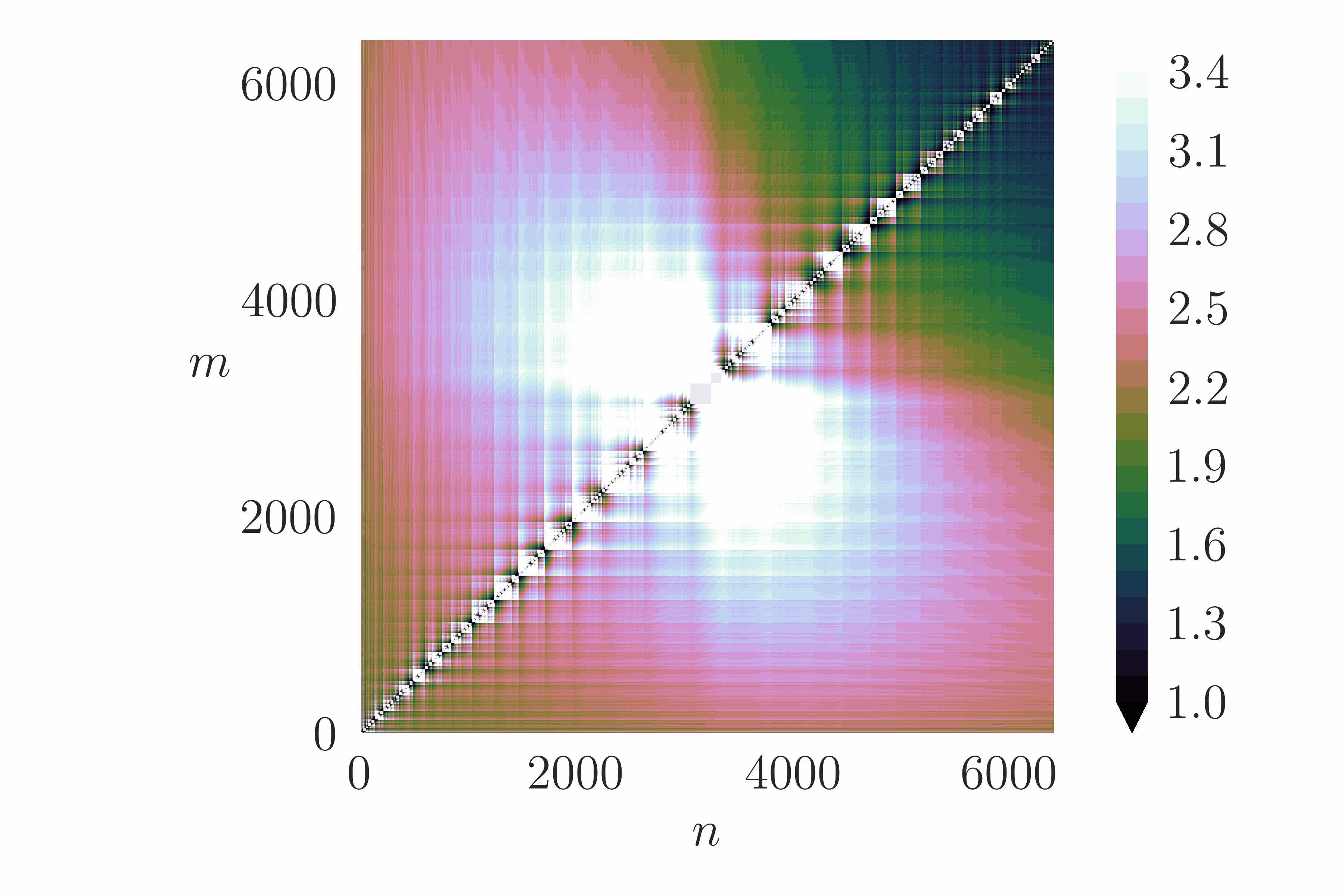}
   \caption{Weyl's law for the square fuzzy torus with $N=40$ and $a=d=1$ $b=c=0$.} \label{fig:WeylLawTori}
\end{figure}

For the random geometries the Dirac operators do not have any high multiplicities and the $W_{nm}$ for various $g_2$ values of the type $(2,0)$ geometries are shown in Figure~\ref{fig:type20WeylLawPT}. One can see that there is a large variation of values depending on which part of the spectrum is looked at. The $W_{nm}$ fall to zero for the largest eigenvalues, where the short wavelength cut-off of the geometry becomes visible. Ignoring the largest eigenvalues, one can see there that $W_{nm}$ takes values between $1.2 - 1.6$ at $g_2=-2.75$, between $1.6-2.0$ at the phase transition $g_2=-2.80$. For $g_2=-2.85$ there is no distinction between the behaviour of the largest eigenvalues suggesting that no stable notion of dimension is present.

While the plots are informative and can be interpreted by eye, there appears to be no good way to extract quantitative information. To get a more robust definition of dimension, other methods are explored in the next sections.

\begin{figure}
 \subfloat[$g_2=-2.75$]{\includegraphics[width=0.32\textwidth]{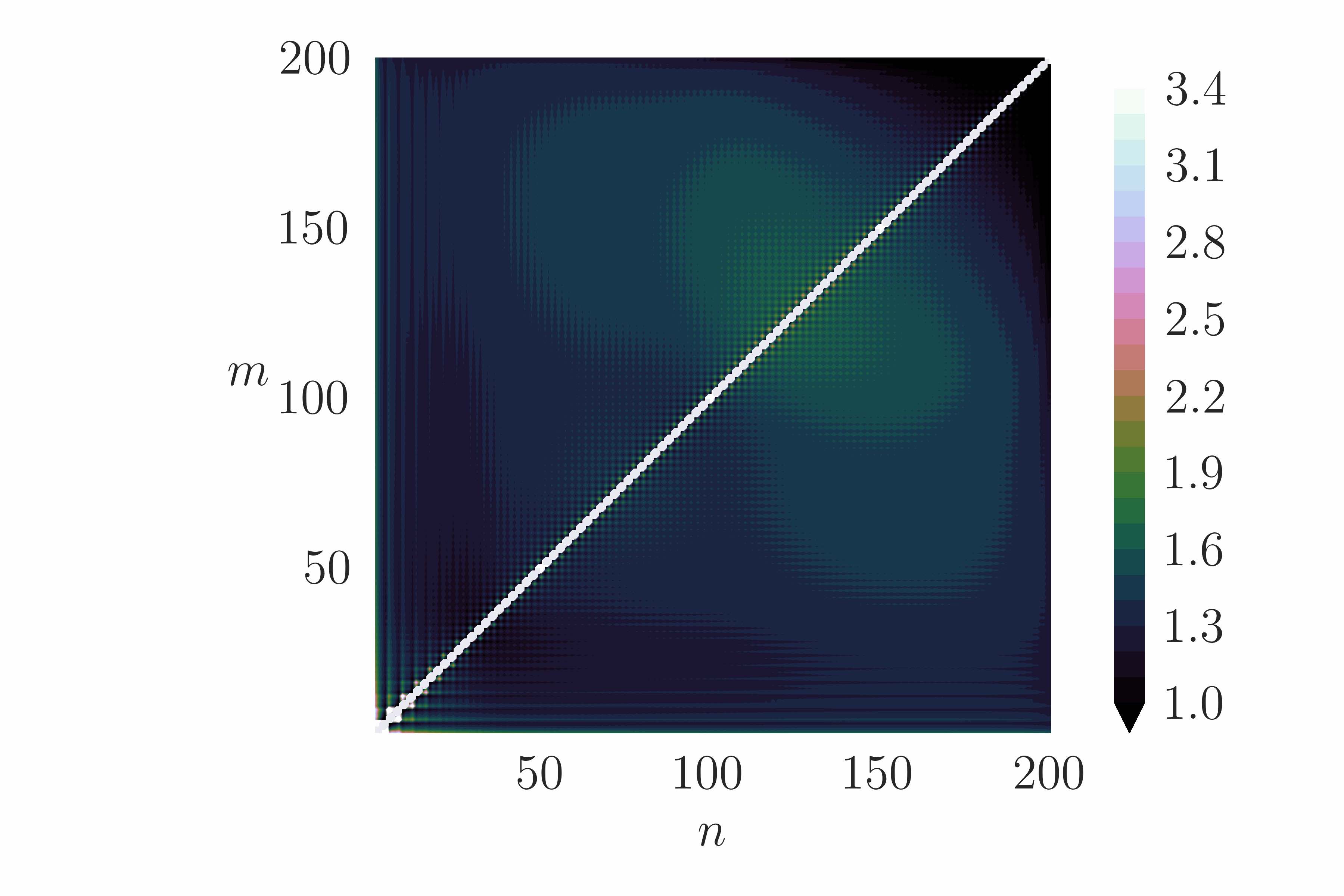}}
  \subfloat[$g_2 = -2.80$, phase transition]{\includegraphics[width=0.32\textwidth]{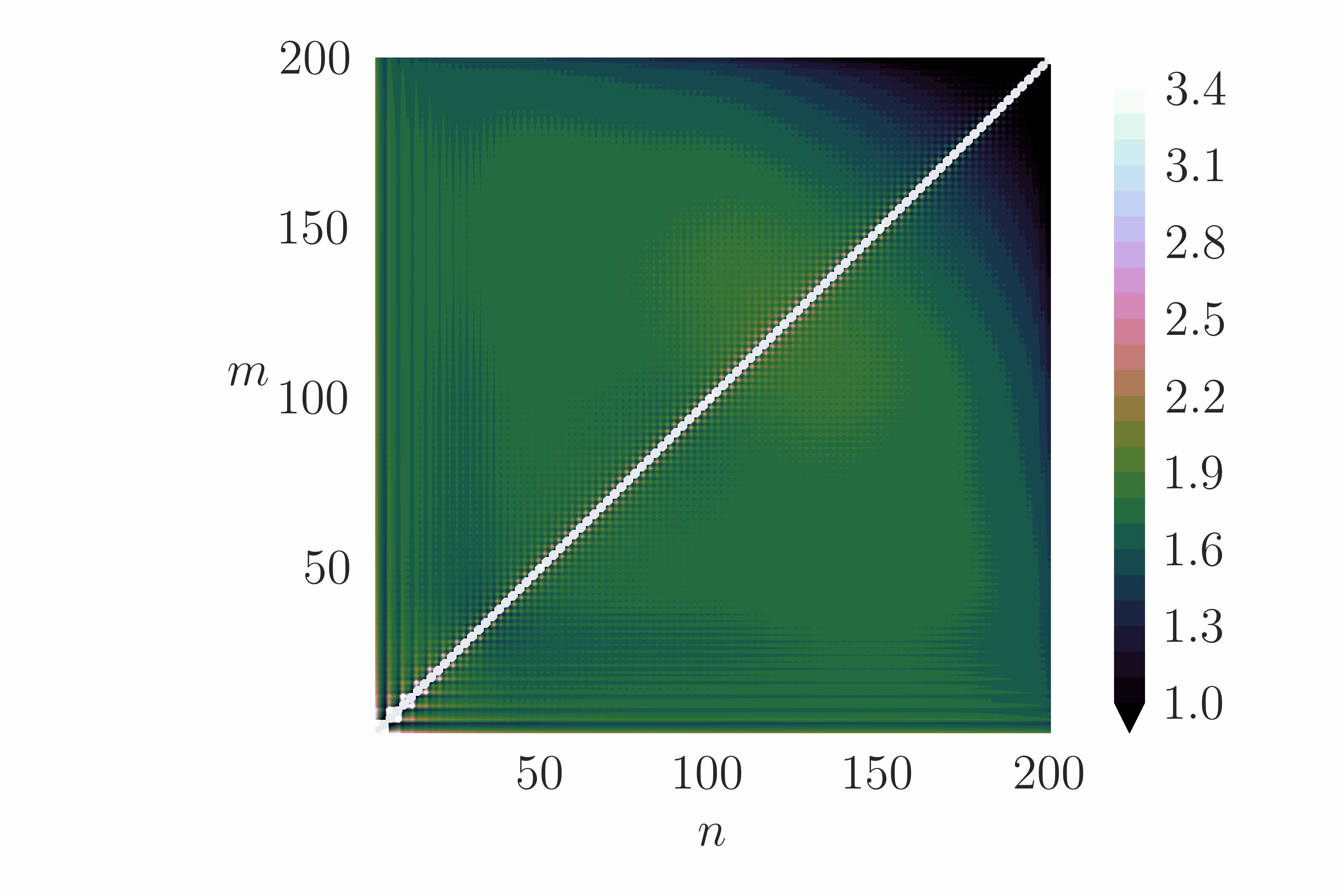}}
\subfloat[$g_2 =-2.85$]{ \includegraphics[width=0.32\textwidth]{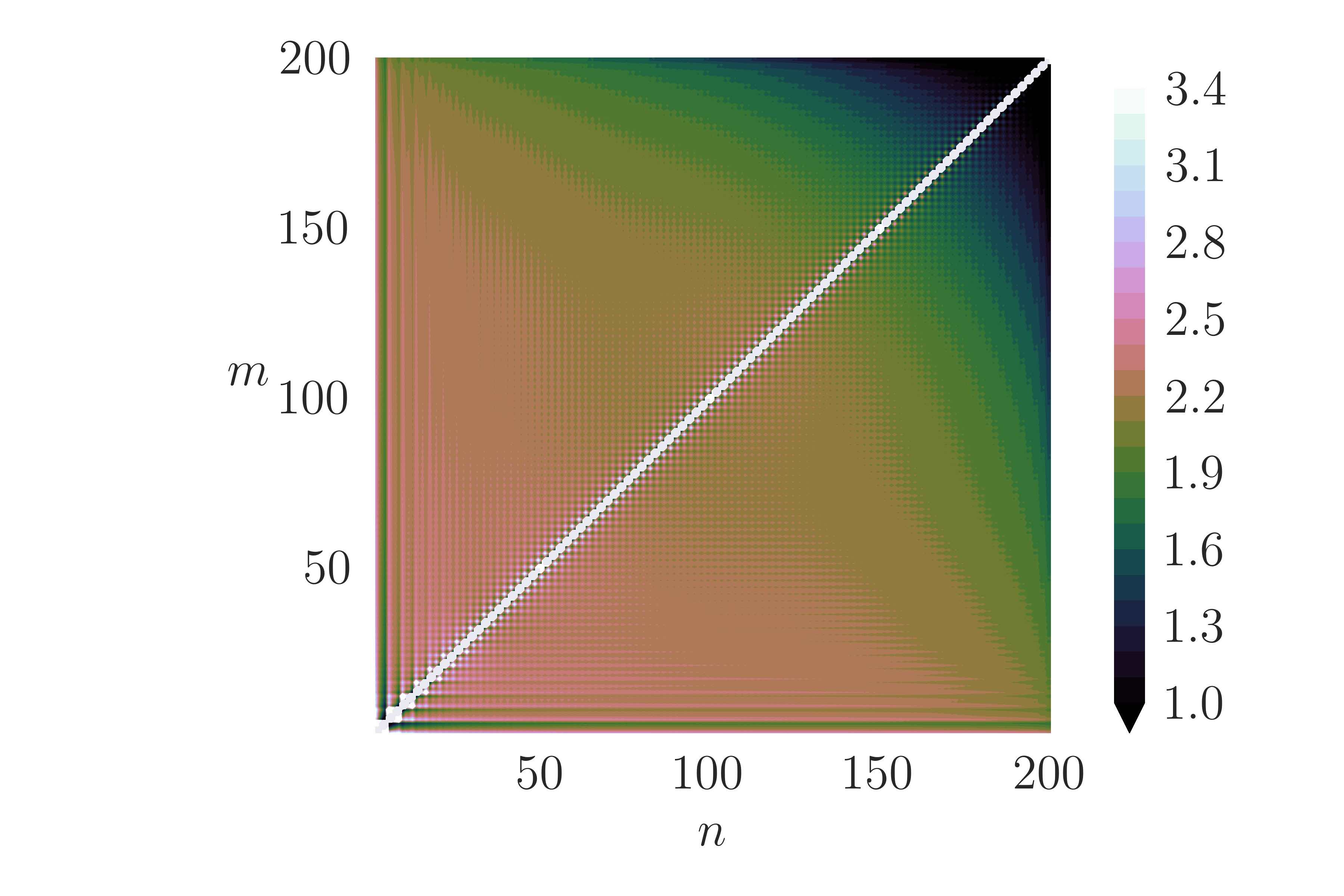}}
  \caption{Type $(2,0)$ random geometries: A 2-parameter plot of the estimate of the dimension $W_{nm}$ for three values of the coupling constant $g_2$ around the phase transition value.}
  \label{fig:type20WeylLawPT}
\end{figure}


\subsection{Spectral dimension}
The spectral dimension~\cite{ambjorn_spectral_2005,benedetti_spectral_2009,Sotiriou:2011aa} is a measure of dimension that depends on a single parameter $t$, which determines a length scale at which the geometry is probed.
It is most commonly defined based on the Laplace-Beltrami operator, but the following formalism can be used for a Dirac operator in both commutative and non-commutative geometry.

Let $D$ be a Dirac operator with a discrete spectrum and denote the eigenvalues by $\lambda$.
A non-negative operator is obtained by squaring the Dirac operator, which results in a Laplace-type operator suitable for the asymptotic analysis.
The heat kernel trace of $D^2$ is defined to be
\begin{align}
K(t)=\tr e^{-t D^2}=\sum_\lambda e^{-t \lambda^2 }.\label{eq:hkt}
\end{align}
The spectral dimension of the Dirac operator is then defined by
\begin{align}
d_s(t)&=-2t \pd{\log[K(t)]}{t} = 2t\, \frac{\sum_\lambda \lambda^2 e^{-t \lambda^2}}{\sum_\lambda e^{-t \lambda^2\mathstrut}}\;.\label{eq:specdim}
\end{align}
The value of the spectral dimension depends on the scale at which it is measured.
One can see from~\eqref{eq:specdim} that at fixed $t$, a term in the numerator $2 t\lambda^2 e^{-t \lambda^2}$ has a maximum at $\lambda=t^{-1/2}$.
Hence the spectral dimension at parameter $t$ is most sensitive to the eigenvalue distribution in some range around the value $\lambda=t^{-1/2}$.

For a smooth compact Riemannian manifold of dimension $d$ without boundary, the heat kernel trace~\eqref{eq:hkt} has an asymptotic expansion~\cite{Vassilevich:2003ik}
as $t\to0$
\begin{equation}
K(t)\sim t^{-d/2}\bigl(a_0+a_2t+a_4t^2+\ldots\bigr).\label{eq:HKexp}
\end{equation}
As a consequence, one can determine the dimension of the manifold from the small $t$ behaviour of $K$.
Calculating the spectral dimension shows that $d_s(0)=d$, the dimension of the manifold.
In fact, the first term in~\eqref{eq:HKexp} on its own contributes a constant value $d$ to $d_s(t)$, but due to the $a_2, a_4,\ldots$  terms, which depend on the curvature of the manifold, the spectral dimension typically deviates from the value $d$ as $t$ increases away from $0$.

A simple example to test this is flat space. This has a continuous spectrum so the sum is replaced by an integral. The spectral dimension for flat space is the constant function $d$. This can be modified with a high-energy cut-off on the spectrum at a maximum $|\lambda|=\Lambda$, which can be viewed as the space having a minimum length scale $\sim 1/\Lambda$.
The spectral dimension for this cut-off operator is
\begin{align}
d_s(t)&=d-\frac{2 e^{-\Lambda^2 t} { \left(\Lambda^2 t\right) }^{d/2}}{\Gamma \left(\frac{d}{2}\right)-\Gamma \left(\frac{d}{2},\Lambda^2 t\right)} \;,
\end{align}
where $\Gamma(a,z)=\int_z^{\infty} t^{a-1} e^{-t} \md t$ is the incomplete $\Gamma$ function.
This is illustrated for different $d$ in Figure~\ref{fig:flatDs}. The high energy cut-off implies that the spectral dimension is $0$ at $t=0$, from where it rises smoothly until it converges towards $d$ for large enough $t$. This example shows that with a cut-off  spectrum, the small $t$ behaviour of the spectral dimension is no longer sufficient to estimate the dimension.
\begin{figure}
\centering
\includegraphics[width=0.7\textwidth]{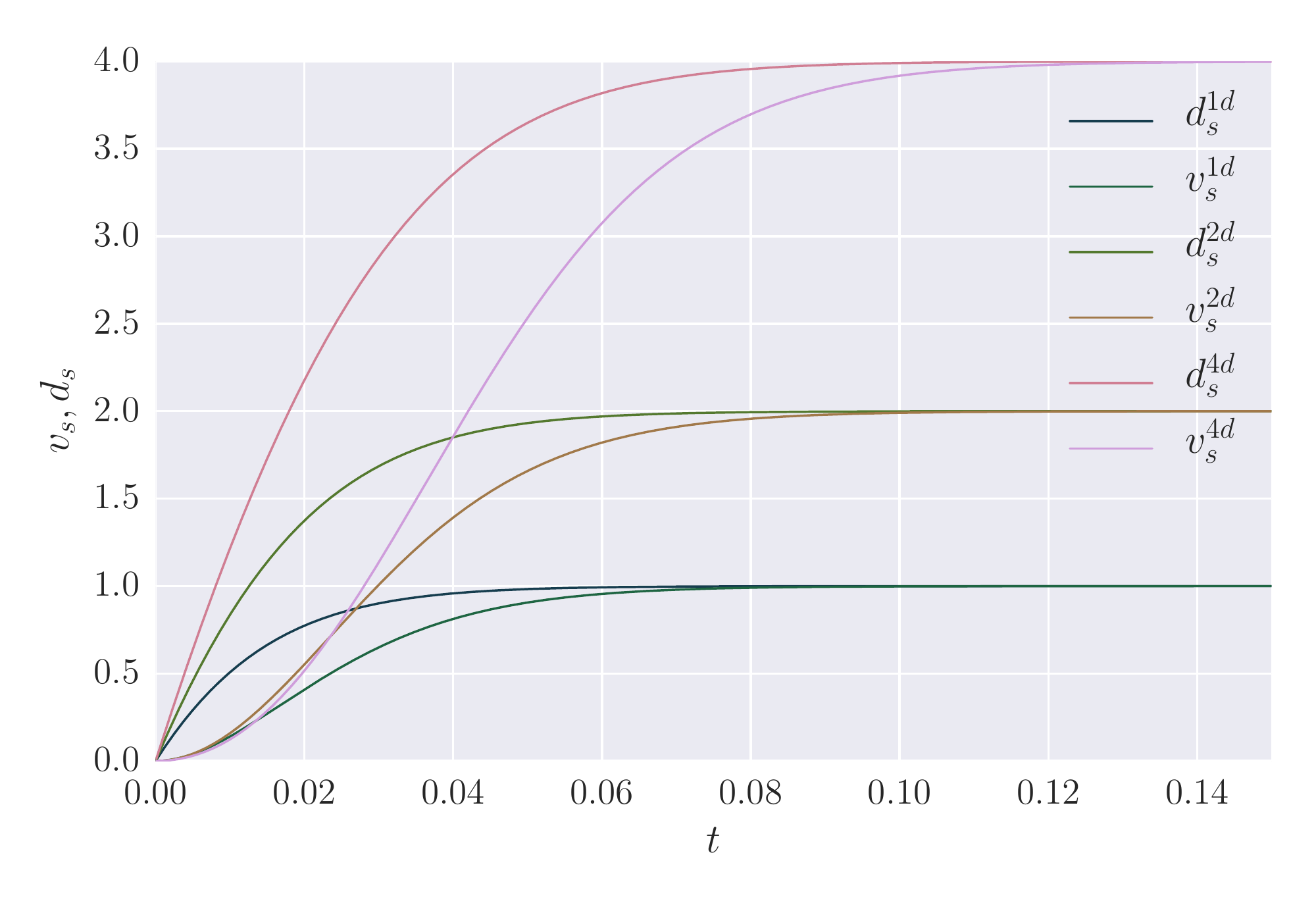}
\caption{\label{fig:flatDs}Spectral dimension and spectral variance for flat space, with a high energy cut-off $\Lambda=10$.}
\end{figure}

\subsection{Spectral variance}
One difference between the Laplace-Beltrami and Dirac operators is that the Laplace-Beltrami operator always has an eigenvector of eigenvalue zero, i.e.\ the constant function, whereas the Dirac operator often does not.
For example, if the scalar curvature is not zero for a compact Riemannian manifold, then zero is never an eigenvalue of $D$~\cite{MR1777332}.
If the magnitude of the lowest non-zero eigenvalue is $\lambda_0$ then it dominates in~\eqref{eq:specdim}  for large $t$ and the spectral dimension for the Dirac operator is asymptotically
\begin{equation}
d_s(t)\sim 2t\lambda_0^2,\label{eq:lowestev}
\end{equation}
whereas for the Laplace-Beltrami operator it always converges to zero at large $t$.
Hence the use of $d_s$ as a dimension estimator is not useful for values of $t$ larger than $t \sim \lambda_0^{-2}$.
A useful, energy dependent, dimension measure on a compact space is expected to go to $0$ for very low energies.
This would correspond to wavelengths larger than the scale of compactness, which would not propagate on the space at all.

Hence the linear growth phenomenon is an unhelpful feature of the spectral dimension.
A more useful refinement is a new measure of dimension called the \emph{spectral variance}.
\begin{align}
v_s(t) = d_s(t) - t\der{d_s}{t}\;,
\end{align}
which removes this linear mode.
The reason this is called the spectral variance is because it can be expressed in terms of the eigenvalues as
\begin{align}
v_s(t) &= 2 t^2 \left(
\frac{ \sum_{\lambda}\lambda^4 e^{-\lambda^2 t} }{ \sum_{\lambda} e^{-\lambda^2 t} } -
{ \left( \frac{ \sum_{\lambda} \lambda^2 e^{- \lambda^2 t} }{ \sum_{\lambda} e^{-\lambda^2 t} } \right) }^2 \; \right)
\end{align}
This formula is $2t^2$ times the variance of $\lambda^2$ in a probability distribution $p(\lambda)=K^{-1} e^{- \lambda^2 t}$.

For large $t$ the spectral variance is dominated by the lowest two eigenvalues $\lambda_0$ and $\lambda_1$. Denoting their multiplicities $\mu_0, \mu_1$, it is asymptotically
\begin{align}\label{eq:vsbump}
v_s(t)
\sim 2 \frac{\mu_1}{\mu_0}t^2 { \left(\lambda_1^2 - \lambda_0^2\right) }^2 e^{-(\lambda_1^2-\lambda_0^2)t}
\end{align}
\noindent The maximum value of this approximate formula is $8e^{-2}\mu_1/\mu_0$ at the point $t=2/(\lambda_1^2 - \lambda_0^2)$, which can be large if $|\lambda_1|-|\lambda_0|$ is small.


For flat space with a high-energy cut-off $\Lambda$, the spectral variance is
\begin{align}
v_s(t)&=d+\frac{e^{-\Lambda^2 t} \left(d-2 \Lambda^2 t-2\right) { \left(\Lambda^2 t\right) }^{d/2}}{\Gamma \left(\frac{d}{2}\right)-\Gamma \left(\frac{d}{2},\Lambda^2 t\right)}-\frac{2 e^{-2 \Lambda^2 t} { \left(\Lambda^2 t\right) }^d}{{\left(\Gamma  { \left( \frac{d}{2} \right)} -\Gamma \left( \frac{d}{2},\Lambda^2 t\right)\right) }^2}\;
\end{align}
as can be seen in Figure~\ref{fig:flatDs}.
There both the spectral variance and spectral dimension converge towards the continuum dimension for large $t$, since the space is non-compact.

The spectral dimension and variance of the fuzzy sphere do not have simple closed expressions but numerical results are readily obtained. Figure~\ref{fig:sphereDV} shows plots of the spectral dimension and variance for fuzzy spheres of two sizes. For large values of $t$ the three plots look similar due to the spectra of the different fuzzy spheres being identical for the smallest eigenvalues.
For the continuum sphere (Figure~\ref{fig:sphereDVc}), for very low $t$, $d_s(t)$ and $v_s(t)$ both go to $2$, while for the fuzzy spheres in Figures~\ref{fig:sphereDVa},~\ref{fig:sphereDVb}, they go to $0$.
This is expected, since low $t$ corresponds to the high energy limit, and the high energy modes are removed from the fuzzy space.
In the large $t$ limit, $d_s(t)\sim t$, as explained by~\eqref{eq:lowestev}.
The spectral variance, on the other hand, satisfies $v_s(t) \to 0$ for $t\to \infty$, which is the desired behaviour for a compact space.

\begin{figure}[t]
\subfloat[][$N=5$\label{fig:sphereDVa}]{\includegraphics[width=0.45\textwidth]{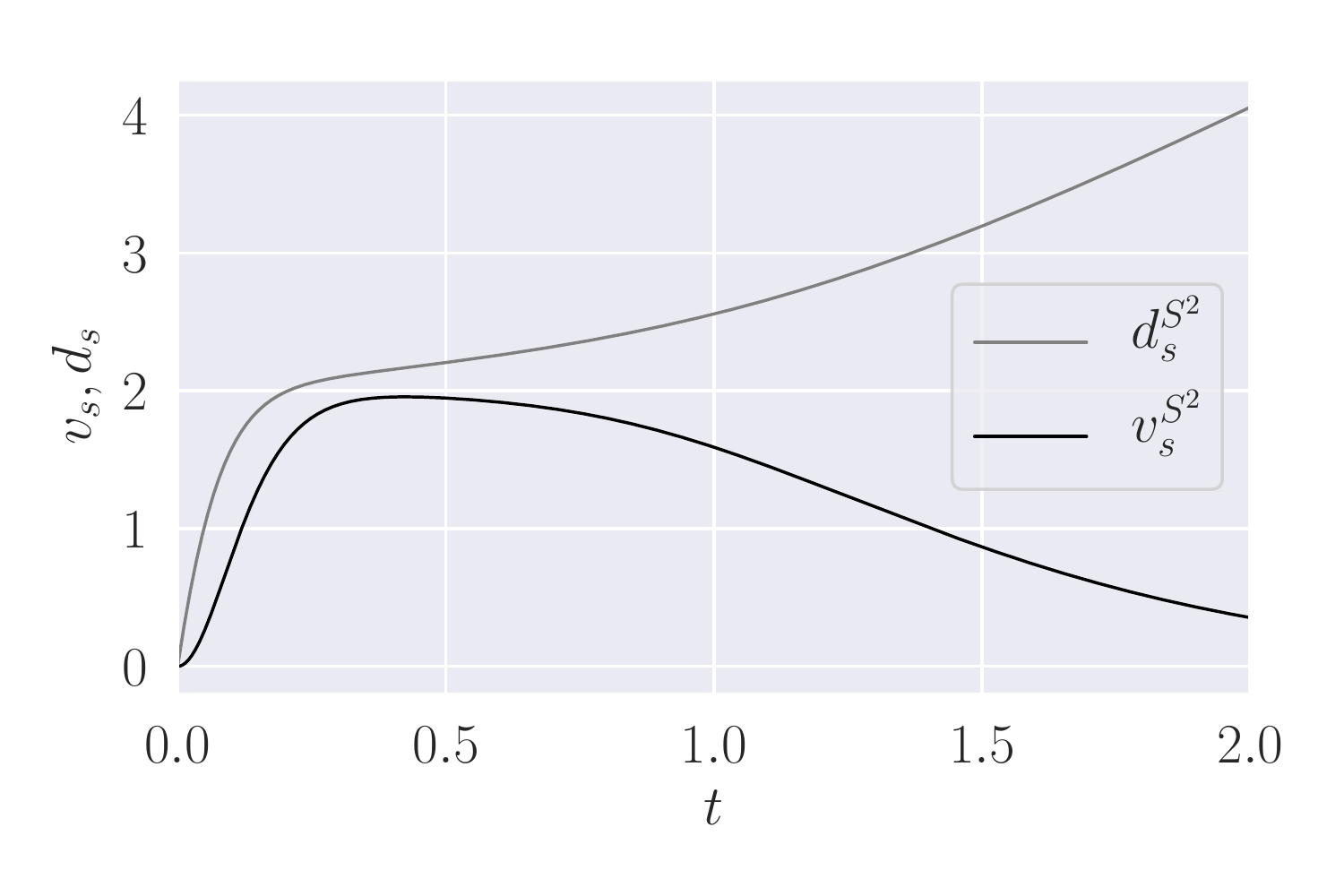}}
\subfloat[][$N=15$\label{fig:sphereDVb}]{\includegraphics[width=0.45\textwidth]{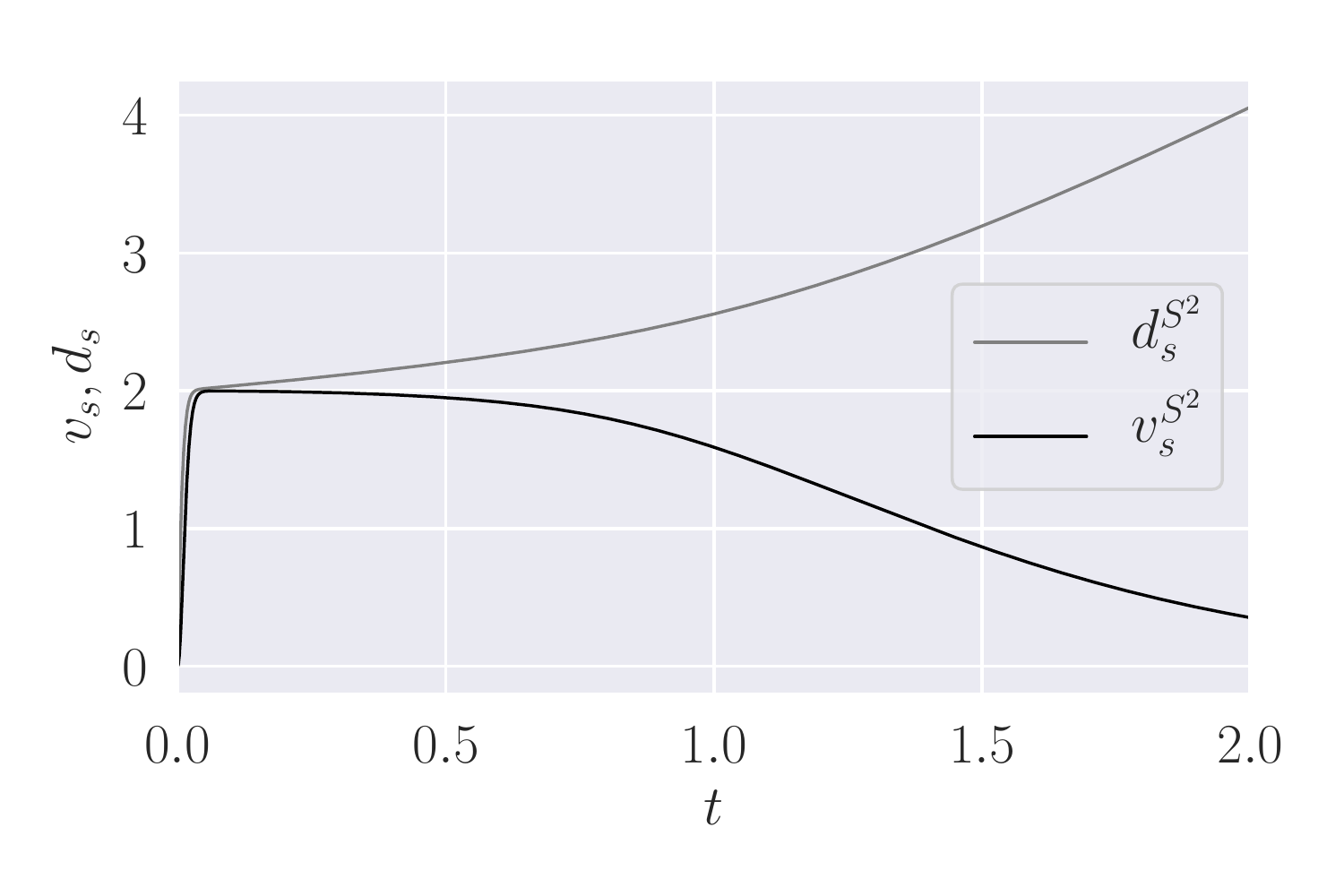}}

\centering
\subfloat[][Continuum sphere\label{fig:sphereDVc}]{\includegraphics[width=0.45\textwidth]{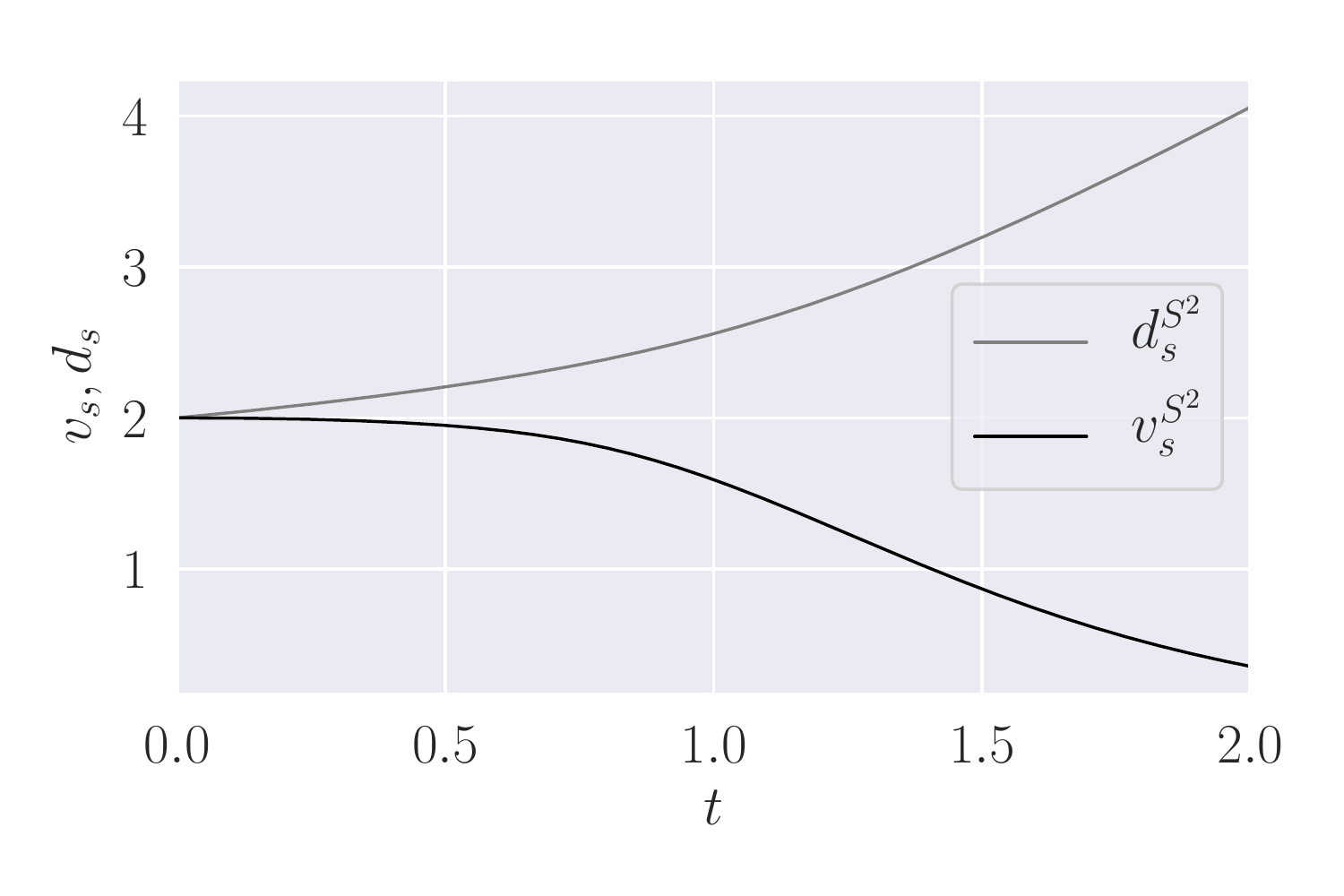}}
	\caption{\label{fig:sphereDV}Spectral dimension and spectral variance for the fuzzy sphere.}
\end{figure}

For the torus, the first thing to examine is how well the exact torus and the fuzzy torus agree.
This is tested for $a=d=1$ and $b=c=0$, which leads to a square torus with sides of length $2\pi$.
Figure~\ref{fig:torus} shows the spectral variance and spectral dimension for fuzzy tori of matrix size $N=10$, $N=30$ and the continuum torus.
The figure shows larger $N$ values than for the sphere since the spectrum of the fuzzy torus is a worse approximation to the spectrum of its continuum analogue.

\begin{figure}[t]
\subfloat[][$N=10$\label{fig:TorusDVa}]{\includegraphics[width=0.45\textwidth]{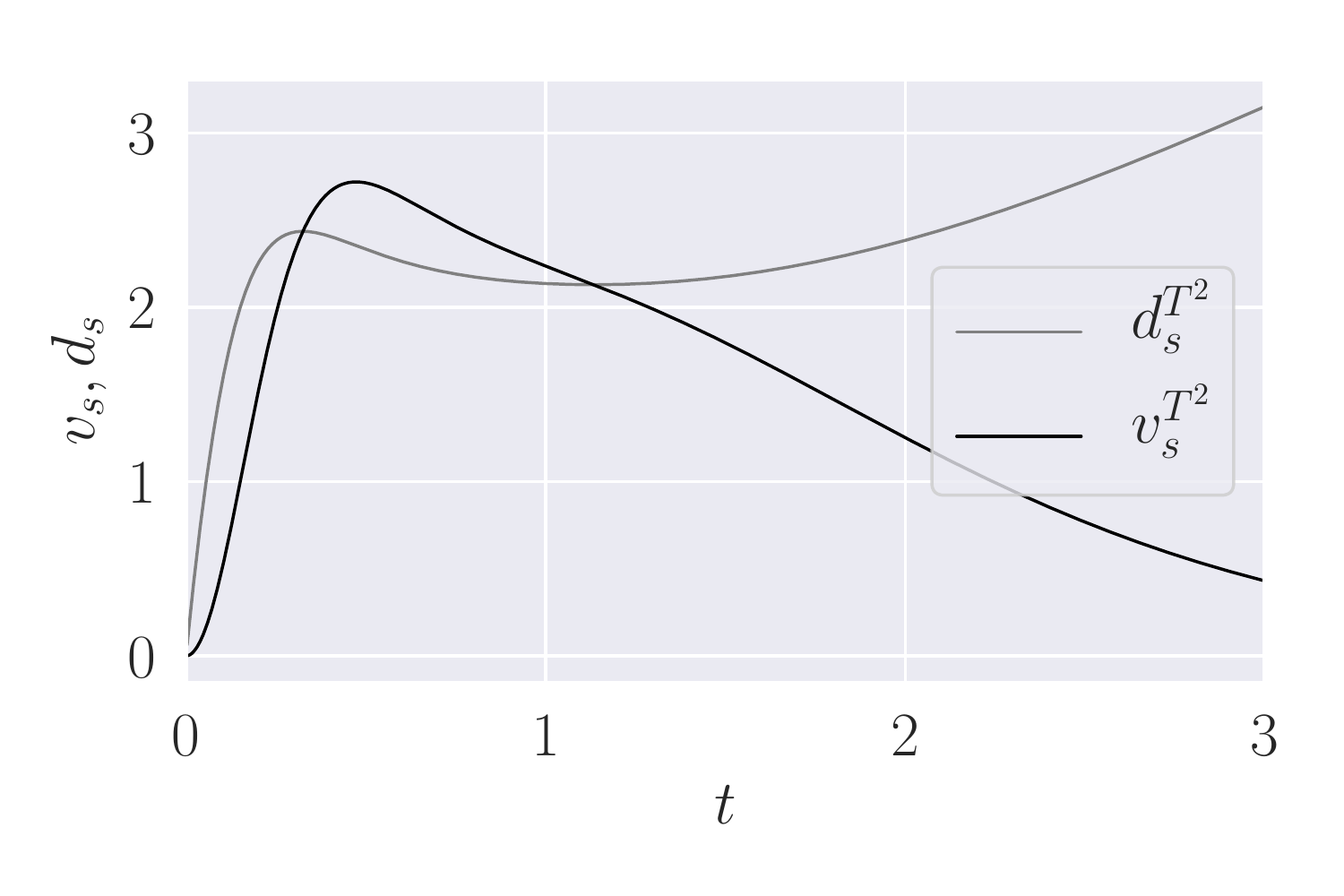}}
\subfloat[][$N=30$\label{fig:TorusDVb}]{\includegraphics[width=0.45\textwidth]{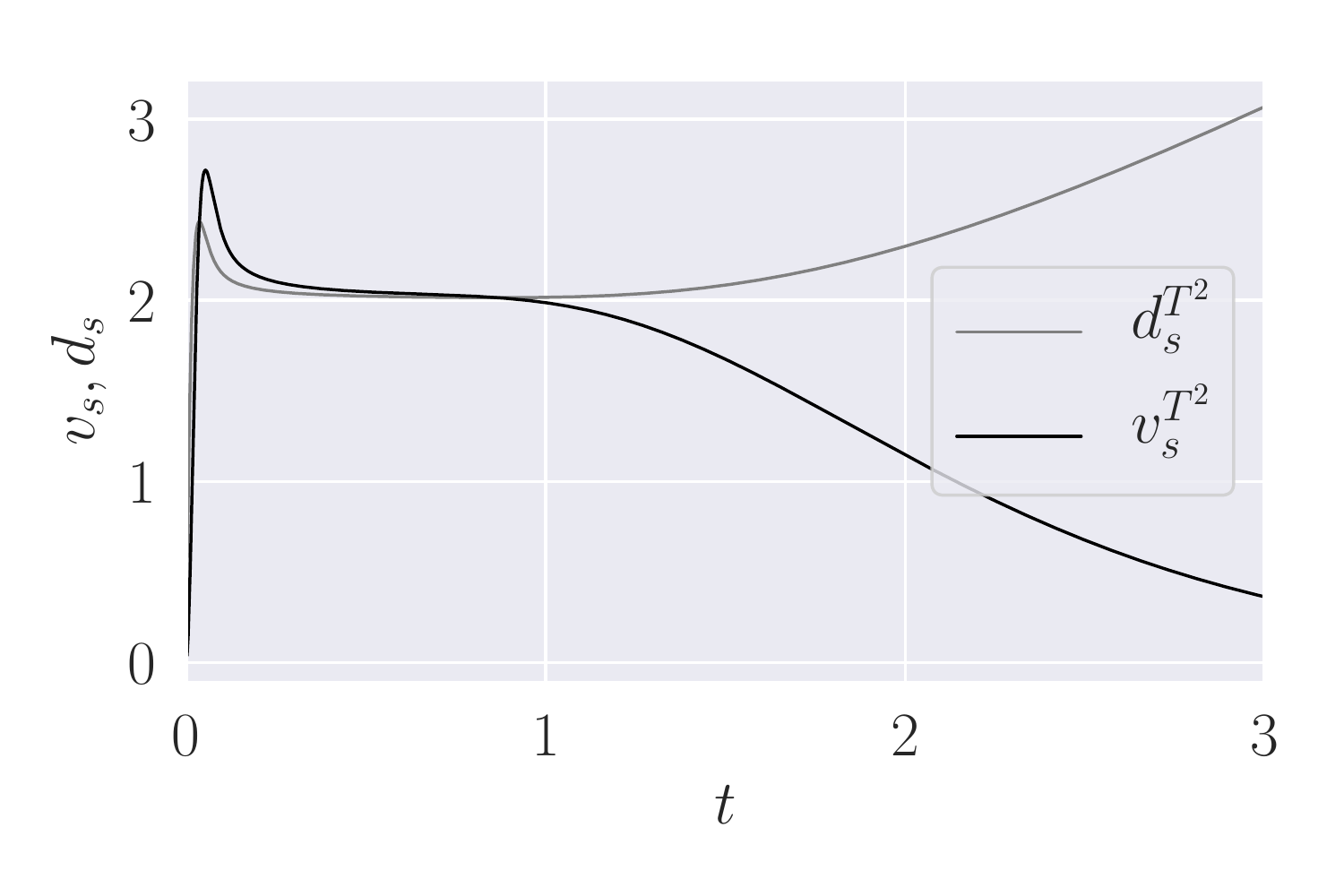}}

\centering
\subfloat[][Continuum torus \label{fig:torusvs}]{\includegraphics[width=0.45\textwidth]{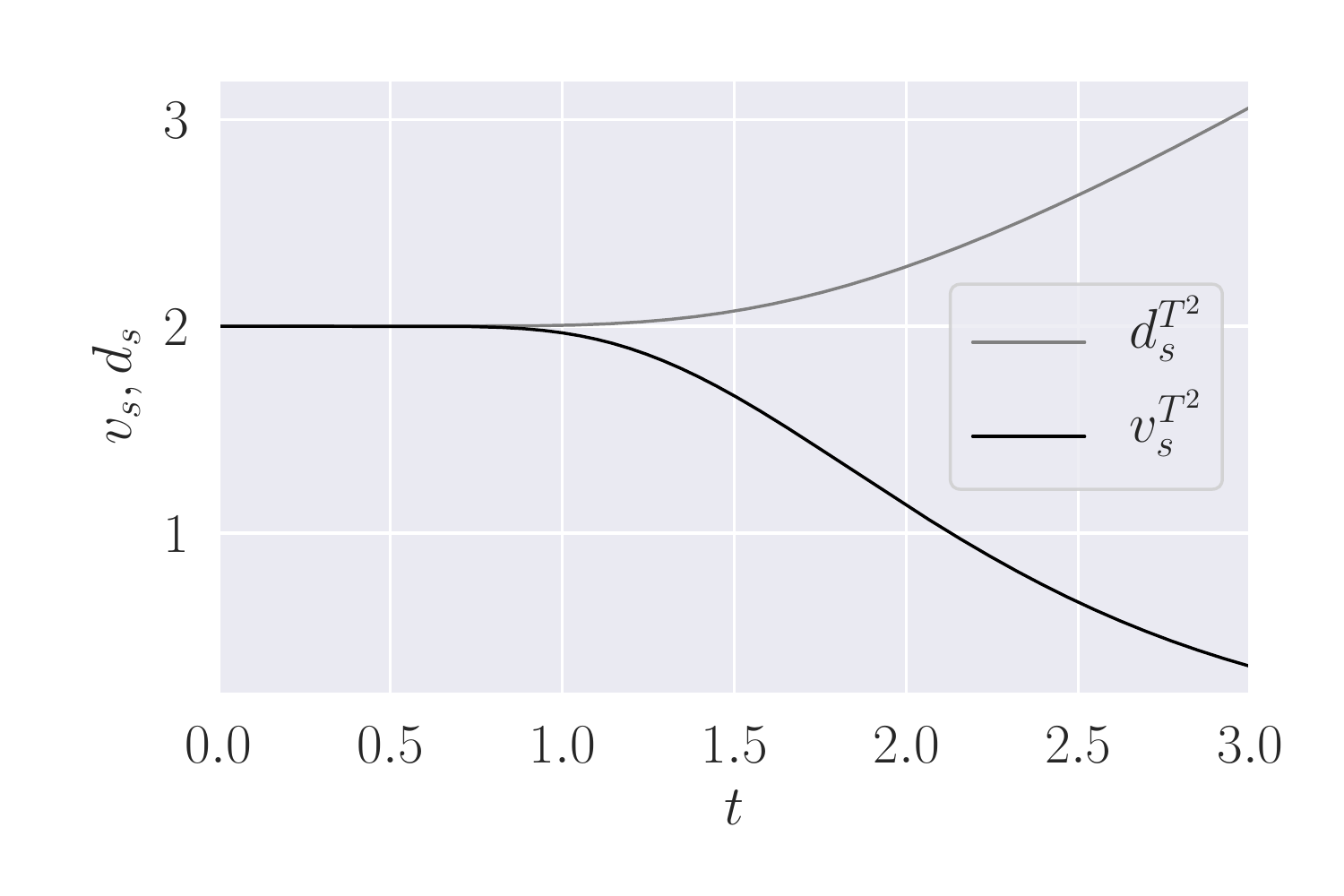}}
	\caption{\label{fig:torus}Spectral dimension and spectral variance for the square fuzzy torus with $a=d=1$ and $b=c=0$}
\end{figure}

For $N\sim 30$ the fuzzy torus and the continuum torus are in good agreement on their spectral variance and dimension.
One exception is the sharp peak in the spectral dimension and variance of the fuzzy torus at small $t$, i.e., high energies. With increasing $N$ this peak becomes sharper but does not grow in height.
It is created by the large $k,l$ eigenvalues, since these eigenvalues are quite far from the continuum value for the same $(k,l)$ pair. They form a slowly rising plateau in the region where the sine function reaches its maximum. This is shown in Figure~\ref{fig:TorFuzzyComm}, which plots the eigenvalues as a function of $k,l$ for both the continuum and the fuzzy torus for $a=d=1$, $b=c=0$.
A very similar phenomenon also occurs for Laplace operators on discretisations of the torus~\cite{Calcagni:2013dna}, so it does not appear to be related to the non-commutativity.

\begin{figure}
  \centering
\includegraphics[width=0.5\textwidth]{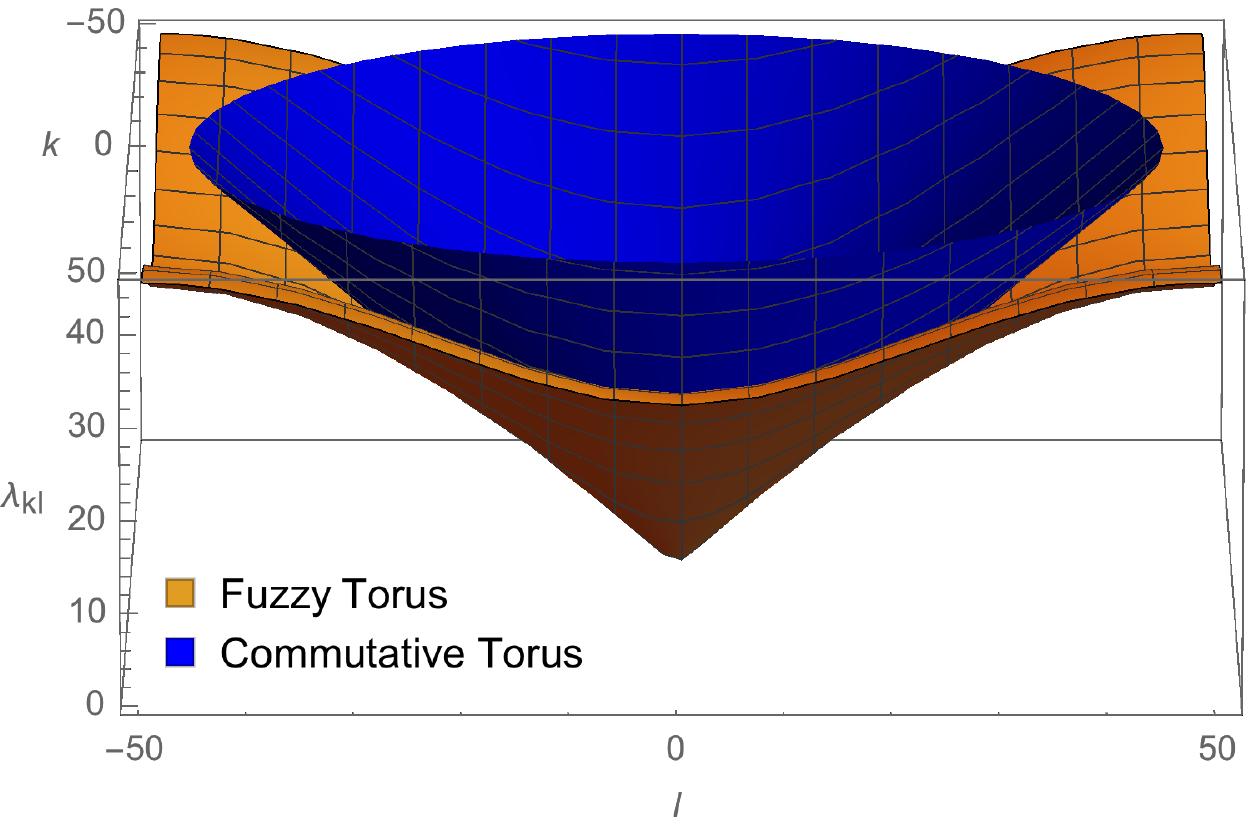}
    \caption{The positive eigenvalues of the Dirac operator for the fuzzy torus for $N=100$ compared to those of the continuum torus, both plotted as functions of $k,l$. This shows that the fuzzy torus deviates significantly from the continuum torus for the majority of its spectrum.}
    \label{fig:TorFuzzyComm}
\end{figure}
One can see that for the fuzzy torus the eigenvalues start to deviate from the continuum torus outside a small region around the origin, and that the majority of eigenvalues are far away from their continuum values (see Appendix~\ref{sec:appendix}).
Overall the eigenvalues show a slower than  linear rise, which leads to the fuzzy torus showing higher-dimensional behaviour at large energies.
The spectral variance shows this as a peak at low $t$ before it falls to $0$ at very low $t$.

\begin{figure}
\subfloat[$N=90$]{\includegraphics[width=0.5\textwidth]{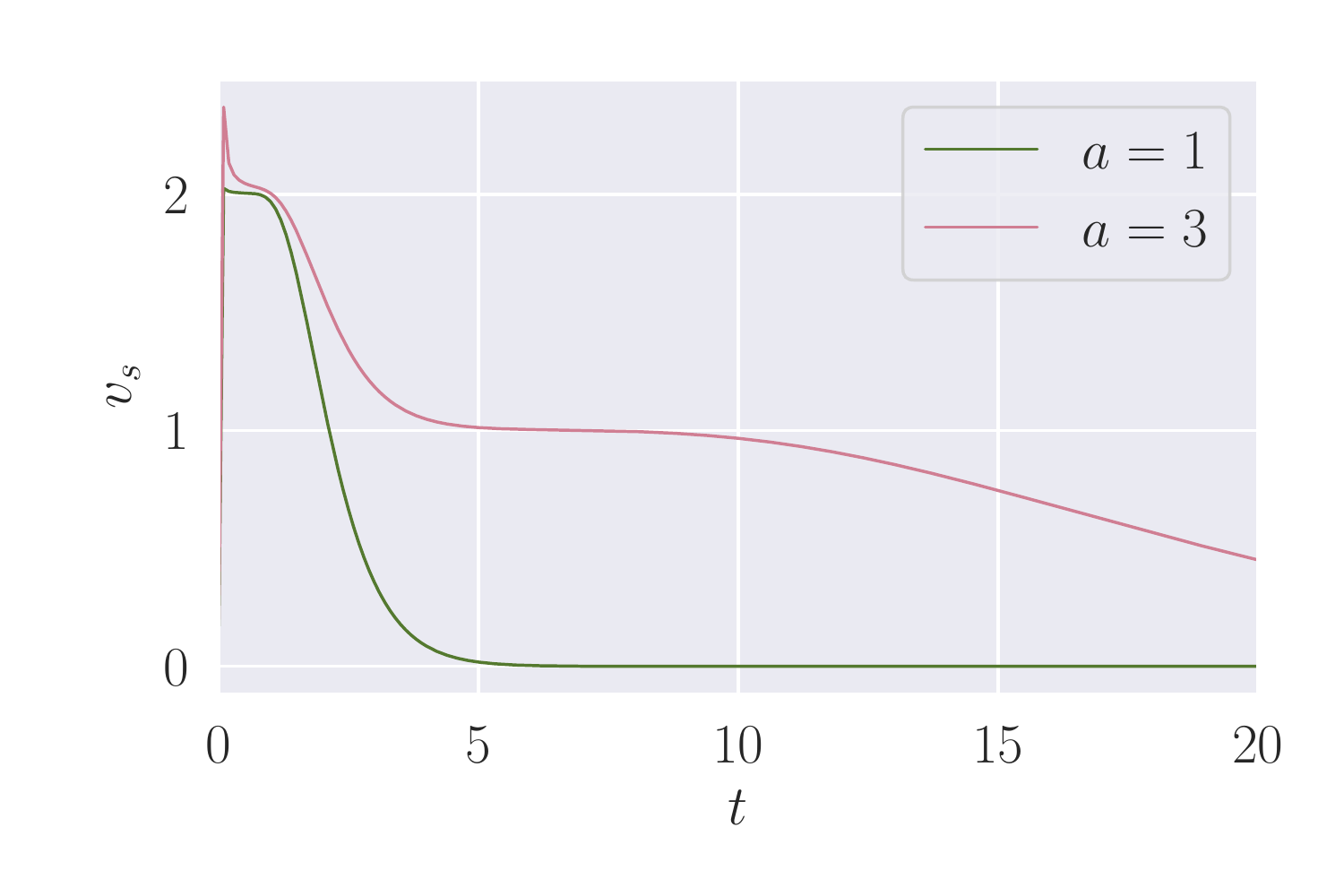}}
\subfloat[continuum torus]{\includegraphics[width=0.5\textwidth]{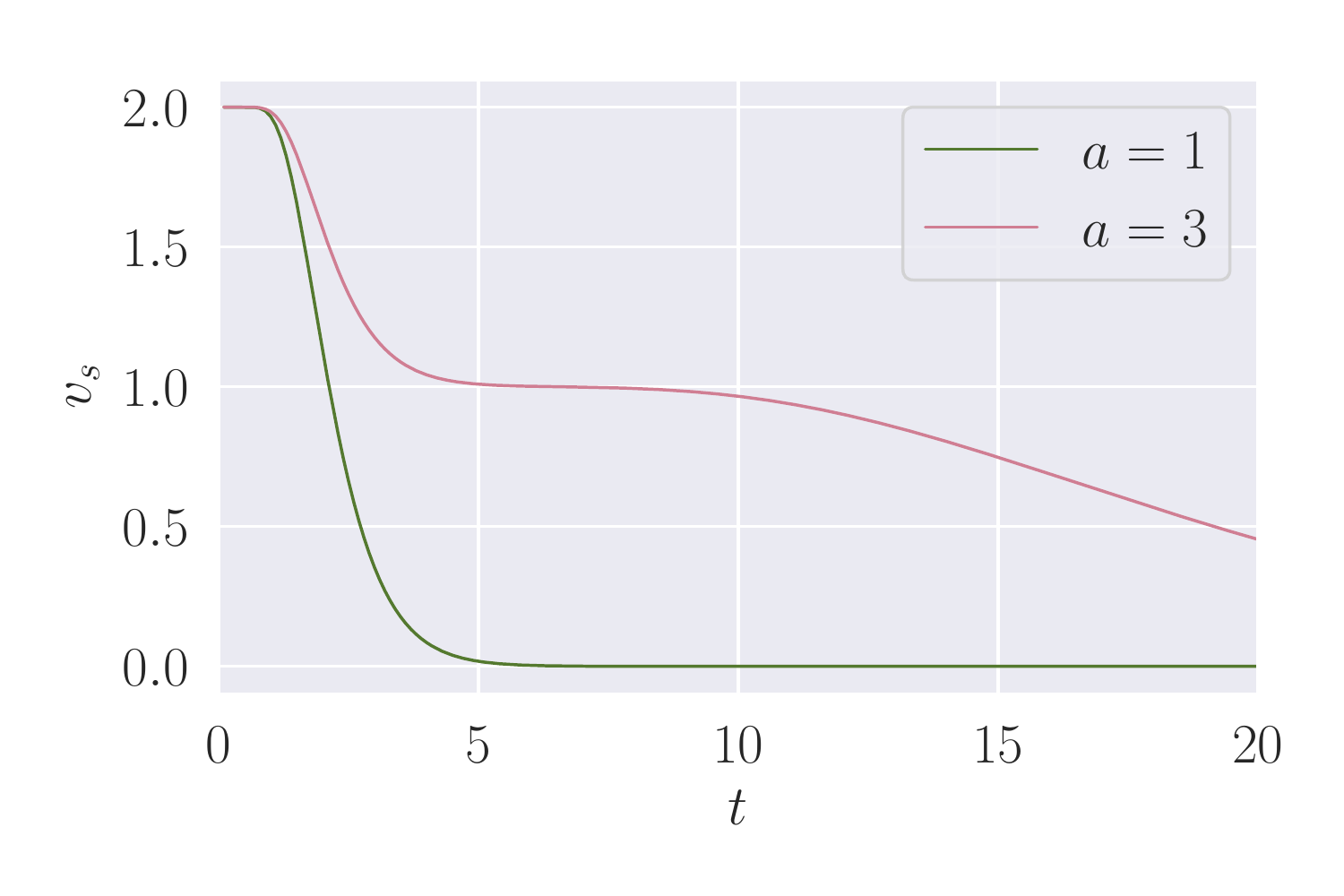}}
    \caption{Spectral variance of a square torus $a=d=1$ and $b=c=0$ in green and a rectangular torus $a=3,d=1$ and $b=c=0$ in pink.}
    \label{fig:torusCvs}
\end{figure}

An interesting feature of energy-dependent dimension measures on a torus, which has for example been seen in~\cite{Calcagni:2013dna}, is that a torus with very different length scales in the different directions will show both 1- and 2-dimensional behaviour depending on the energy scale probed.
This is a Kaluza-Klein type effect, in which only modes with high enough energy will see the true 2-dimensional structure of the torus, while lower energy modes remain constrained to one direction.
This is shown for a rectangular torus with $a=3$, $b=c=0$ and $d=1$, which has edge lengths $2\pi$ and $18\pi$. This has the same spin structure as for $a=1$, which makes it easier to compare the high $t$ behaviour.
Figure~\ref{fig:torusCvs} shows the spectral variance for $N=90$ and the continuum torus.
The first thing to notice is that for $a=3$ the high eigenvalue peak at low $t$ is wider and more pronounced than for $a=1$.
After this peak there is a region of values of $t$ for which the dimension is 2.
The interesting behaviour is after this $2$d region. While the spectral variance for $a=1$ falls off towards $0$, the spectral variance for $a=3$ shows a clear $1$d region. The fuzzy torus and the continuum torus in this case agree well, except for the behaviour at small $t$.

The proposal in this paper is to use the spectral variance of a Dirac operator as a measure of dimension. This requires some justification as it is not always clear what the dimension of a non-commutative space means. The minimum requirement is that if a spectrum approximates the spectrum of a manifold, then the dimensions should agree (to an appropriate level of approximation). For a manifold, both the spectral dimension and spectral variance give exactly the dimension at $t=0$, which is the high energy limit. When the spectrum is truncated to a finite number of eigenvalues, this is no longer true, and one has to look at the spectral dimension or spectral variance at some non-zero value of $t$. From the examples presented so far, it is apparent that this is where the graph plateaus, i.e., has an interval of $t$ for which the graph is approximately constant. These plateaus mimic the behaviour of the graphs for flat space and show there is a region of energies for which the growth of the eigenvalues is similar to the behaviour of flat space with the same dimension.

The spectral dimension is problematic when applied to the Dirac operator because of the linearly rising mode which, in some cases, dominates for sufficiently low values of $t$, obscuring any other features. The spectral variance does not have this problem, though the graph can show significant features at large $t$ when several low eigenvalues have an unusual configuration (see section~\ref{subsec:EVscale}).

Thus the focus is on the spectral variance in regions of $t$ for which the graph is approximately stationary. The fuzzy sphere has just one stationary point and this is therefore the maximum of the graph. The value of the spectral variance at this maximum agrees well with the continuum dimension, $2$.
For the torus the graph shows surprisingly good agreement with the continuum spectral variance, and the value $2$ for a wide region, despite the fact that the actual spectra are substantially different.
The only marked difference is that the fuzzy torus has a higher peak at very small $t$.
This peak arises because the large eigenvalues of the fuzzy torus rise much slower than those of the continuum torus.
This example makes it clear that it is not the global maximum that is important, rather it is the larger region of $t$ for which the graph is approximately constant that matters.

Another way of understanding the spectral variance is to realise that the probability distribution $p(\lambda)$ is that of a thermal equilibrium of a system with Hamiltonian $D^2$ and inverse temperature $T^{-1}=t$.
Using this distribution, one can write
\begin{align}
d_s&=2 t\, \langle \lambda^2\rangle \\
v_s&=2 t^2\, \left( \langle \lambda^4 \rangle - \langle \lambda^2 \rangle^2 \right) \;.
\end{align}
In the thermodynamic analogy, $K$ is the partition function, the spectral dimension is $2U/T$, with $U=\langle \lambda^2\rangle$, the internal energy, and the spectral variance is twice the heat capacity
$2\der U T=-2t^2\der Ut.$
It is clear that adding a constant to the internal energy does not affect the heat capacity.
When interpreted in terms of the eigenvalues, adding a constant to the squared eigenvalues does not affect the spectral variance, however it will affect the spectral dimension.

For flat space, this thermodynamic system is equivalent to a molecule of an ideal gas in the space $\R^d$. The fact that the heat capacity of the ideal gas is $d/2$ is a standard result (it is normally quoted for $d=3$). Thus one can think of the spectral variance at parameter $t$ as a determination of the dimension of a space from the heat capacity of an ideal gas that is  in it at the temperature $T=1/t$.
This makes an interesting connection with other work which proposes to use the thermal behaviour as a dimension measure~\cite{Amelino-Camelia_Brighenti_Gubitosi_Santos_2017}.


\subsection{Dimension of random fuzzy spaces}\label{sec:ranspec}
The spectral dimension and spectral variance can be calculated for fuzzy spaces whose continuum geometry is unknown. In this section, this is done for the class of random non-commutative geometries defined in~\cite{barrett_monte_2015}.
Before examining the spectral properties of ensembles of random geometries one needs to answer two questions. The first is how to average over the ensemble to extract meaningful geometric quantities. The second question is how to compare geometries so that differences in scaling do not dominate the comparison. For example, one may want to know if a geometry is close to a sphere of any radius, instead of a sphere of fixed radius. These issues are discussed before presenting specific results.

\subsubsection{Averaging dimension measures}
To examine the spectral geometry of the random fuzzy spaces there are three possibilities:
\begin{itemize}
\item Calculate the average spectrum, by calculating the averages of the eigenvalues, sorted by magnitude, and then determine the spectral dimension and variance of this.
\item Combine the spectra of all the geometries in the ensemble into one large spectrum
 and then calculate the spectral dimension and variance of this.
\item Calculate the spectral variance or dimension for each geometry in our ensemble of random fuzzy spaces and then average over these.
\end{itemize}
The first option is computationally efficient, however the exact physical meaning of the average of a sorted eigenvalue is unclear. It also ignores that the fluctuations in the eigenvalues might be correlated with each other which could have an important effect on the spectral estimators.

The second option has a certain elegance, since it uses all the data and it is customary to think of eigenvalue densities of a random ensemble as approximating a density of states.
However using the density as an approximation removes an important part of the spectral information, namely the spacing of the eigenvalues, from the system.
The artificial density then has little to do with the true compact space, since in the limit of infinitely many measurements, it leads to a continuous spectrum, which is characteristic of a non-compact space.
A related problem is that the lowest eigenvalues typically have a large variance, leading to unstable large $t$ behaviour of the spectral dimension.
Figure~\ref{fig:evHist} shows the distribution of individual eigenvalues of the Dirac operator for all three types at their respective phase transitions.
In particular for types $(1,3)$ and $(2,0)$ the distribution of the lowest eigenvalues is very wide.

\begin{figure}
  \subfloat[][Type $(1,3)$ $g_2=-3.70$  $N=8$ ]{\includegraphics[width=0.33\textwidth]{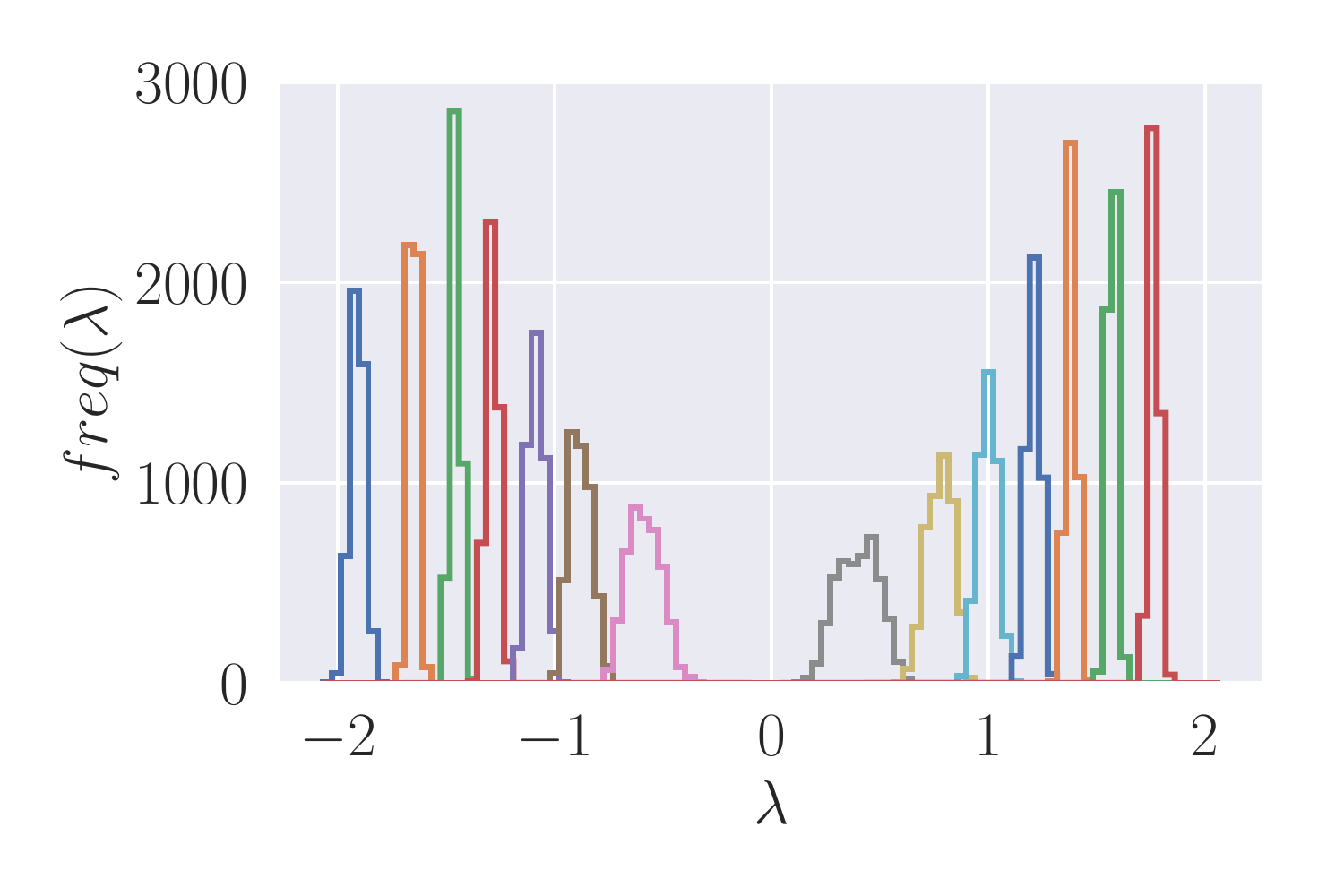}}
  \subfloat[][Type $(1,1)$ $g_2=-2.40$ $N=10$ ]{\includegraphics[width=0.33\textwidth]{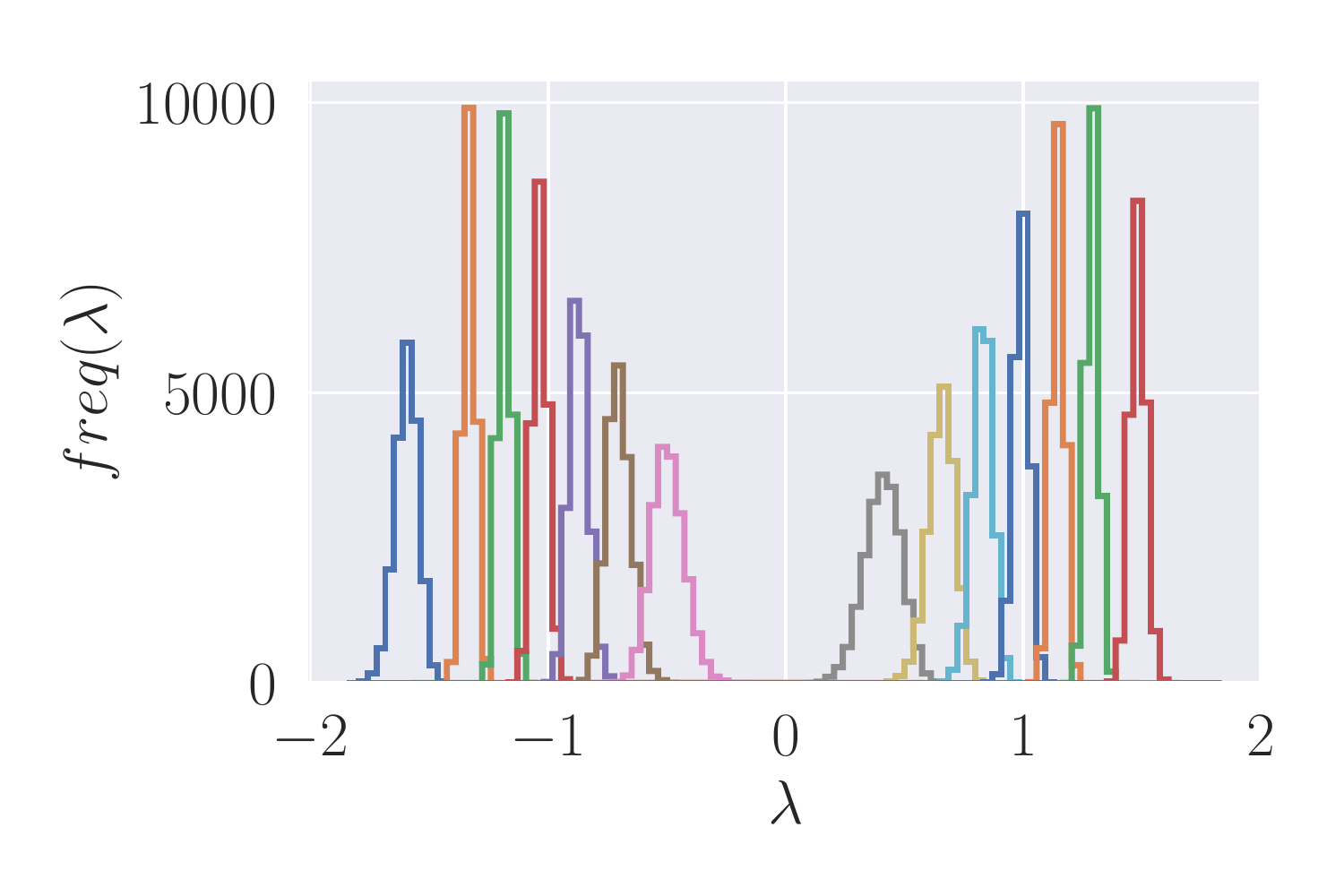}}
  \subfloat[][Type $(2,0)$ $g_2=-2.80$ $N=10$ ]{\includegraphics[width=0.33\textwidth]{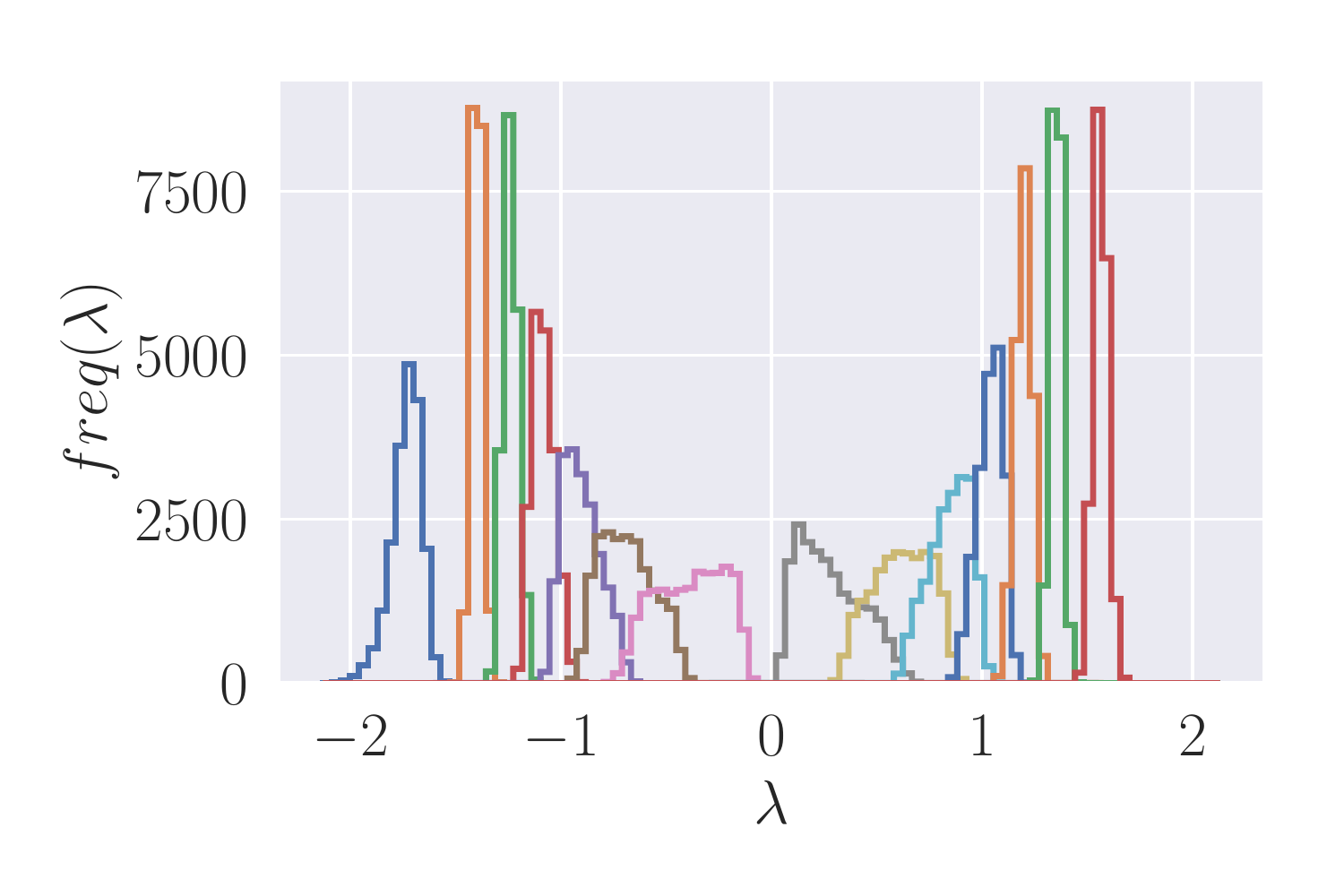}}
  \caption{The distribution of individual eigenvalues at the phase transition.
To make it possible to distinguish individual eigenvalues, only a subset of 13  equidistant eigenvalues $\lambda_0,\lambda_{N/13},\lambda_{2N/13},\dots$ is plotted.
Particularly for types $(1,3)$ and $(2,0)$, the smallest eigenvalues have a very wide distribution.}\label{fig:evHist}
\end{figure}

Hence this article uses the third option, which is comparable to the definition used in other studies of quantum space-time.
Since the spectral variance and dimension are, to some degree, physical, it is reasonable to average them over different geometries in our ensemble and expect the result to make sense.

It is interesting that away from the phase transition the spectral variance from the average eigenvalues and the average of the spectral variances are very similar, while at the phase transition these two definitions lead to quite different values.
This is illustrated in the left-hand plots in Figure~\ref{fig:vscomp}.
The right-hand plots of Figure~\ref{fig:vscomp} show the spectral variance for a sample of geometries taken from the respective ensembles.
This shows that the geometries that contribute to one ensemble can be very different, particularly close to the phase transition.
From the graphs one can see that the variance of the spectral variance at a given value of $t$ is large, i.e., of order 1, at $N=10$. A good question is whether this variance decreases for larger $N$. This will need further data to determine.

Looking at the plots in detail, one can see that several of the curves exhibit a maximum at a large value of $t$, in accordance with the discussion of~\eqref{eq:vsbump}. Apart from this feature, the curves in Figure~\ref{fig:vscomp}(a) are qualitatively similar to each other, as are the ones in Figure~\ref{fig:vscomp}(c). However at the phase transition, in Figure~\ref{fig:vscomp}(b), one sees curves similar to both (a) and (c), suggesting that the system spends some time on each side of the transition.

\begin{figure}
  \subfloat[][Type $(2,0)$ $g_2=-2.5$ $N=10$]{\includegraphics[width=0.42\textwidth]{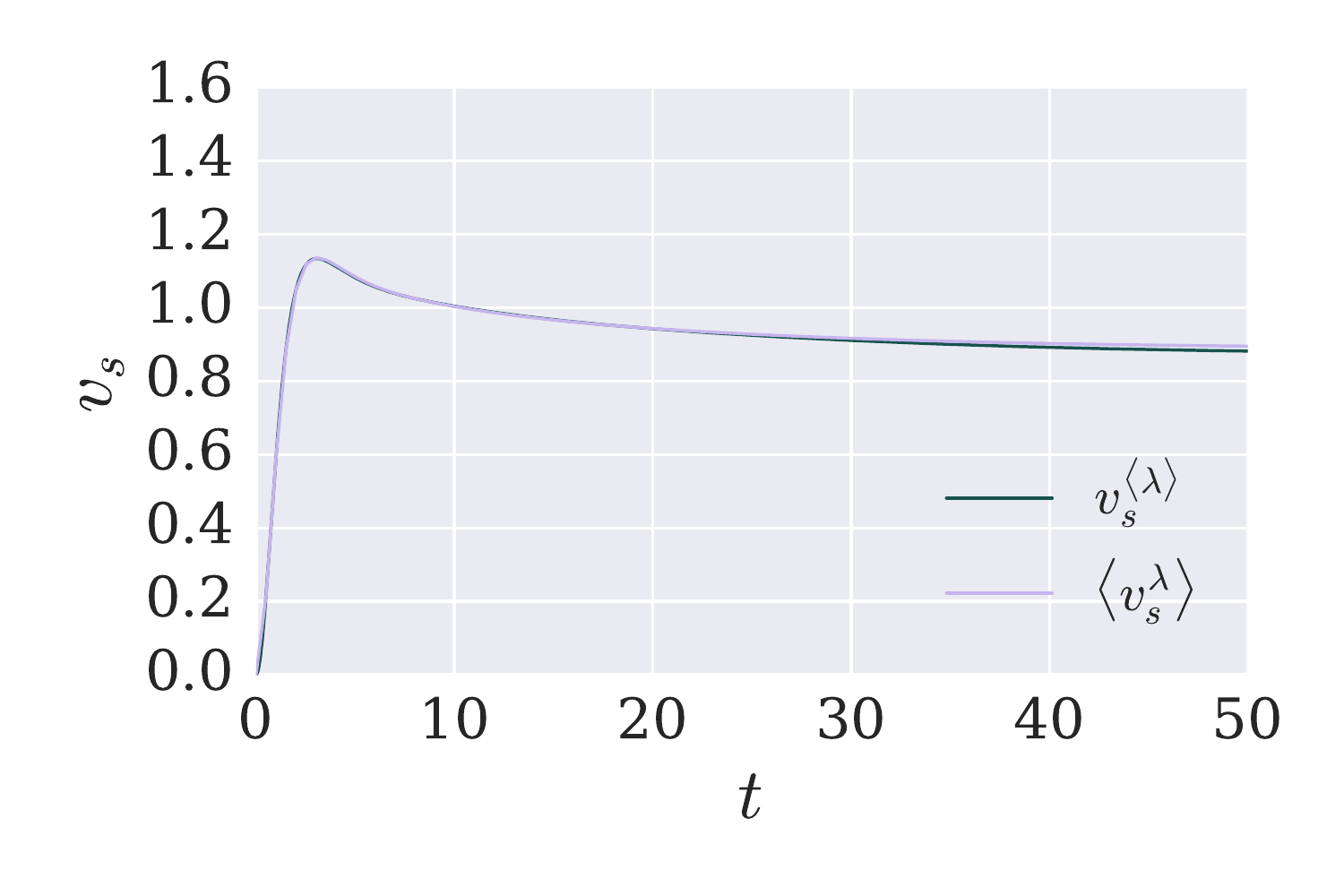}
  \includegraphics[width=0.42\textwidth]{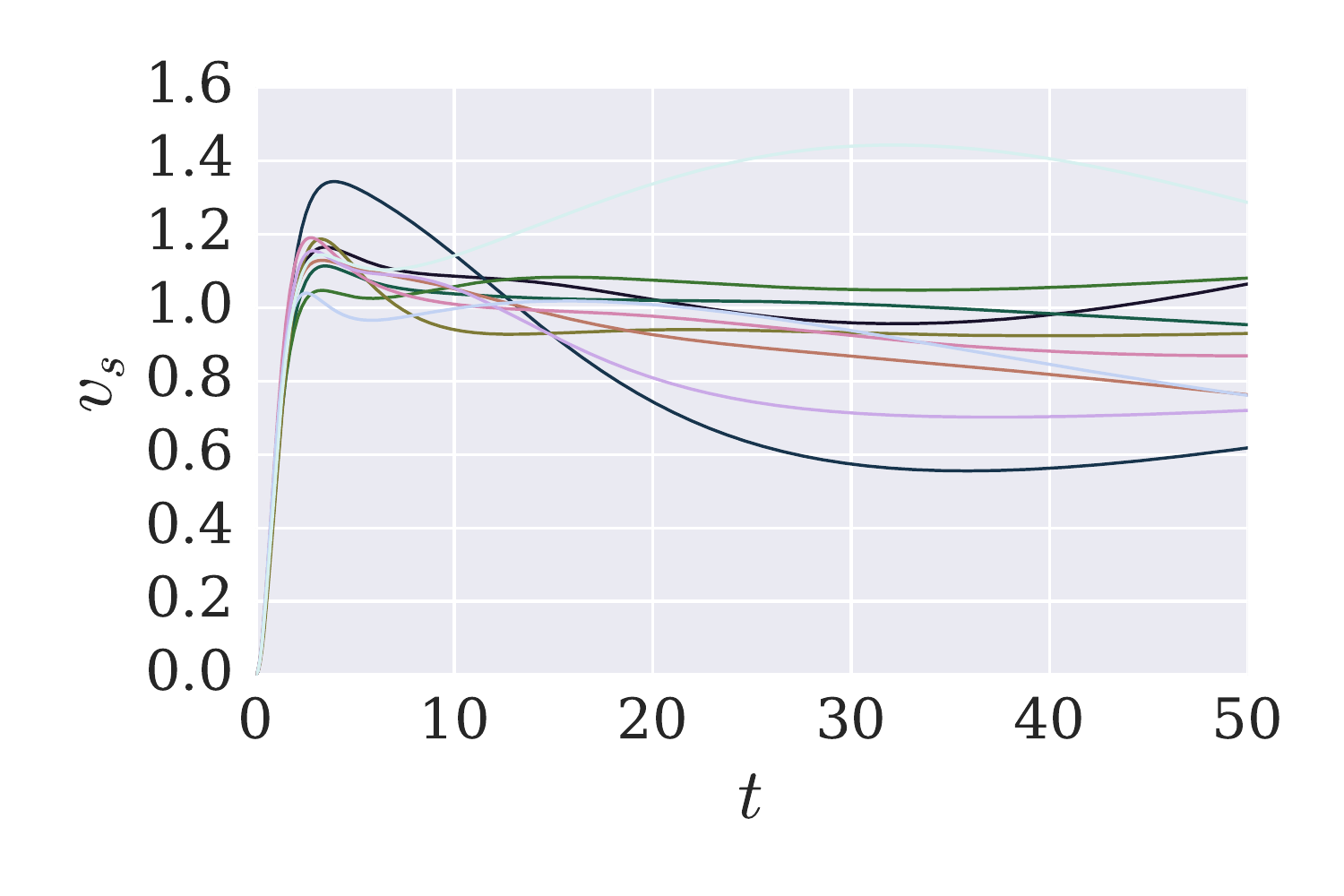}
  }\\
  \subfloat[][Type $(2,0)$ $g_2=-2.8$ $N=10$]{\includegraphics[width=0.42\textwidth]{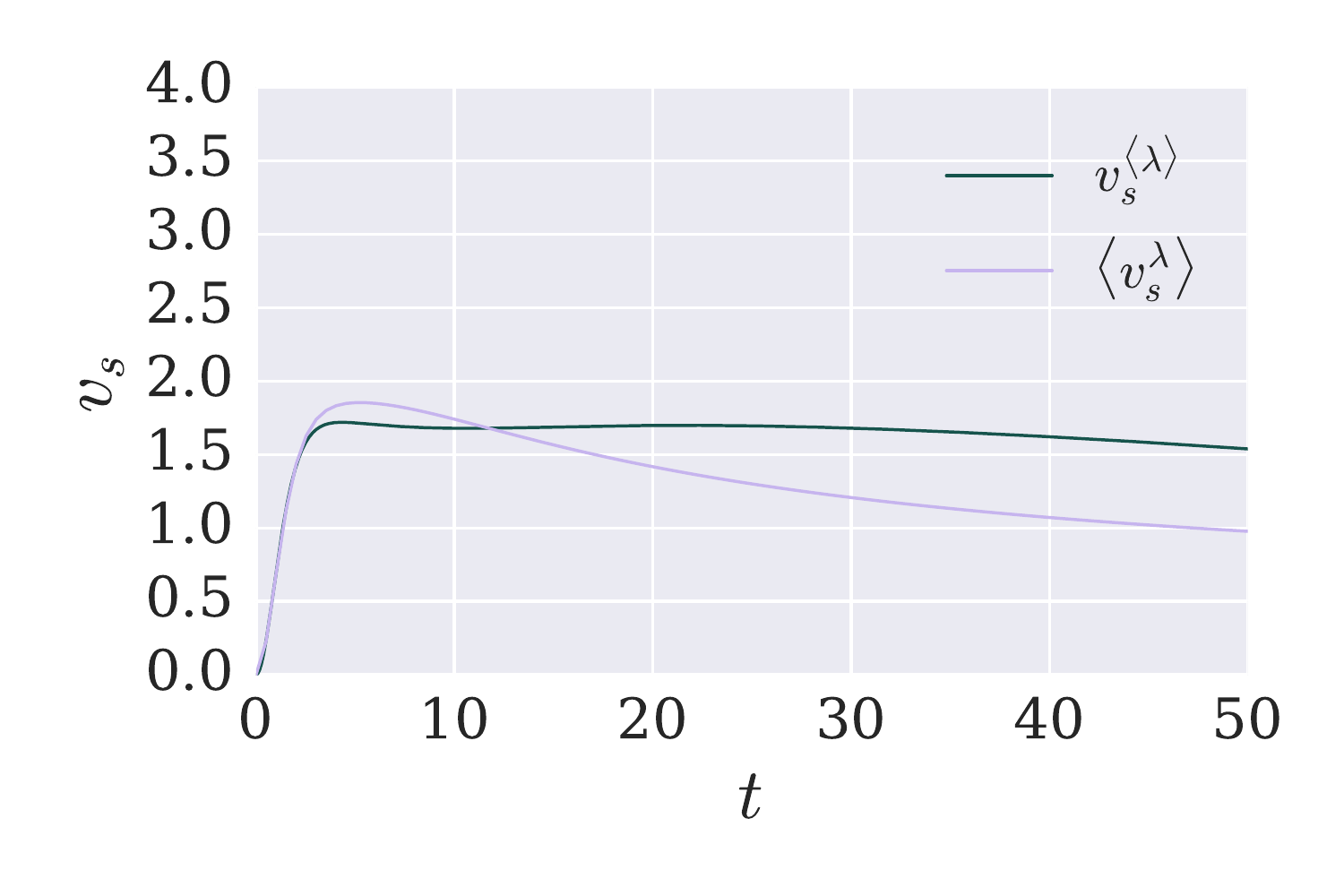}
  \includegraphics[width=0.42\textwidth]{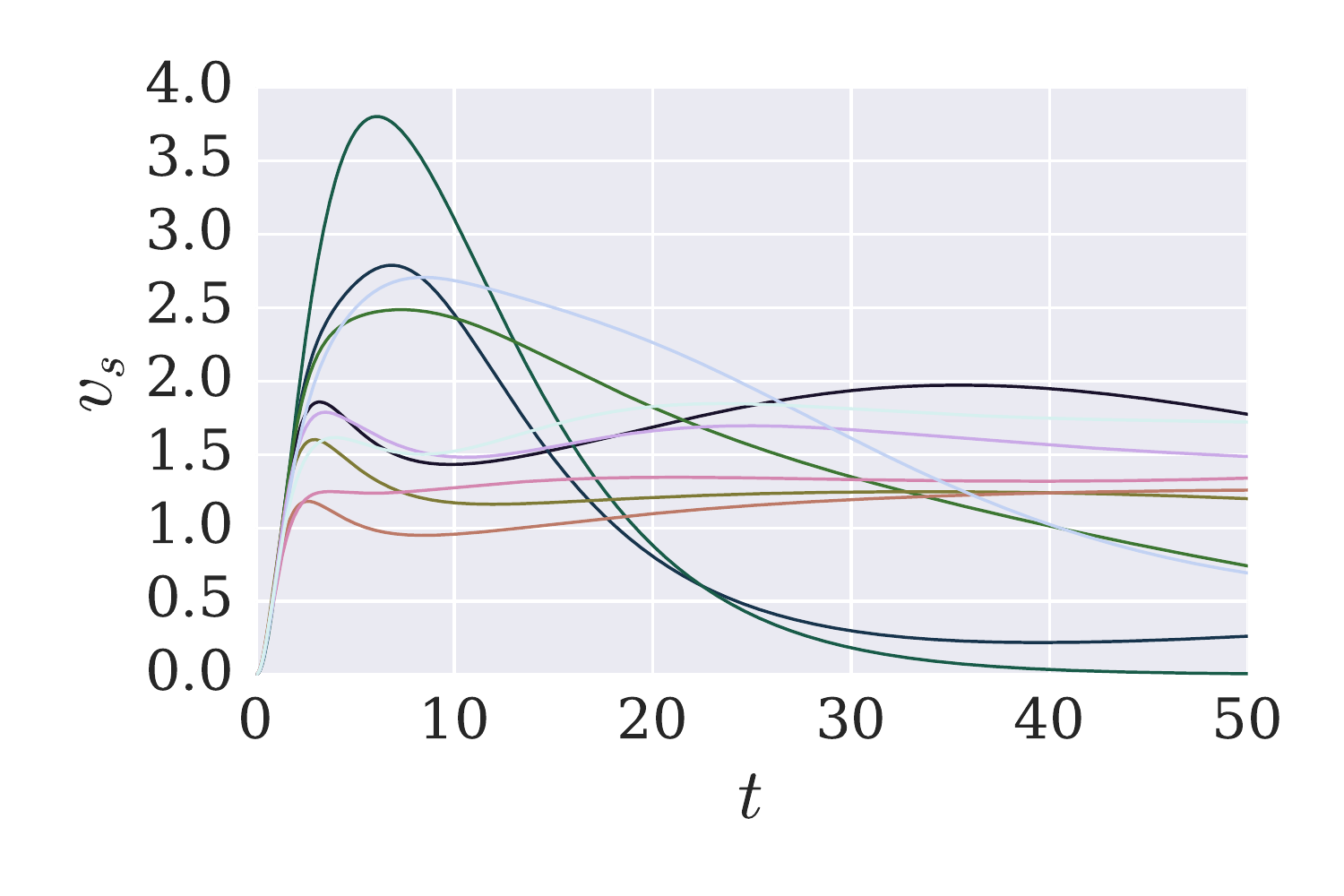}
  }\\
  \subfloat[][Type $(2,0)$ $g_2=-3.5$ $N=10$]{\includegraphics[width=0.42\textwidth]{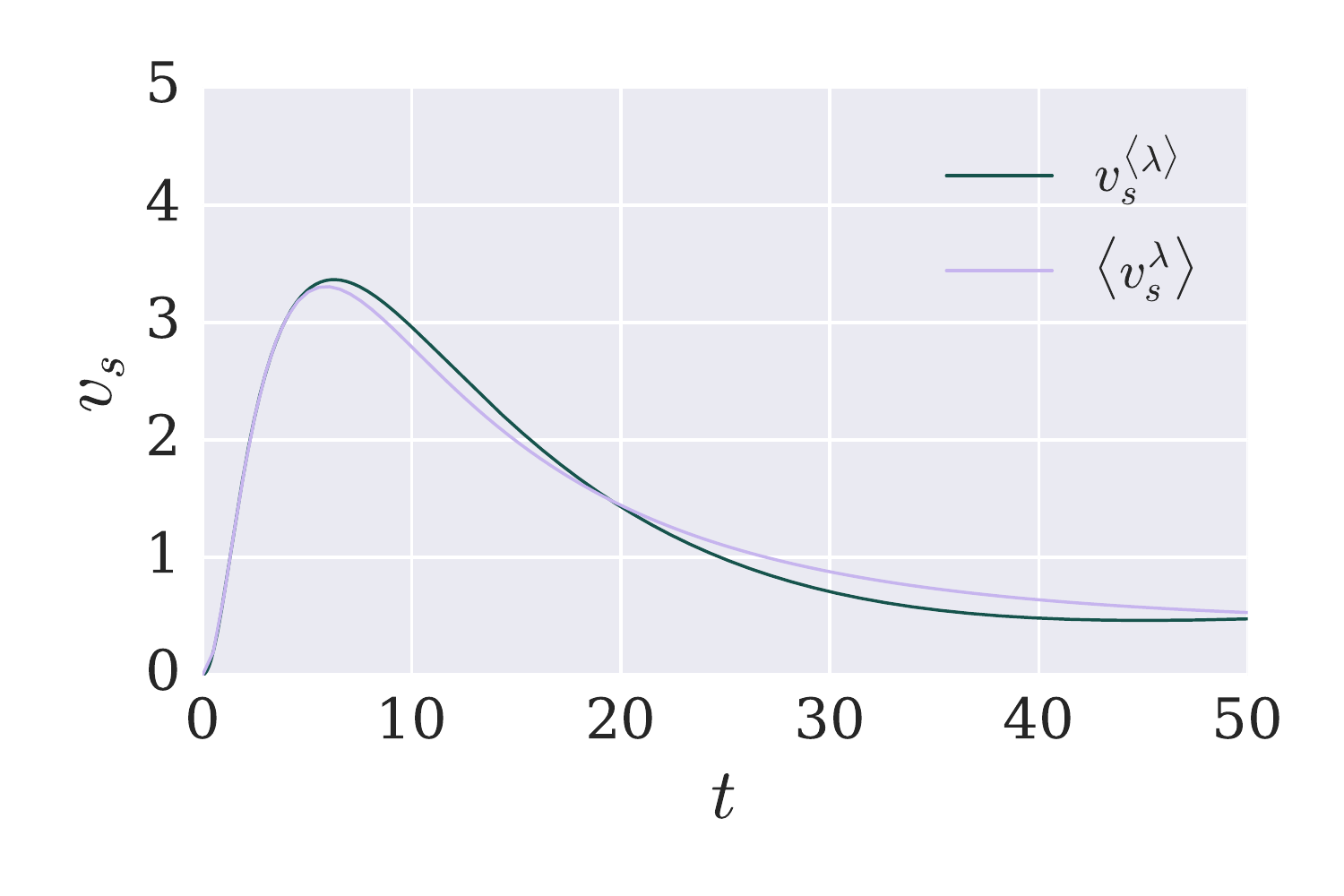}
  \includegraphics[width=0.42\textwidth]{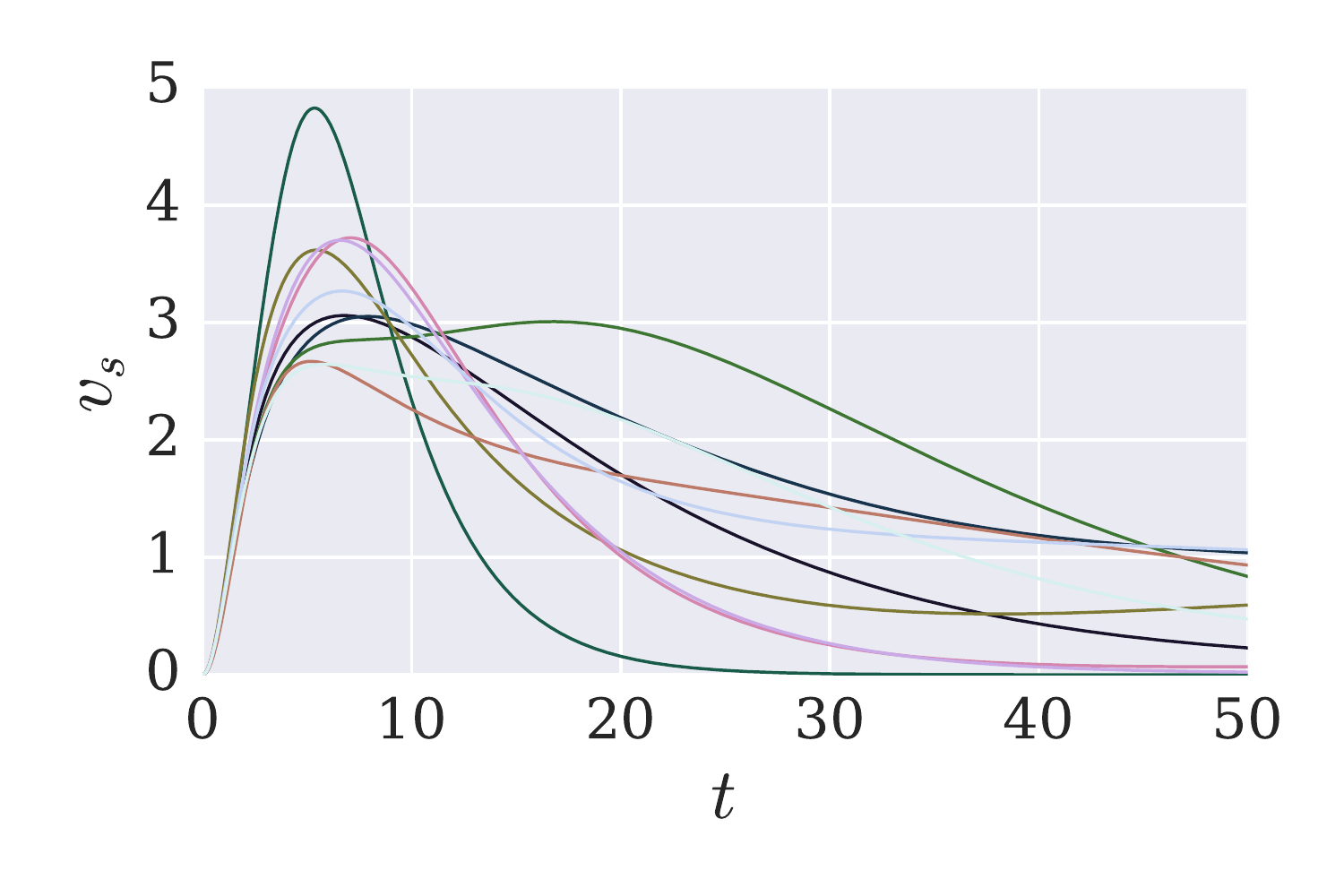}
  }\\
  \caption{\label{fig:vscomp} The left-hand plots are comparing the spectral variance as calculated from the average eigenvalues, $v_s^{\av{\lambda}}(t)$ with the average spectral variance $\av{v_s^{\lambda}(t)}$. In figure (a) the two lines are so similar that they almost can not be distinguished in this plot. The right-hand plots are showing 10 uncorrelated examples of spectral variances from the ensemble of geometries.}
\end{figure}

\subsubsection{Scaling}\label{subsec:EVscale}
If $D$ is a Dirac operator, then so is $\mu D$ for any constant $\mu\in\R$, with the eigenvalues scaling in the same manner. For a manifold the distances between points would scale as $\mu^{-1}$.
The spectral dimension and variance change by a rescaling of the parameter,
\begin{equation}
  d_s(\mu D,t)=d_s(D,\mu^2t), \quad   v_s(\mu D,t)=v_s(D,\mu^2t).
\end{equation}
 If this scaling is regarded as a trivial difference then to compare two Dirac operators effectively, one should rescale one of them to eliminate differences caused by scaling. There are several ways one could do this: the scaling can be such that
\begin{itemize}
	\item the minimum squared eigenvalues $\lambda^2_{\min }$ match
	\item the maximum squared eigenvalues $\lambda^2_{\max }$ match
	\item a scaling depending only on $N$ is used
\end{itemize}

In the case of random geometries, the first two are subject to some statistical uncertainty as one has to compute the scaling from the data.
Experience shows that the minimum eigenvalue is too `random' to be useful for this: since it varies in a range that includes zero, one may find a geometry in which the minimum eigenvalue is uncharacteristically small, leading to a large rescaling that does not reflect the properties of the rest of the geometry.

The maximum eigenvalue is much more stable; in our numerical simulations its variance is small and decreasing with $N$.
So scaling to match the largest eigenvalues is used in the rest of this section when comparing different geometries.
To illustrate this an example is shown in Figure~\ref{fig:lambdamaxfix}, in which random (2,0) geometries are compared with the fuzzy sphere. This is done by rescaling the Dirac operator so that $\lambda_{\max} \to N$.

\begin{figure}
\subfloat[][\label{fig:lmaxfixsmall}$g_2=-2.5$]{\includegraphics[width=0.5\textwidth]{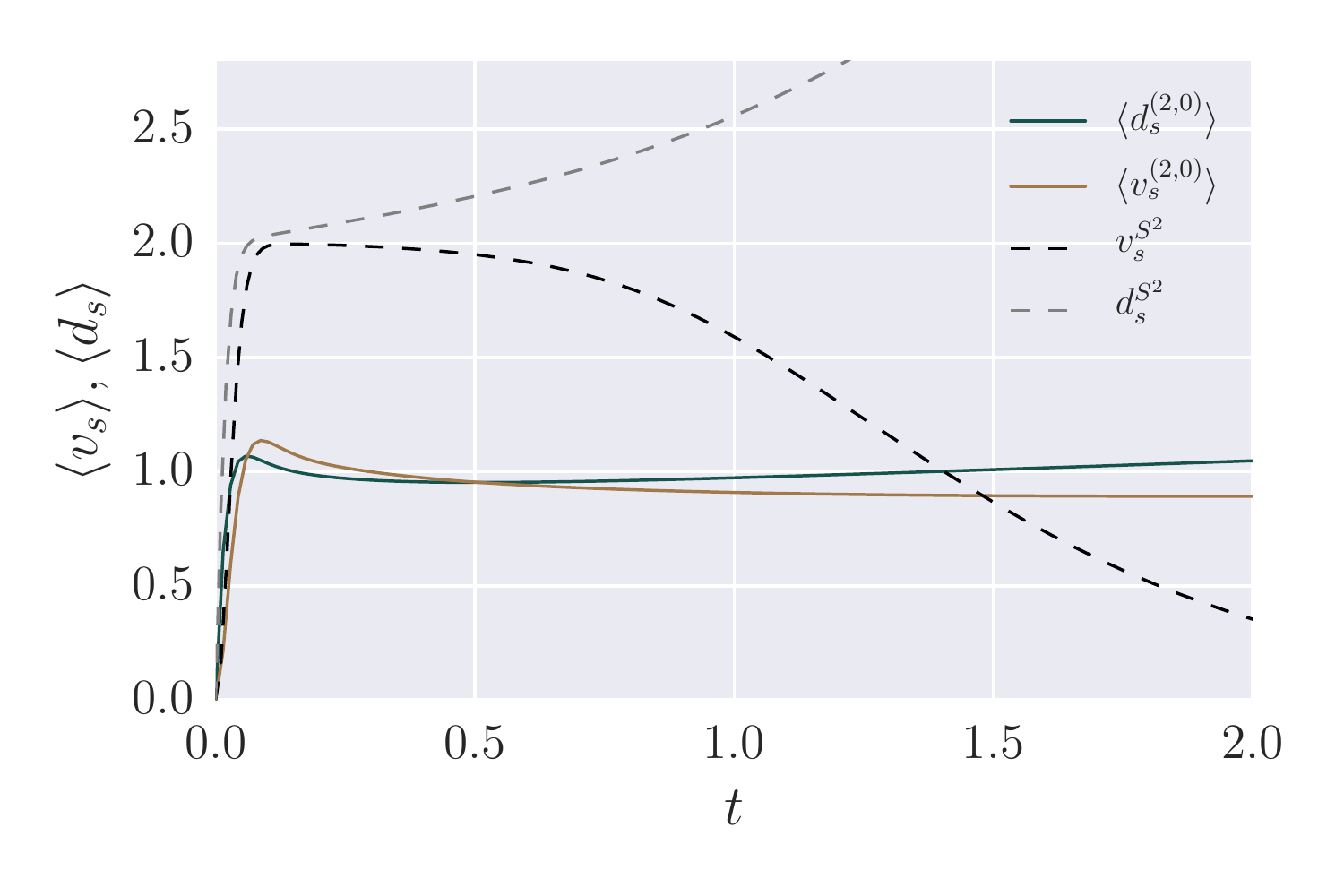}}
\subfloat[][\label{fig:lmaxfixlarge}$g_2=-3.5$]{\includegraphics[width=0.5\textwidth]{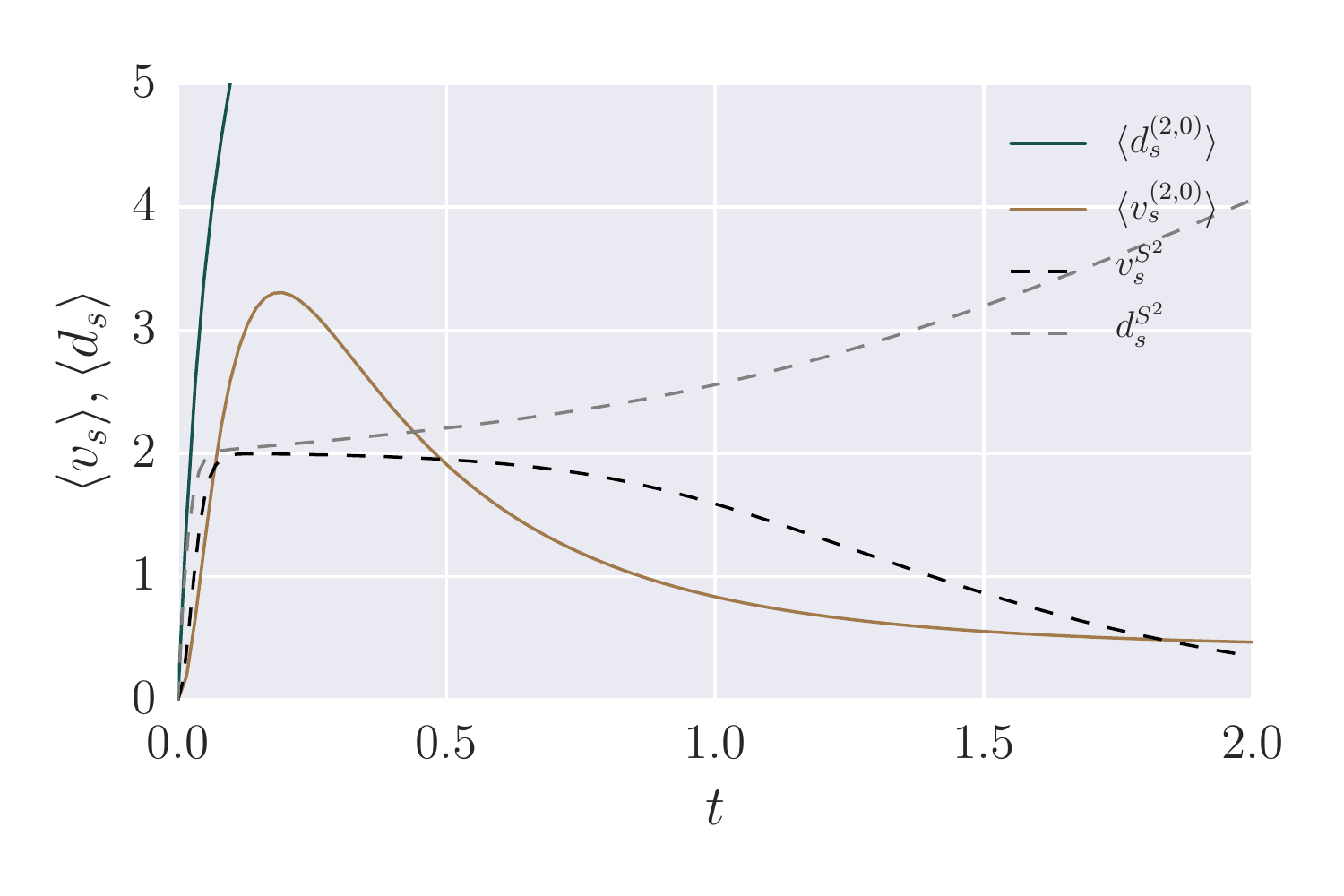}}
	\caption{\label{fig:lambdamaxfix} The spectral variance and dimension for type $(2,0)$ geometries with $N=10$ at $g_2=-2.5,-3.5$. In this figure the spectra were rescaled to match their largest eigenvalues to the largest eigenvalue of the fuzzy sphere for $N=10$.}
\end{figure}

In these plots the spectral dimension and variance for both coupling constants show their characteristic behaviour on similar scales as those of the fuzzy sphere.
In an intuitive sense, rescaling with the maximum eigenvalue can be thought of as an attempt to compare discrete spaces with the same discreteness scale or Planck length.

The choice of action strongly restricts the maximum eigenvalue of the distribution, hence rescaling all eigenvalues with an $N$-dependent factor leads to a result very similar to rescaling them with the maximum eigenvalue.
It is thus sufficient to show only the second option in detail.

The spectral dimension and variance plotted here have error bars, corresponding to the uncertainty on the  average value calculated.
These uncertainties are shown as shaded regions, and are calculated using a jackknife algorithm as described in~\cite{Newman_Barkema_1999}.

\subsubsection{Comparing type $(1,3)$ with $S^2$}
Random geometries of type $(1,3)$ have the same Clifford type as the fuzzy sphere. In addition, the spectra appear to be similar at the phase transition~\cite{barrett_monte_2015,glaser_scaling_2016}. This is investigated further here by comparing the spectral dimension and spectral variance.

The first step taken to explore these random geometries is to look at how the spectral variance changes with the action coupling $g_2$, examined in Figure~\ref{fig:13g2comp}.
The first impression is that the curve changes rapidly around $g_2=-3.7$, which is the  $g_2$ value of the phase transition as found in~\cite{glaser_scaling_2016}. The curve at $g_2=-3.7$ has an interesting behaviour, with a maximum approaching $2$, which is the dimension of the fuzzy sphere. The spectral dimension and spectral variance curves for this $g_2$ are rescaled and compared with the fuzzy sphere in Figure~\ref{fig:13vsS2}.

\begin{figure}[b]
  \begin{minipage}{0.49\textwidth}
    \includegraphics[width=\textwidth]{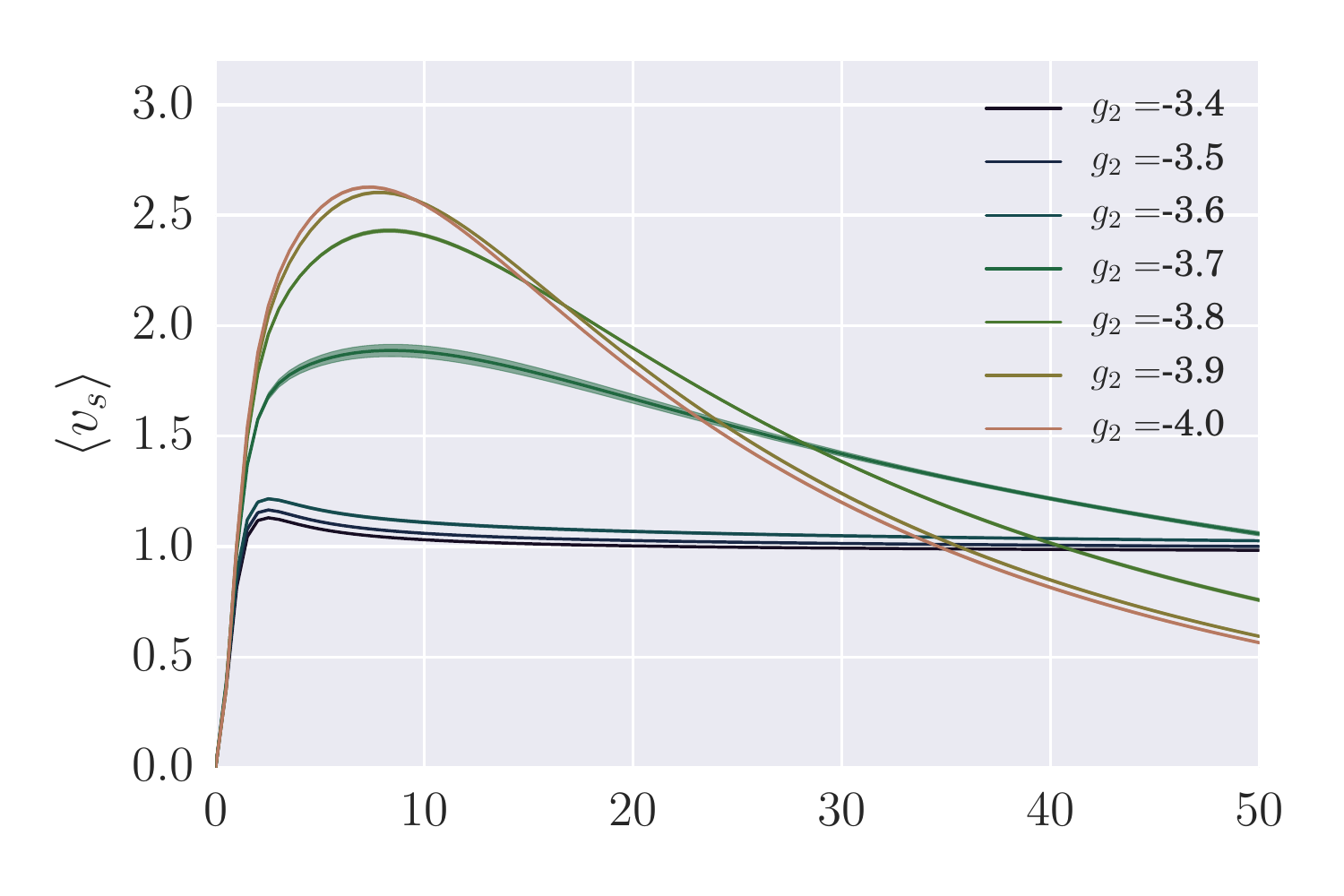}
    \caption{\label{fig:13g2comp}Comparing the spectral variance for geometries of type $(1,3)$ at $N=8$ for different $g_2$. } \end{minipage}\hfill
  \begin{minipage}{0.49\textwidth}
    \includegraphics[width=\textwidth]{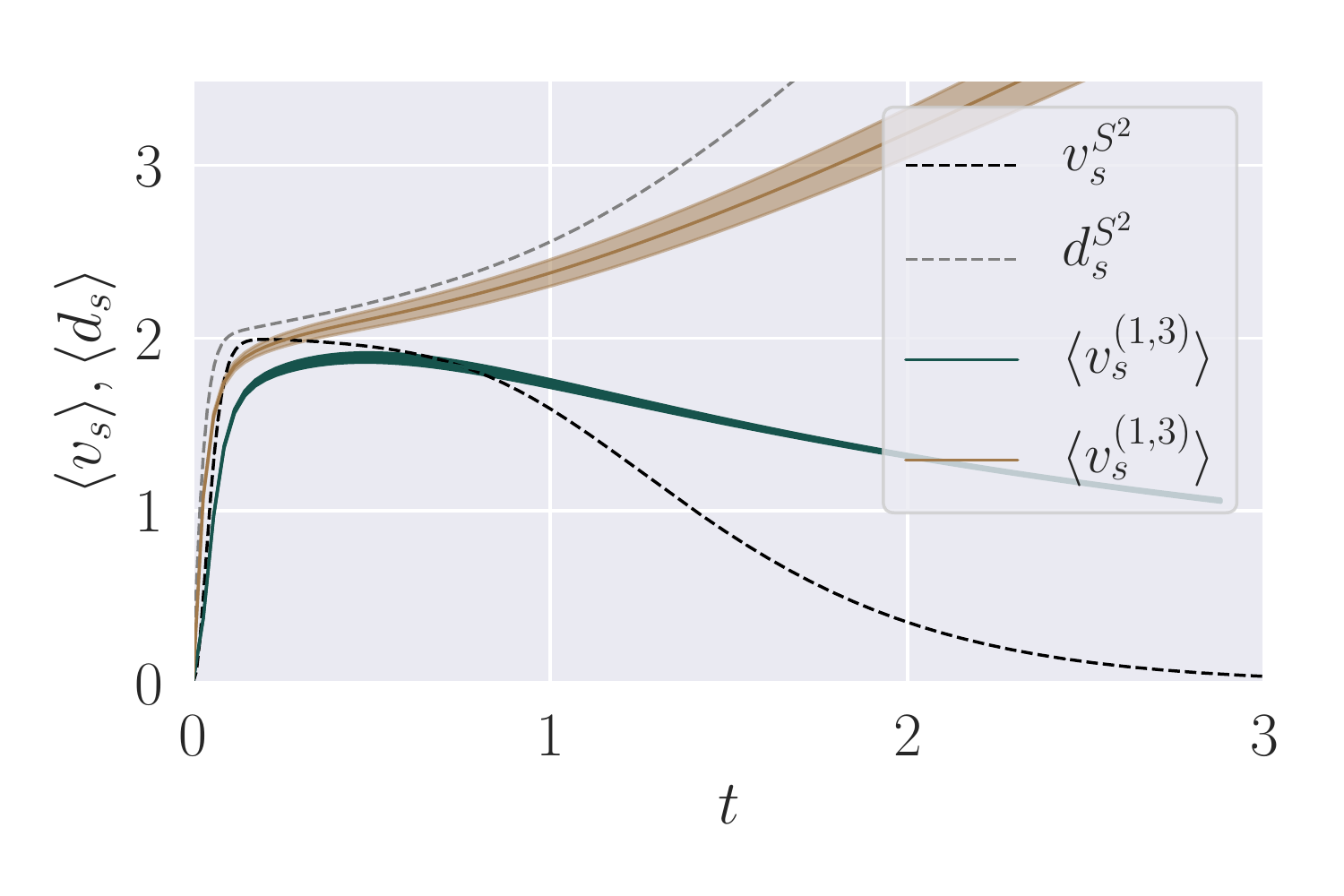}
      \caption{\label{fig:13vsS2}Comparing the spectral dimension and variance for random geometries at the critical point $g_2=-3.7$ for $N=8$ with eigenvalue rescaling to match the maximum eigenvalue to that of the fuzzy $S^2$ of the same size.}
  \end{minipage}
\end{figure}

While the $(1,3)$ geometry and the fuzzy sphere do not agree completely, they behave in a similar way.
For the random geometries the spectral variance and dimension both rise slower at small $t$.
This is due to the fact that the density of eigenvalues for the fuzzy sphere continues to rise linearly almost up to the end, while the growth of the density of eigenvalues for the random fuzzy spaces slows down.

For the random geometries the measures also decay or rise slower at large $t$. This is because the lowest eigenvalue of the random fuzzy geometries is typically much smaller than that of the fuzzy sphere, even after rescaling.
The reason for this is that the fuzzy $S^2$ has degenerate eigenvalues due to the spherical symmetry, while for a generic random fuzzy space the eigenvalues have the minimal allowed multiplicity (which is $2$ for $(1,3)$ geometries, as explained in~\cite{barrett_monte_2015}), to maximise the entropy.

\subsubsection{Comparing type $(2,0)$ and $(1,1)$}
The geometries of type $(2,0)$ and $(1,1)$ were already examined in~\cite{barrett_monte_2015,glaser_scaling_2016}, so it is interesting to add the understanding of their dimension.
The most important difference found in~\cite{glaser_scaling_2016} is that the phase transition for the type $(1,1)$ leads to much weaker correlations and the accompanying shift in behaviour is much more gradual than for type $(2,0)$, hence no large jumps between different $g_2$ values for the type $(1,1)$ geometries are expected.

For type $(2,0)$, the maximum of the spectral variance rises with lowering $g_2$, just as for type $(1,3)$.
As shown in Figure~\ref{fig:type20g2comp}, the spectral variance near the phase transition ($g_2=-2.8$) has a maximum value close to $2$ and is qualitatively similar to the fuzzy sphere. This suggests a 2-dimensional geometry at the phase transition but more work would be needed to substantiate this.

For type $(1,1)$, the spectral variance close to the phase transition ($g_2=-2.4$) reaches values above $2$ and the overall shape of the curve is very different from that of the $2$-sphere. In particular, the curves just below and just above the phase transition are not as different from the curve at the phase transition as they are for type $(2,0)$. This is shown in Figure~\ref{fig:type11g2comp}.

\begin{figure}
\begin{minipage}{0.49\textwidth}
  \includegraphics[width=\textwidth]{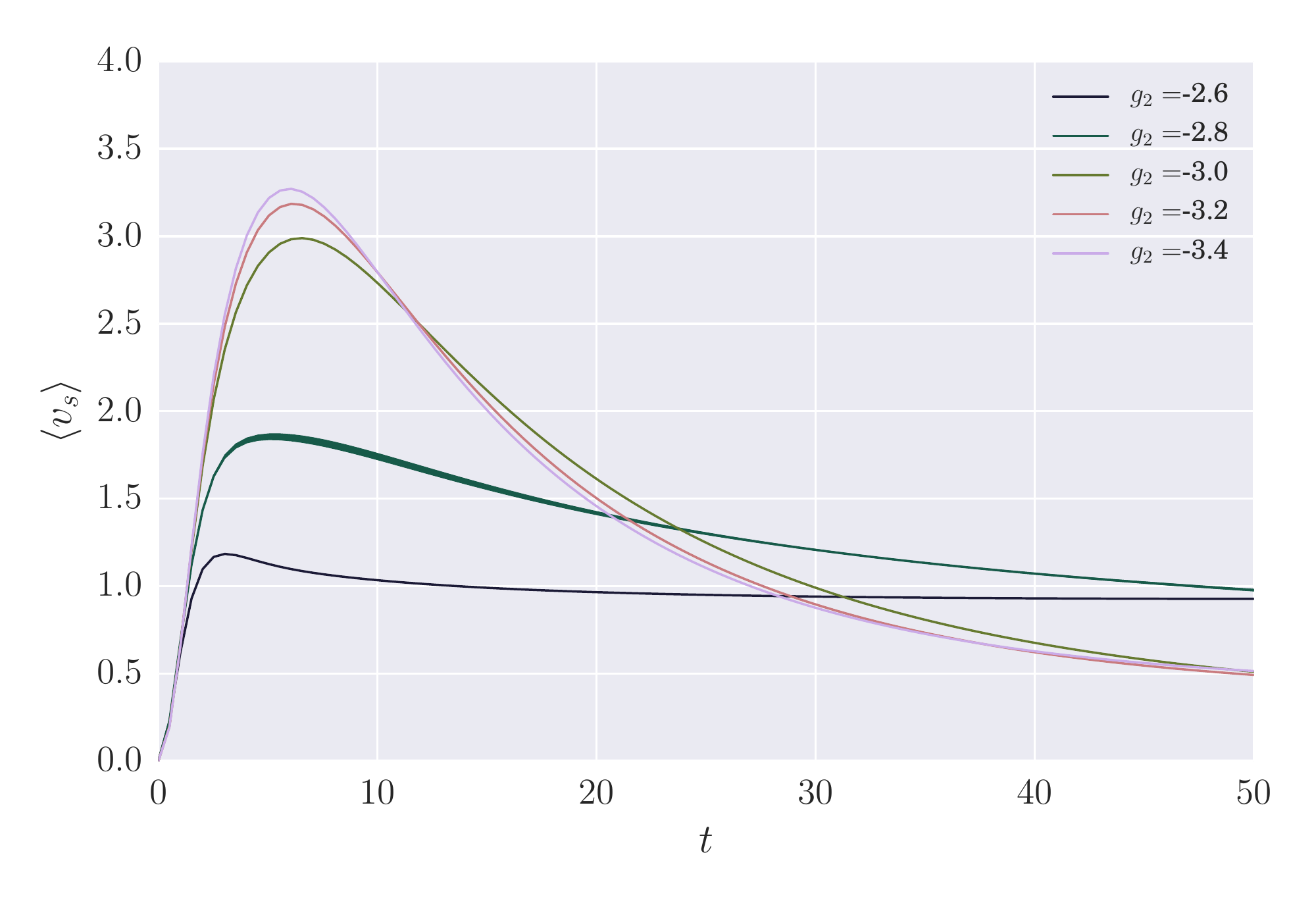}
  \caption{Spectral variance for type $(2,0)$ at different values of $g_2$ for $N=10$.
}\label{fig:type20g2comp}
\end{minipage}\hfill
\begin{minipage}{0.49\textwidth}
  \includegraphics[width=\textwidth]{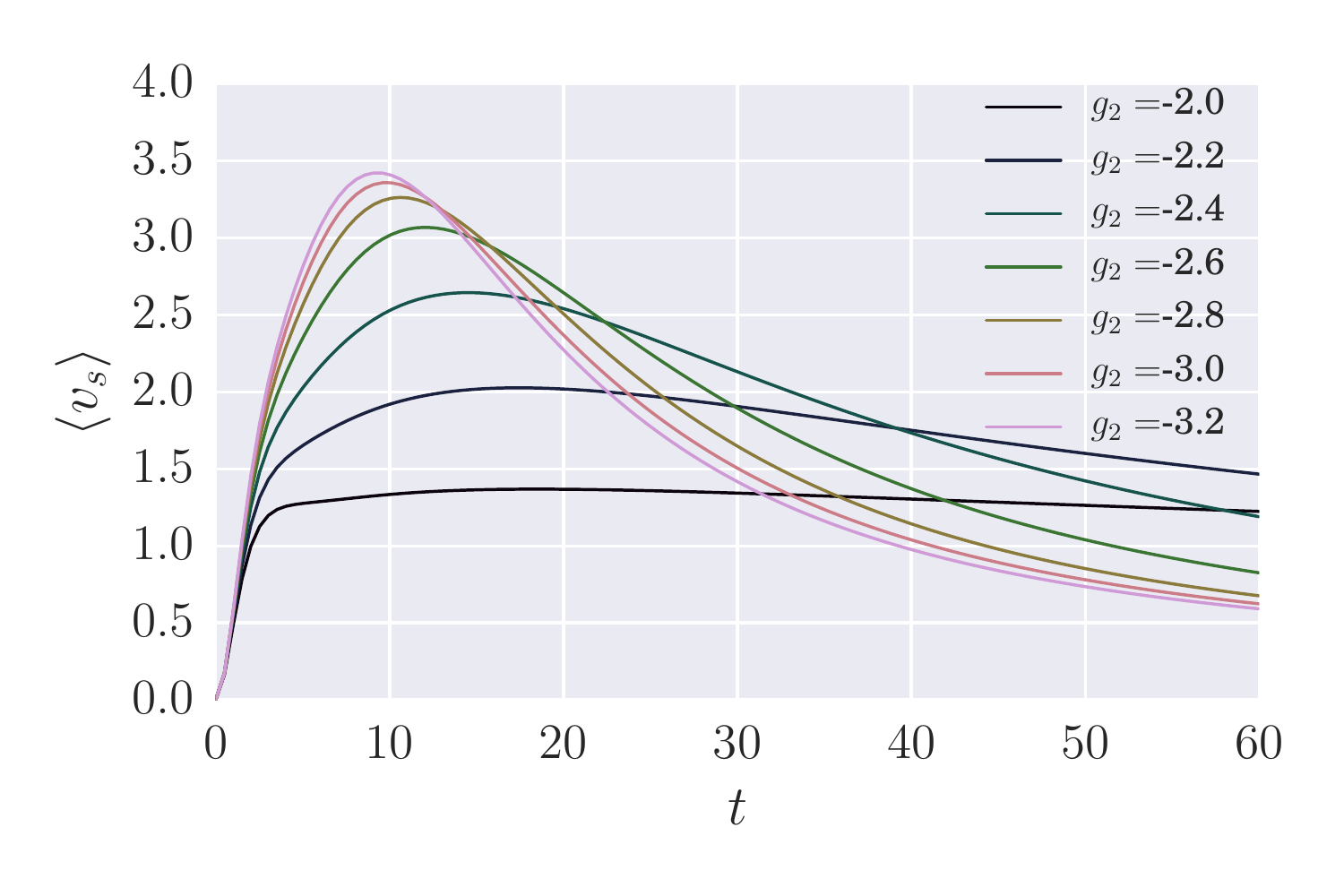}
\caption{Spectral variance for type $(1,1)$ at different values of $g_2$ at $N=10$.}\label{fig:type11g2comp}
\end{minipage}
\end{figure}

It is instructive to plot the spectral variance for all values of $g_2$ for both types as in Figure~\ref{fig:1120comp_mono}, where type $(1,1)$ is petrol-blue and type $(2,0)$ is yellow.
To ensure comparability both are rescaled so that their maximum eigenvalue is $N$.
For type $(1,1)$ the maximum value of the spectral variance is changing slowly, with almost equal distance between the different $g_2$ lines, up to very large $g_2$ for which a saturation seems to be reached.
For type $(2,0)$ on the other hand the curves lie very close together for low $g_2$ and high $g_2$ with a fast change happening in an intermediate regime.
The black dashed line gives the spectral variance of the sphere rescaled in the same way.
Comparing the spectral variances of the two types to this line shows that none of them has particularly similar behavior.

\begin{figure}
\includegraphics[width=0.8\textwidth]{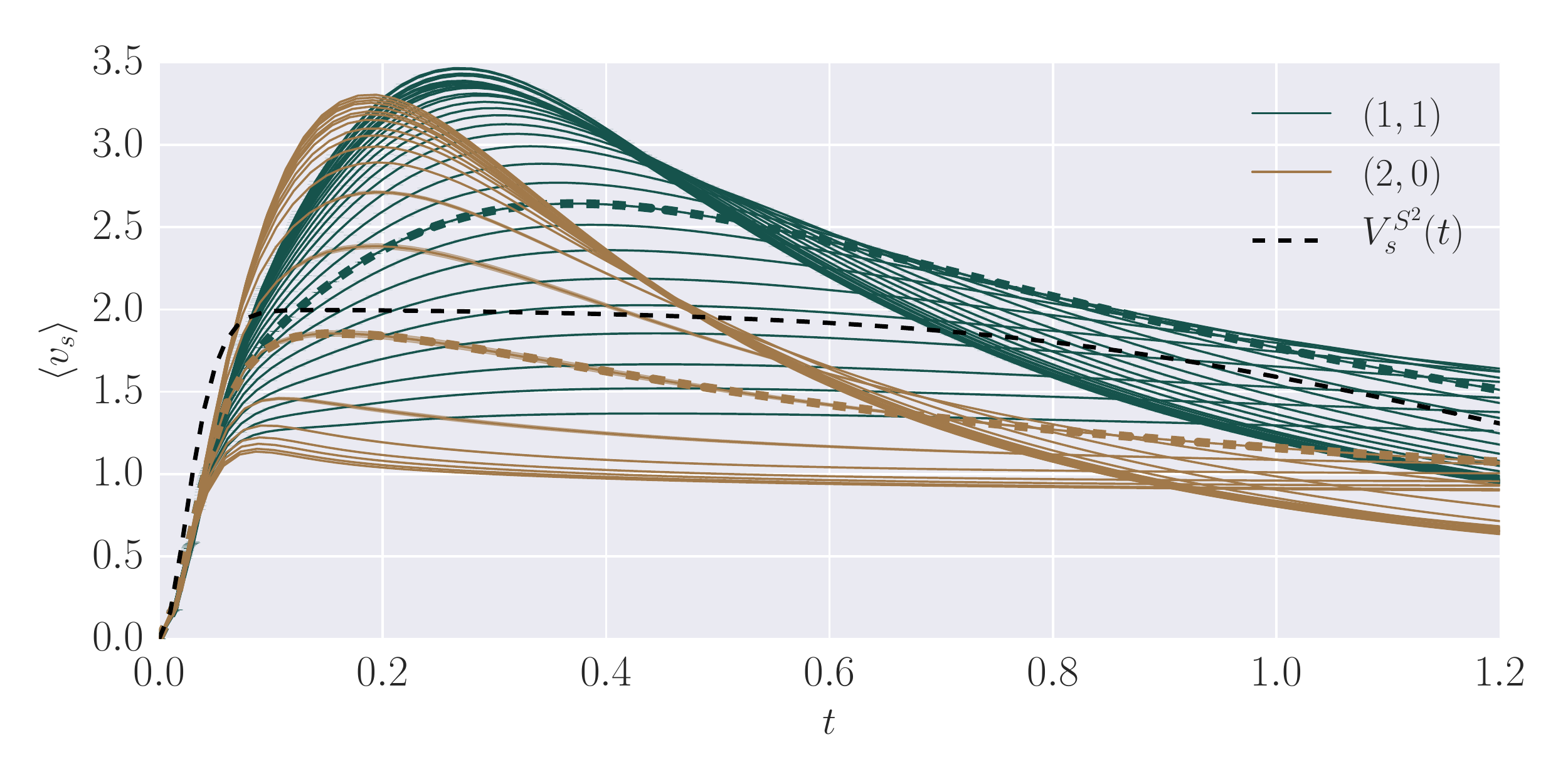}
\caption{Comparing the spectral variance of type $(2,0)$ and type $(1,1)$ random geometries for different couplings $g_2$, rescaled to have maximum eigenvalue $\lambda_{\max}=N$ at $N=10$. The thicker dotted lines for the spectral dimensions of type $(1,1)$ and $(2,0)$ mark the phase transition points.}\label{fig:1120comp_mono}
\end{figure}

These plots show that the tentative conclusion of~\cite{barrett_monte_2015} that the geometries behave similarly  does not survive more detailed examination. It confirms the differences found in~\cite{glaser_scaling_2016} but remains purely qualitative.
To make quantitative judgements one needs more tools, like the zeta-function distance to be introduced in section~\ref{sec:zetadist}.

\subsubsection{The maximum spectral variance}

The average spectral variance curve is zero at $t=0$ and $t\to\infty$ and so has a maximum value. There has been only one local maximum in all of the cases of random geometries studied here. This maximum value $max(\av{v_s})$ is therefore a very crude estimate of the dimension. The limitation of this approach is that the maximum varies widely within each ensemble and so the interpretation is not so clear. Nevertheless, it still proves instructive to plot $max(\av{v_s})$ for the random geometries.

\begin{figure}
\subfloat[][Type $(1,1)$]{\includegraphics[width=0.5\textwidth]{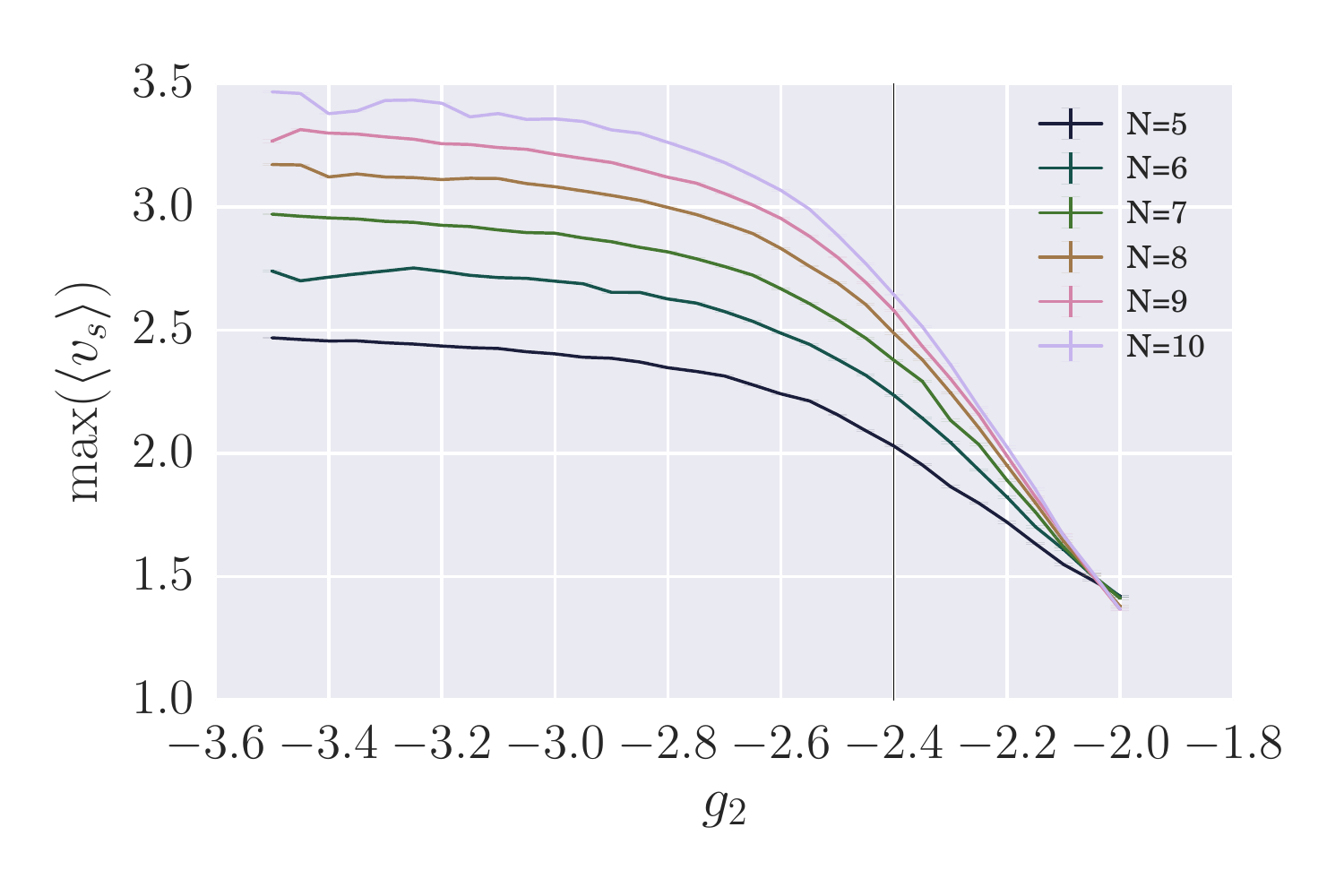}}
\subfloat[][Type $(2,0)$]{\includegraphics[width=0.5\textwidth]{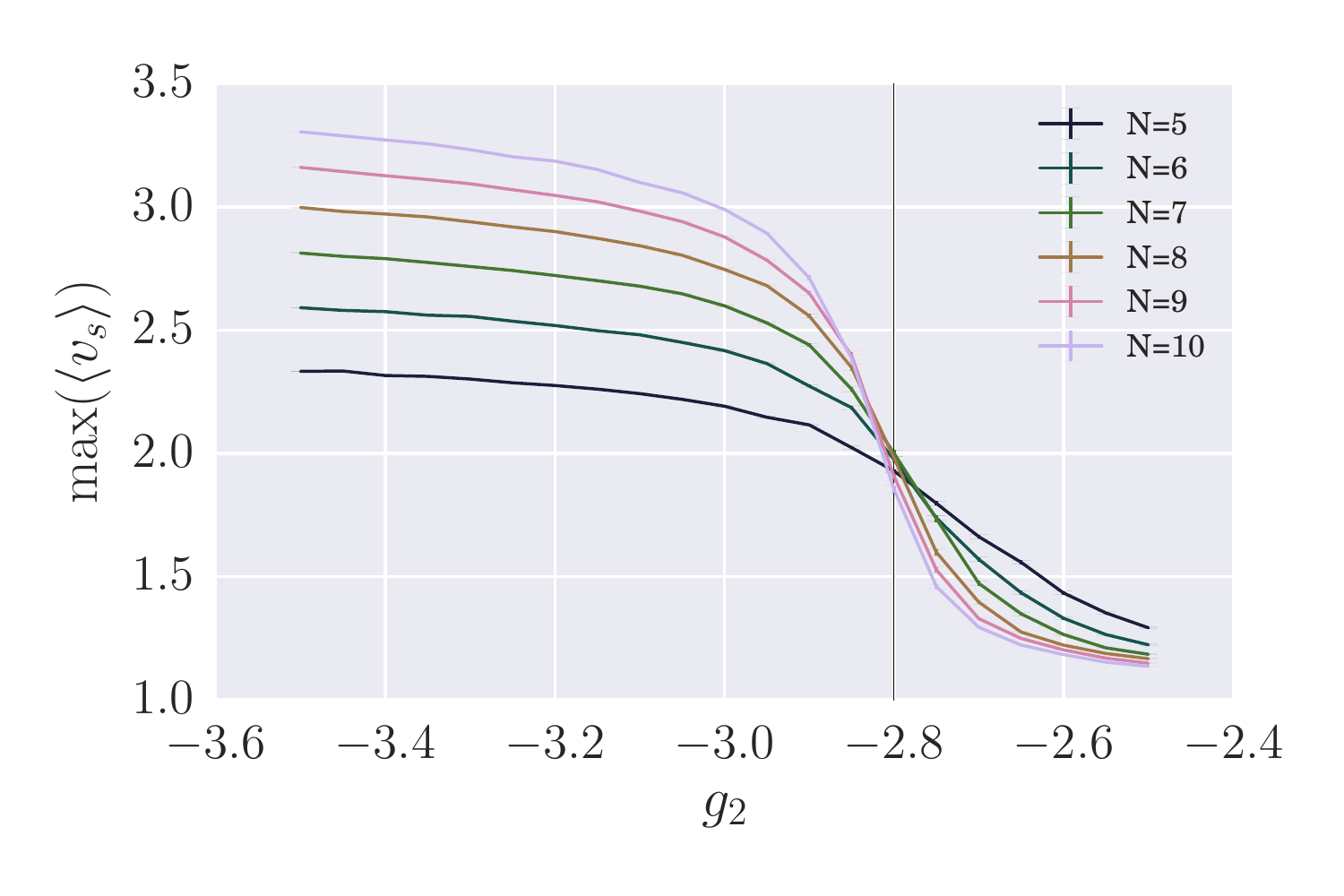}}

\subfloat[][Type $(1,3)$]{\includegraphics[width=0.6\textwidth]{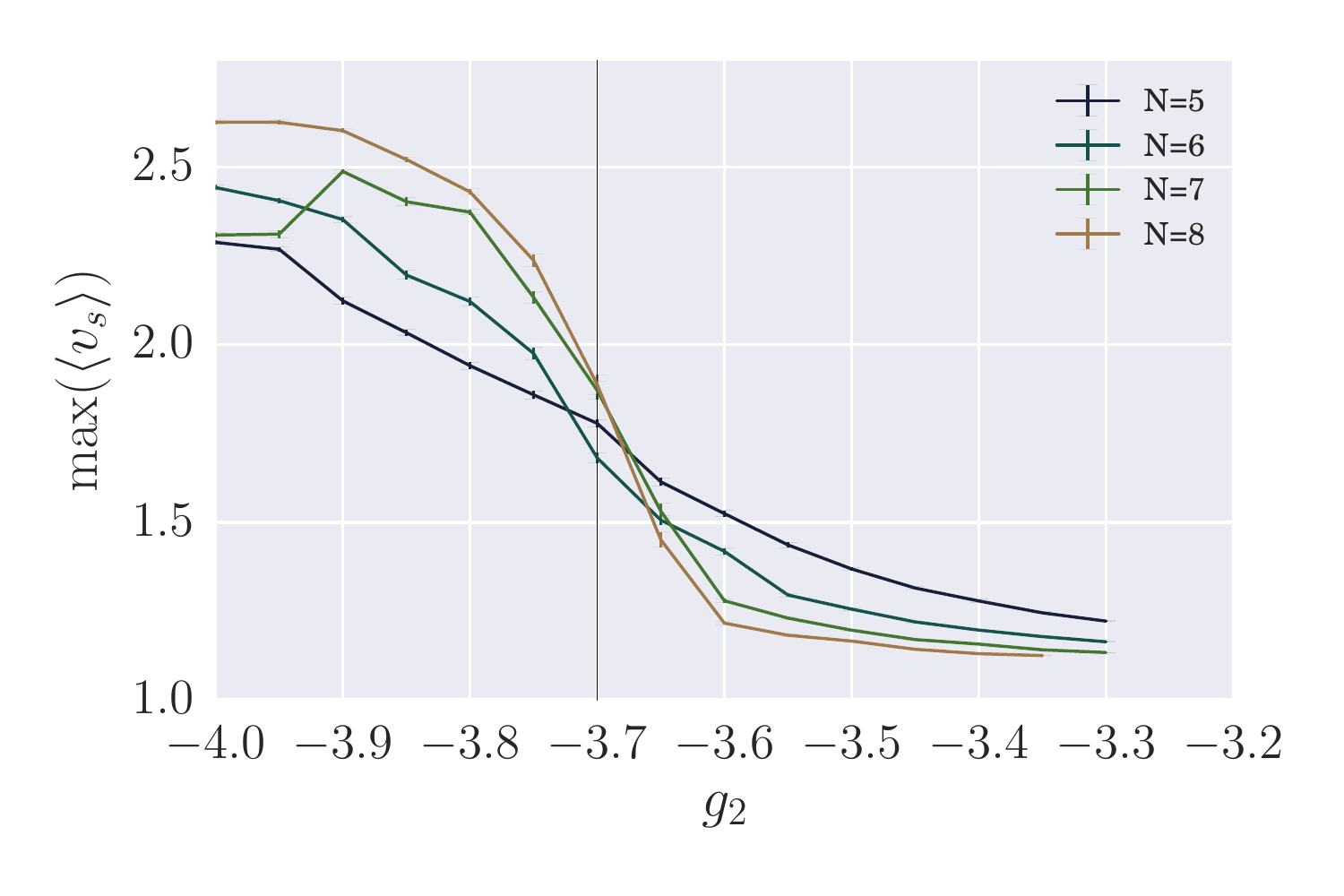}}
  \caption{The maximum value of the average spectral variance plotted against $g_2$.
The vertical lines indicate the phase transition points determined in~\cite{glaser_scaling_2016}.}\label{fig:maxVs}
\end{figure}

Plotting the maximum value against $g_2$ leads to the curves in Figure~\ref{fig:maxVs}.
The separate lines in that figure are for different values of the matrix size $N$.
For fixed $N$, the maximum rises as $g_2$ becomes more negative.
It starts out around $1$ and then rises to large values after the phase transition.
A particularly interesting feature is that for types $(2,0)$ and $(1,3)$ the maximum of the spectral variance around the phase transition seems to be close to $2$ independent of $N$.
An optimistic interpretation of this would be that the behaviour at the phase transition shows a certain scale freedom and might remain the same in the large $N$, continuum, limit.
For type $(1,1)$ on the other hand the point of intersection seems to lie around $g_2=-2.05$ and has a value of about $1.5$. This $g_2$-value is considerably below the phase transition point determined in~\cite{glaser_scaling_2016}, however of the three geometries examined type $(1,1)$ had the least clear signal at the phase transition, which makes the determination of the phase transition in this case less certain.
A priori, no reason exists to expect all of the spectral variances to cross in close proximity.  Therefore this intersection is another marker of interesting behaviour at the phase transition.

\section{Volume measures}\label{sec:zeta}

As discussed above, Weyl's law connects the dimension and the volume of a Riemannian manifold to the growth of the eigenvalues of the Dirac operator. This suggests that it might be possible to extend the definition of volume to the case of fuzzy spaces.
This section summarises the problems encountered in trying to do this.

\subsection{The zeta function and asymptotic volume measures}
\label{sec:ZetaAndVol}
One method to extract information about a manifold from its spectrum is to study the spectral zeta function of the Laplace-Beltrami operator~\cite{minakshisundaram_properties_1949}.
Many of the results obtained thus generalise to operators of Laplace type, such as the squared Dirac operator examined in this work~\cite{LawsonJr:1989ub}.
Let $\{ \lambda_i\}$ be the set of non-zero eigenvalues of the Dirac operator $D$ on a compact Riemannian spin manifold. Then, for large enough $\mathrm{Re}(s)$, its spectral zeta function (\SPZ) is defined as
\begin{align}
\zeta(s) = \sum\limits_n {(\lambda_n^2) }^{-s} \;,
\label{eq:zeta1}
\end{align}
and can be analytically continued to a meromorphic function on the whole complex plane.
As $D^2$ is a positive operator, the spectral zeta function is the Mellin transform of the heat kernel~\cite{Gilkey}
\begin{align}
\zeta(s) &= \frac{1}{\Gamma(s)}\int\limits_0^\infty t^{s-1} K(t) \md t \;.
	\label{eq:HKzeta}
\end{align}
Here,  $\Gamma(s)$ is Euler's gamma function. Using the asymptotic expansion of the heat kernel of~\eqref{eq:HKexp}, one can compute the expansion coefficients $a_i$ by
\begin{align}
	a_i = \Res_{s=(d-i)/{2}}(\Gamma(s)\zeta(s)).
\end{align}
Note that these definitions also work in the case of the Dirac operator from a finite spectral triple, in which case the only poles are those of the gamma function.

It is known that for a closed Riemannian manifold, $M$, with $D^2$ an operator of Laplace type on a vector bundle over $M$, the expansion coefficients are only non-zero for even values of $i$
and can be expressed in terms of local geometric invariants~\cite{Vassilevich:2003ik}.
The first heat kernel coefficient is
\begin{align}
  a_0(D^2) &= \frac{k}{{(4\pi)}^{{d}/{2}}}  \Vol(M)
\end{align}
where $k$ is the dimension of the spinor space at a point.
This term corresponds to the right-most pole of the zeta function at $s= {d}/{2}$.
This pole can be seen at $s=1$ in the continuum $2$-sphere, for which the \SPZ{} is proportional to the Riemann zeta
function $\zeta_R(2s-1)$ (c.f.~Figure~\ref{fig:ZetaS}).
\begin{figure}
\centering
\includegraphics[width=0.5\textwidth]{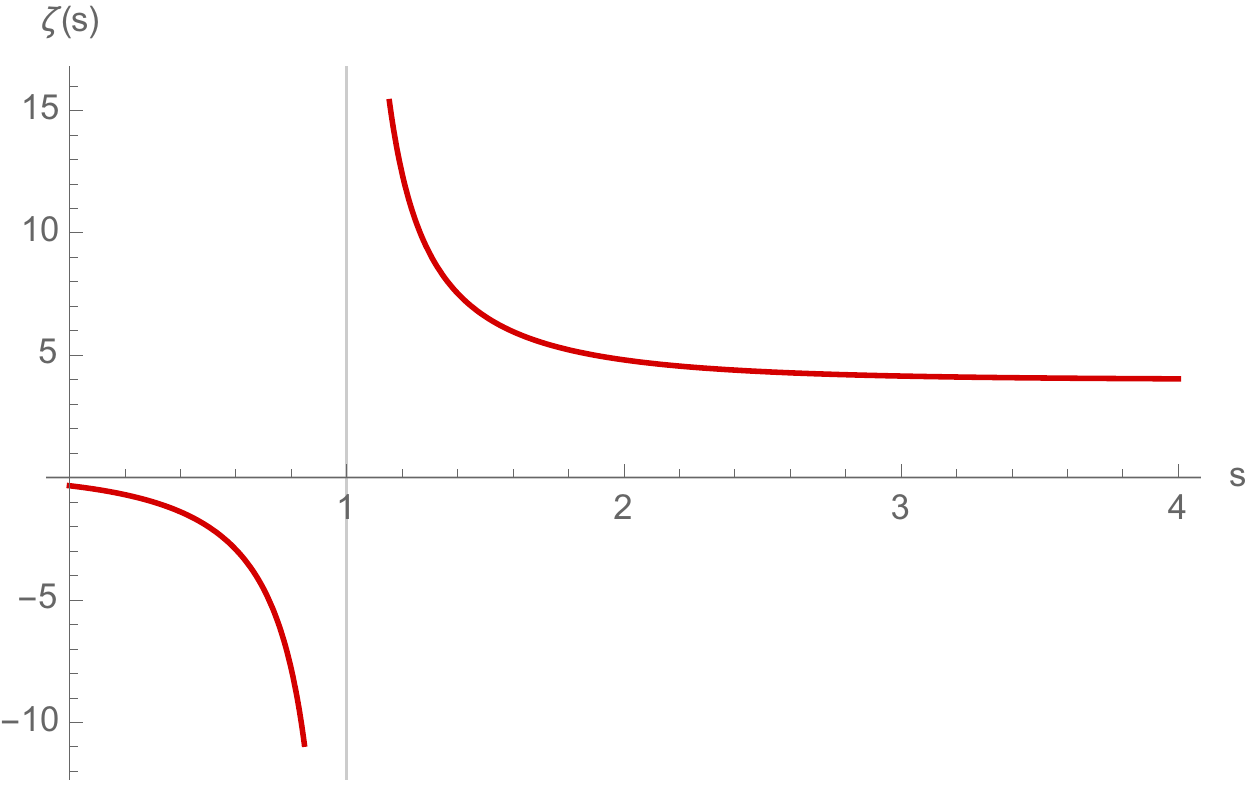}
\caption{\label{fig:ZetaS}The \SPZ{} for the continuum $2$-sphere. The pole at $s=d/2=1$ indicates the dimension.}
\end{figure}

For a fuzzy space, such as those outlined in~\cite{barrett_matrix_2015}, there are a finite number of eigenvalues, and hence the sum in~\eqref{eq:zeta1} is finite, which automatically regularises the poles. One way to investigate the finite analogue of the poles is to use a finite analogue of the Dixmier trace. On a manifold, with the eigenvalues ordered so that the sequence $\lambda_n^2$ is non-decreasing,
one can use the following
\begin{equation}
b_j(s)=\frac1{ \log  j}\sum_{n=1}^j (\lambda_n^2)^{-s} ,
\end{equation}
with $s=d/2$ to approximate the \SPZ{} residue at $s=d/2$.
The Dixmier trace is the limit of the bounded sequence $b_j(d/2)$~\cite[Proposition 4 p.306]{Connes1994}  and
\begin{equation}
\lim_{j\to\infty} b_j(\frac{d}{2})=\frac{2}{d} \Res_{s=d/2}\zeta(s).\label{eq:dixmierlimit}
\end{equation}
Thus it is assumed that $b_j(d/2)$ is a good approximation for $2/d$ times the residue of the pole, according to~\eqref{eq:dixmierlimit}.

In the spirit of the previous sections, one might consider using the poles of the zeta function to define a dimension measure.
For fuzzy spaces that approximate a continuum manifold in the large $N$ limit, the location of this pole should become apparent when comparing the zeta function for different $N$ values.
In this case the largest value of $s$ for which the zeta function diverges logarithmically in $j$ would be the location of the pole.
This is done in Figure~\ref{fig:ZetaFuzzyS} for the fuzzy sphere, where the sum $b_j(s)$ is plotted.
For the fuzzy sphere defined with matrix size $N$ the number of eigenvalues is $j=4 N^2$.
The pole is where the plots cross, as $j\to\infty$.
\begin{figure}
\includegraphics[width=0.65\textwidth]{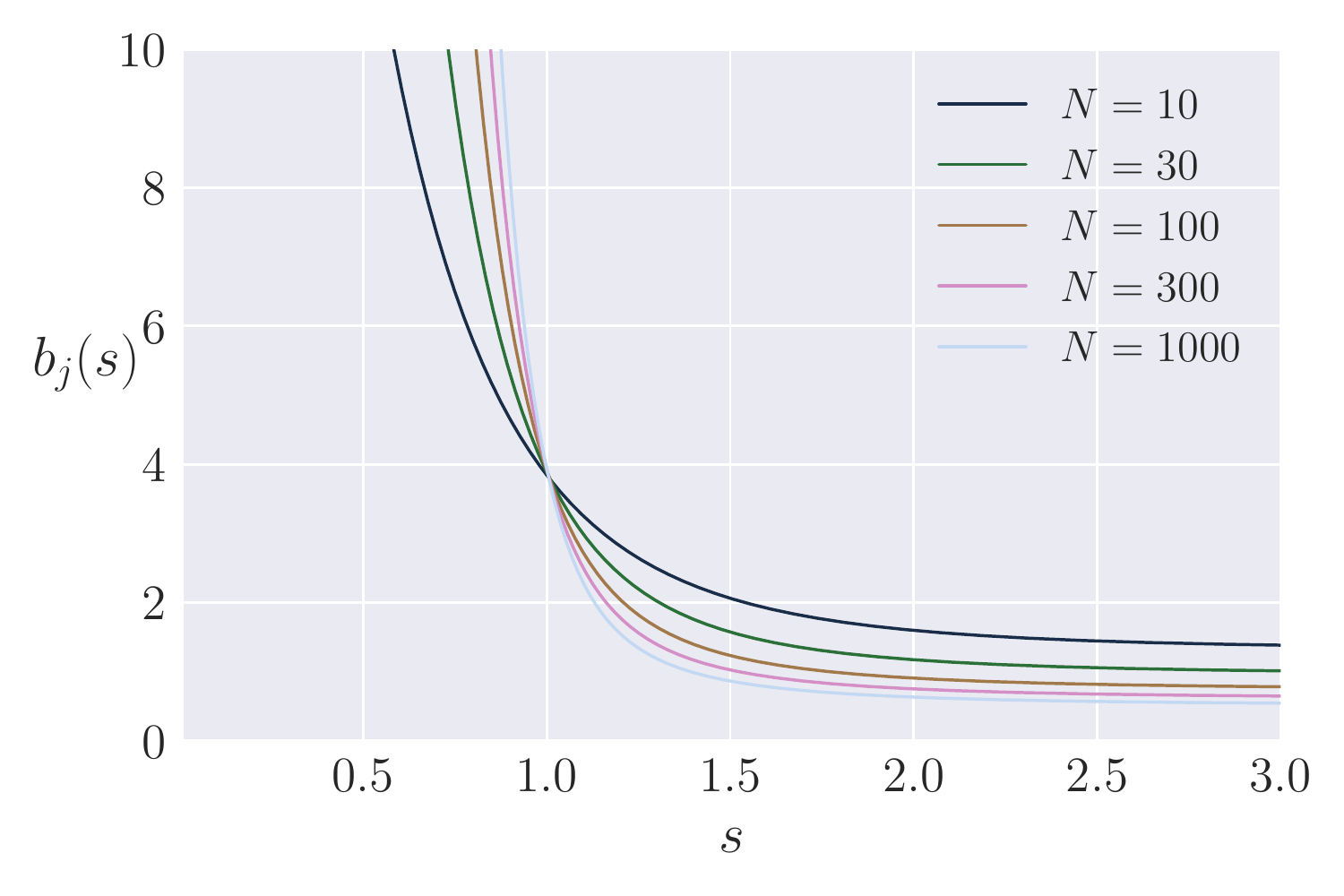}
\caption{\label{fig:ZetaFuzzyS}The \SPZ{} for the fuzzy $2$-sphere for different matrix sizes $N$ rescaled by $\log(j)$, with $j = 4 N^2$ the total number of eigenvalues. The intersection point of the graphs identifies the dimension at $s=d/2$. }
\end{figure}

While this works nicely on the sphere, it does not to lead to a practical method in more general cases.
For one, because locating the pole only works if one has the `same' geometry at each $N$, and so is unsuitable for random geometries.
In addition, the convergence is logarithmic. Hence even for well-understood geometries large matrix sizes are needed to obtain good results.

It is, however, possible to use this expression for the Dixmier trace to give a new definition of a volume from a finite number of eigenvalues, providing there is already an estimate $d$ of the dimension,
\begin{equation}\label{eq:vol}
  Vol^\mathrm{Dix}(D) = \frac{d}{2} \frac{(4\pi)^{d/2}}{k}  \Gamma\left(\frac{d}{2}\right) \frac{1}{\log j}\sum\limits_{\lambda}(\lambda_n^2)^{-d/2}\;.
\end{equation}
For a manifold, the sum is over the $j$ smallest eigenvalues and this gives the volume approximation considered above.
For Dirac operators with a finite number of eigenvalues (e.g., a fuzzy space), the sum is over all $j$ eigenvalues and this gives a definition of a volume measure for this non-commutative `space'.

This expression obeys the scaling property
\begin{equation}
\Vol^\mathrm{Dix}(\mu D)= \mu^{-d}\Vol^\mathrm{Dix}(D)
\end{equation}
that also holds for the volume of a Riemannian manifold.
However it is not additive,
\begin{equation}
\Vol^\mathrm{Dix}(D_1\oplus D_2)\ne \Vol^\mathrm{Dix}(D_1)+\Vol^\mathrm{Dix}(D_2)
\end{equation}
in general.
It is not additive even if $D_1=D_2$, since $\log(2j)\ne\log j$. In this case, $D=D_1 \oplus D_1$ just doubles the multiplicities of the eigenvalues of $D_1$.
Thus the definition is not consistent because taking two particles on a fuzzy space (which means doubling the spinors and replacing $k$ with $2k$) would lead to a change in the volume.

Another, more recent, proposal to calculate the volume was given in~\cite{stern_finite-rank_2017}, in the context of truncating the spectrum of the Dirac operator to retain the eigenvalues below a cut-off, $|\lambda|\le\Lambda$.
The expression was constructed to have improved convergence properties in the $\Lambda \to \infty$ limit, and gives the volume as
\begin{align}
    \Vol_\Lambda^\mathrm{St}(D)= \frac{ \left(4 \pi \epsilon(\Lambda) \right )^{d/2}}{ek \Gamma(1-\frac{d}{2},1)}
        \, \sum_{\lambda\colon|\lambda|\le\Lambda}
     \frac{e^{- \lambda^2\,\epsilon(\Lambda)}}
     {1+\lambda^2 \,\epsilon(\Lambda)}
     \label{eq:volA}
\end{align}
where $\epsilon(\Lambda)=(2 \log{\Lambda})/\Lambda^2$, and $\Gamma(s,u)$ is the upper incomplete gamma function with lower integration limit $u$.
Stern's result is that for the Dirac operator of a manifold $\Vol_\Lambda^\mathrm{St}(D)\to \Vol(M)$ as $\Lambda\to\infty$.

Again, this volume measure can be adapted to the case of a fuzzy space by summing over all eigenvalues.
There are now two parameters that are needed for this formula, $d$ and $\Lambda$.
It is not necessary for the definition to insist that
$\Lambda\ge|\lambda|_\mathrm{max}$ for all eigenvalues, though for compatibility with the original formula $\Lambda$ may be taken to be either the maximum $|\lambda|$ or possibly an estimate for it. This could be determined by a dimensionful coupling constant in the action for a random fuzzy space, for example ${g_4}^{-1/4}$ in~\eqref{eq:action}.

The formula does not have the scaling property with $\Lambda$ fixed, i.e., $\Vol_{\Lambda}^\mathrm{St}(\mu D)\ne \mu^{-d}\Vol_\Lambda^\mathrm{St}(D)$, as one might expect since $\Lambda$ determines a fundamental length scale. The scaling property could be restored by simultaneously changing the value of $\Lambda$.

The formula is additive, so providing $d$ is the same, one has
\begin{equation}
\Vol^\mathrm{St}_\Lambda(D_1\oplus D_2)= \Vol^\mathrm{St}_\Lambda(D_1)+\Vol^\mathrm{St}_\Lambda(D_2).
\end{equation}

Equations~\eqref{eq:vol} and~\eqref{eq:volA} can be used to calculate the volume of a Dirac operator, assuming the dimension of the fuzzy space is known.
To test these expressions, they can be applied to the spectra of the fuzzy sphere and the fuzzy torus using $d=2$ and $\Lambda$ the maximum eigenvalue.
This is shown in Figure~\ref{fig:volumeST}, together with the volume of the continuum sphere and the continuum torus.
In the captions, equation~\eqref{eq:vol} is referred to as the Dixmier trace and equation~\eqref{eq:volA} as the Stern volume.
\begin{figure}
    \subfloat[][Fuzzy Sphere]{
    \includegraphics[width=0.45\textwidth]{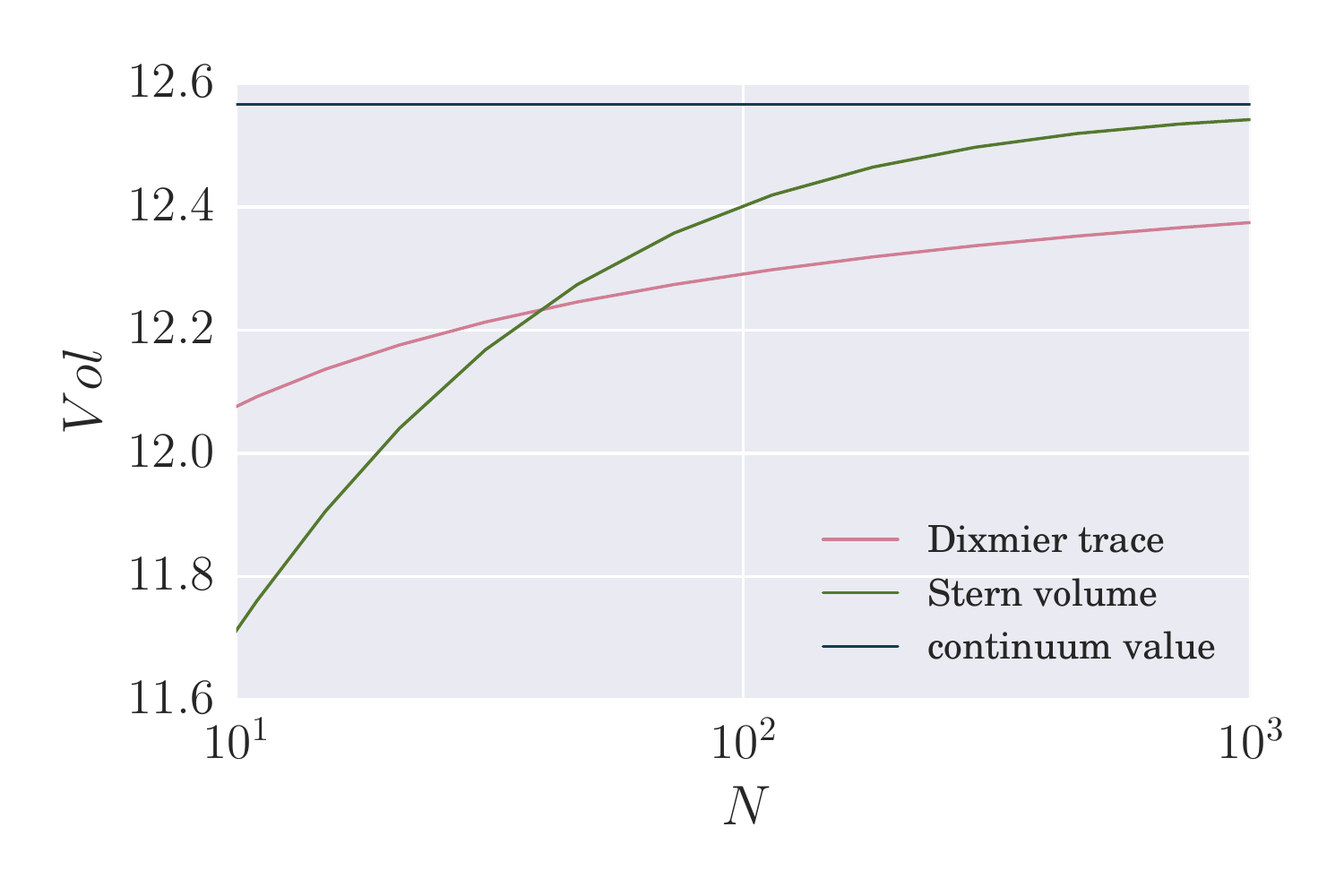}}
    \subfloat[][Fuzzy Torus]{\includegraphics[width=0.45\textwidth]{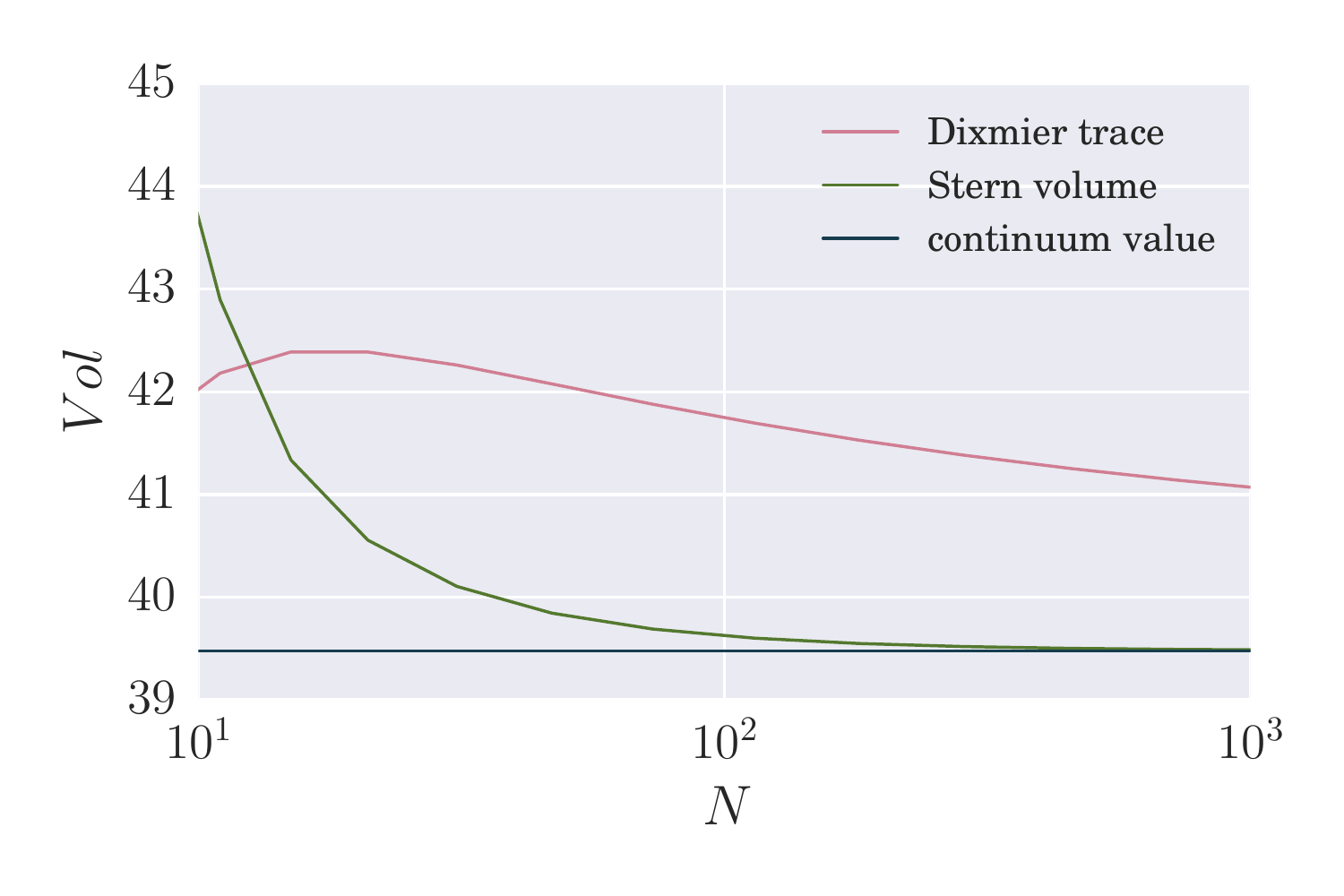}}
    \caption{Volumes of the fuzzy sphere and fuzzy torus. The volumes of the continuum geometries are included for comparison.}\label{fig:volumeST}
\end{figure}
As expected, the Stern volume converges much faster with $N$.

If the dimension parameter is not known, it can be estimated using the spectral variance at a given value of the parameter $t$, which determines an energy scale. The volume of the Dirac operator of two different fuzzy tori are shown in Figure~\ref{fig:torusVolumeEnergy}, plotted against $t$. Of course this agrees with the values of the volume shown in Figure~\ref{fig:volumeST}(b) when the spectral variance (Figure~\ref{fig:torusCvs}) is close to the value $2$, which is approximately the region $0.2 < t < 1$. It also shows that the volume changes away from this quite rapidly as the dimension estimator changes.
\begin{figure}
\subfloat[][Torus $a=d=1$, $b=c=0$]{\includegraphics[width=0.5\textwidth]{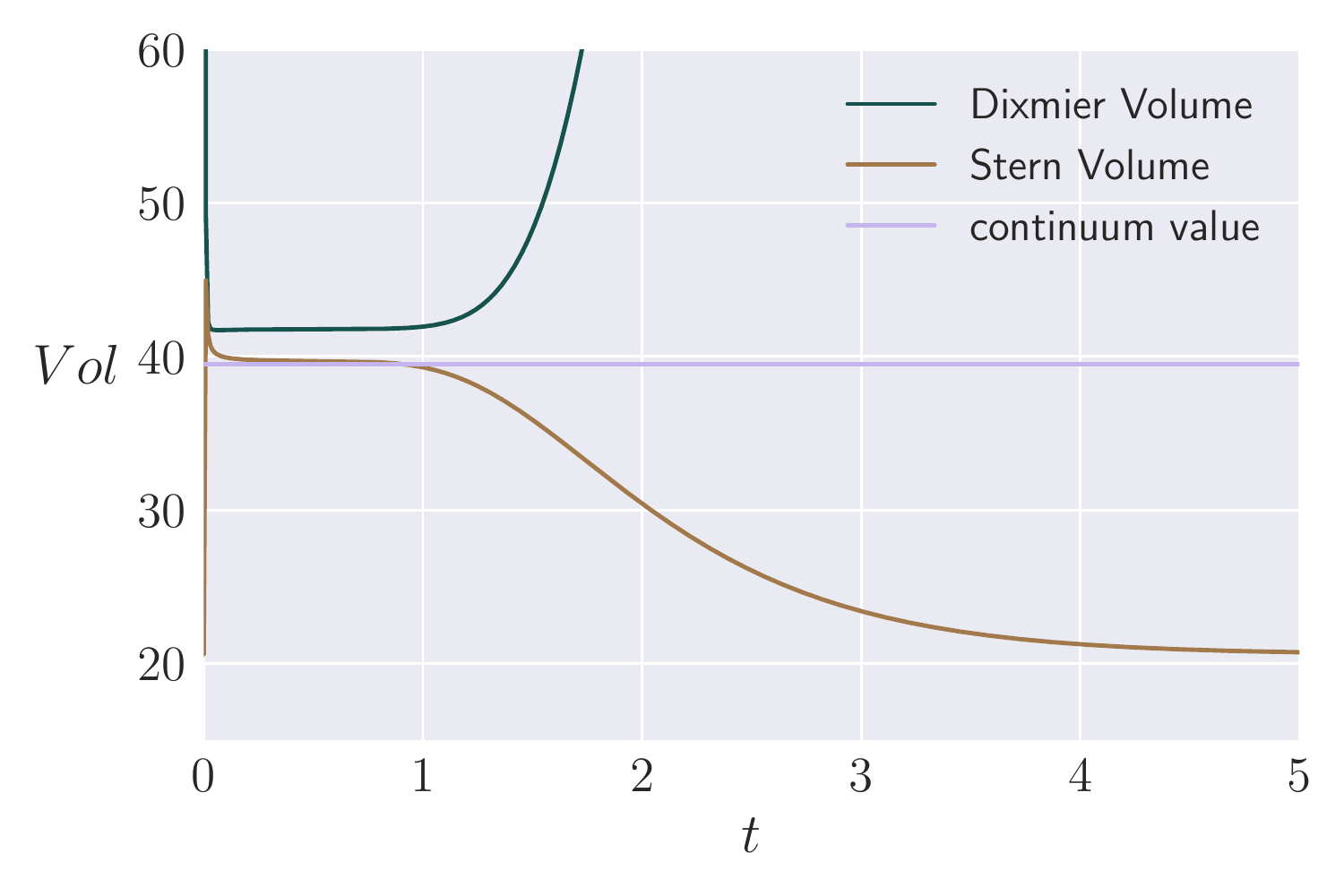}}
\subfloat[][Torus $a=3$ $d=1$, $b=c=0$]{\includegraphics[width=0.5\textwidth]{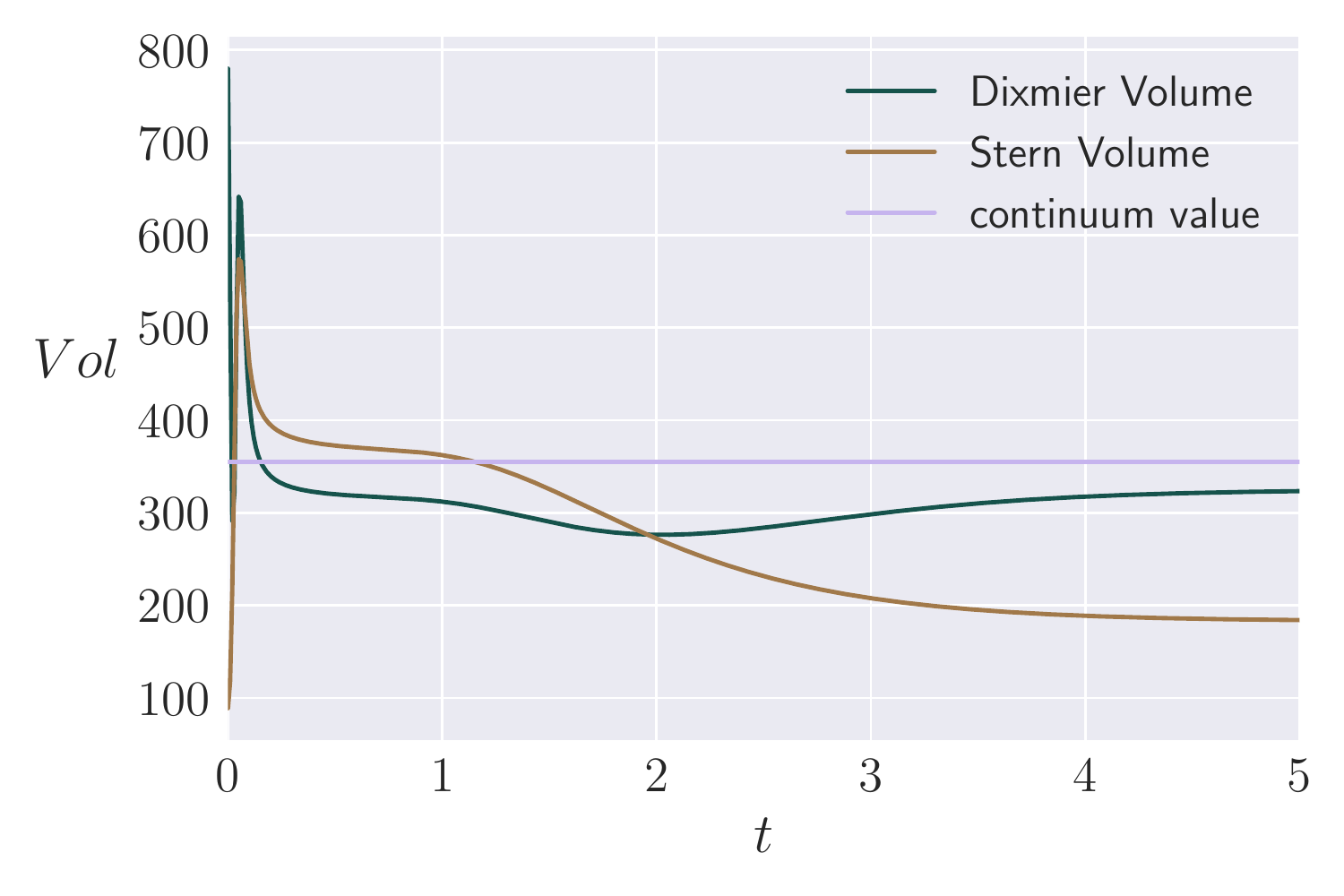}}
    \caption{Volumes of two different tori plotted against $t$, using the dimension estimator $d=v_s(t)$ and $N=90$ for both. }
 \label{fig:torusVolumeEnergy}
\end{figure}

Testing the volume measures on random fuzzy geometries leads to several methodological conundrums.
The first of these is the question of dimension.
For the torus and the sphere their topological dimensions are known, and can be used to calculate their volume; however no such information is available about the random fuzzy spaces.
One could use the spectral variance defined above to define the dimension in the volume, but this is scale dependent. Even more troubling is that in the ensembles of random geometries used in this paper, the spectral variance fluctuates substantially, as can be seen in Figure~\ref{fig:vscomp} on the right. The maxima of the spectral variances of the individual Dirac operators have a large variance (of order 1), and moreover, the value of $t$ for these maxima also has a large variance.
The question is, how could one average over the volume of different geometries if these geometries are of different dimension?
The average volume of a cube, a square and a line makes no physical sense, since their volumes are different quantities.

There are possible solutions to this problem, such as measuring the volume of an average geometry, e.g. defined by the averages of the sorted eigenvalues, using the spectral variance curves in Figure~\ref{fig:vscomp} on the left. This could lead to a unique number for the volume but it is not clear if that usefully represents the properties of the ensemble.

Another option would be to define a dimensionless volume, by dividing by some characteristic volume scale, such as the Planck scale, leading to a measure effectively counting the number of Planck volumes within a given volume. However, there are a number of possible ways of defining the Planck volume, so this is not a clearly defined proposal.
Another solution might be to average the volume measured as a function of dimension.

In summary, there are a number of possible ways to define the apparently simple notion of the volume of a geometry but it is not yet clear if there is a useful definition in the case of random spaces.

\section{Zeta distance}\label{sec:zetadist}

In~\cite{cornelissen_distances_2017} a distance between two geometries is defined using the ratio of the spectral zeta functions of the Laplace-Beltrami operator. This definition is adapted here to use Dirac operators instead. Let $D_1$, $D_2$ be Dirac operators and $\zeta_1$, $\zeta_2$ their zeta functions, as in~\eqref{eq:zeta1}. A real number $\gamma$ is chosen so that it is greater than the real part of any pole of either zeta function.

 The distance between the geometries is defined to be
\begin{align}
  \sigma(D_1,D_2)= \sup_{\gamma\le s\le\gamma+1} \bigg| \log{\bigg( \frac{\zeta_1(s)}{\zeta_2(s)}\bigg) } \bigg| \;.
\end{align}

Note that the zeta functions are positive for the relevant values of $s$. The distance has the property that $\sigma(D_1,D_2)=0$ if and only if the spectra are the same\cite{cornelissen_distances_2017}.
It also obeys the triangle inequality
\begin{equation} \sigma(D_1,D_2)+ \sigma(D_2,D_3)\ge \sigma(D_1,D_3).\end{equation}
The definition uses the closed interval $[\gamma,\gamma+1]$ but the $1$ is just for convenience and in fact any finite interval with a suitable lower limit $\gamma$ will do.
The definition also holds for finite spectra, in which case there is no restriction on $\gamma$. It can even be used to compare a finite spectrum to an infinite one. Several examples are examined in detail in the following sections.


\subsection{Convergence of fuzzy spectra to the continuum}

The distance measure can be used to define the convergence of spectra. If $\sigma(D_n,D)\to0$ for $D_n$ a sequence of Dirac operators, then $\zeta_n(s)\to\zeta(s)$ for each $s\in[\gamma,\gamma+1]$, providing the poles all lie below $\gamma$. According to ~\cite[Thm 3.2]{cornelissen_distances_2017}, this implies that the spectra and multiplicities converge pointwise in a suitable sense. The converse situation is examined here using the sphere and torus as concrete examples. In these examples, the spectra converge pointwise and this implies that the zeta functions converge pointwise in $s$. However, this is not quite good enough to show that $\sigma$ converges to $0$, as this needs uniform convergence in $s$.

The fuzzy spheres define a sequence of Dirac operators $D_n$, with $n=N$, the matrix size. This is compared to the Dirac operator $D$ on the spin bundle of $S^2$ tensored with $\C^2$, so that the multiplicities are doubled. Then the $D_n$ have the same spectrum as $D$ but with a cut-off. Therefore $\zeta_n(s)\to\zeta(s)$ as $n \to \infty$ for $s>1$ and this is uniform in $[\gamma,\gamma+1]$ for any $\gamma>1$. Therefore $\sigma(D_n,D)\to0$. This conclusion generalises to any sequence of fuzzy spaces for which the spectra are obtained by truncating the spectrum of the limiting geometry.

Now let $D_n$ be a sequence of square fuzzy tori with $N=n$, and $D$ the continuum torus, again with multiplicities doubled. Here, the spectra do not coincide but converge pointwise, when ordered in increasing value. The fact that the zeta functions converge uniformly for $\gamma>1$  is shown in Appendix~\ref{sec:appendix}. This implies again that $\sigma(D_n,D)\to0$.

\begin{figure}[b]
    \includegraphics[width=0.49\textwidth]{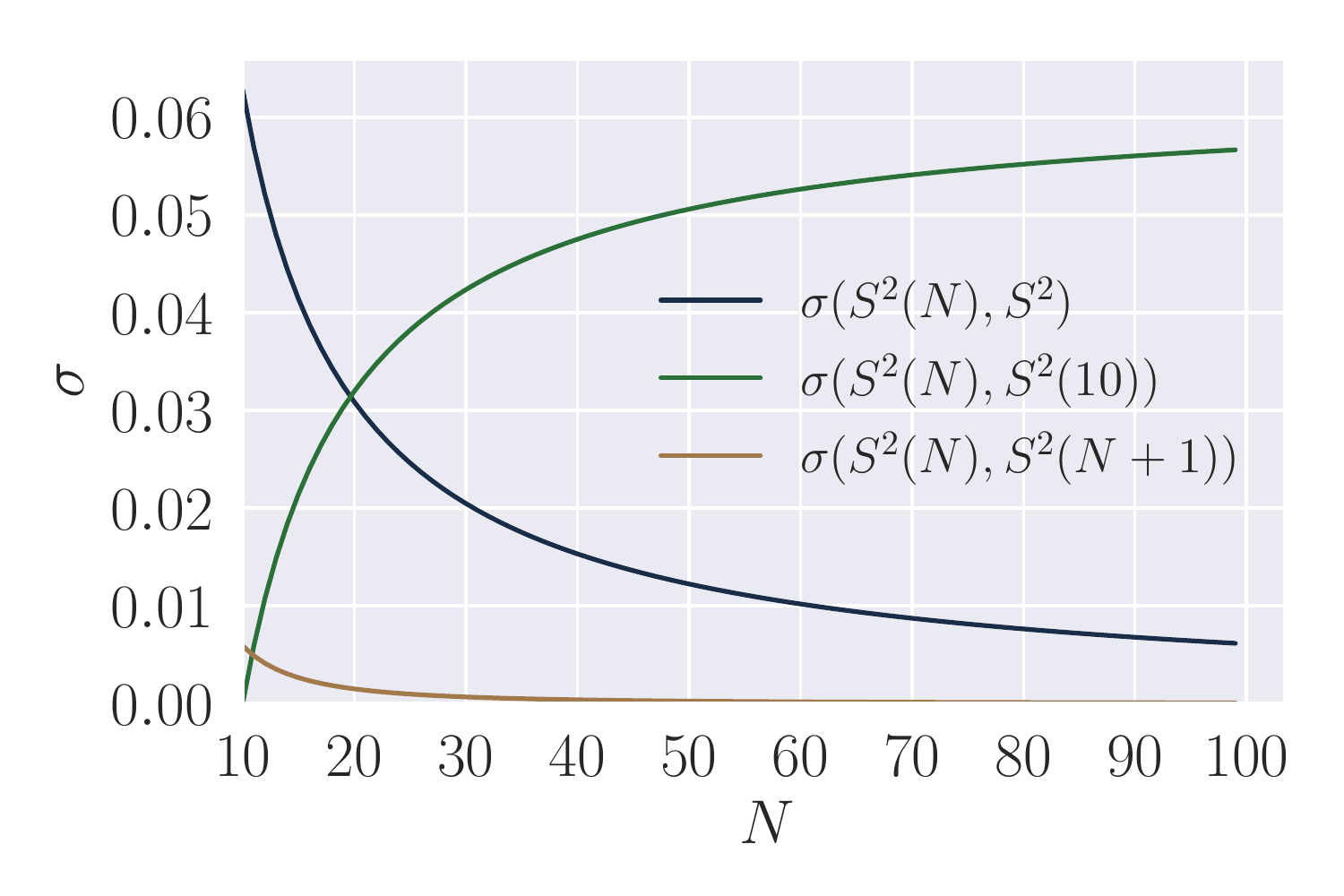} \includegraphics[width=0.49\textwidth]{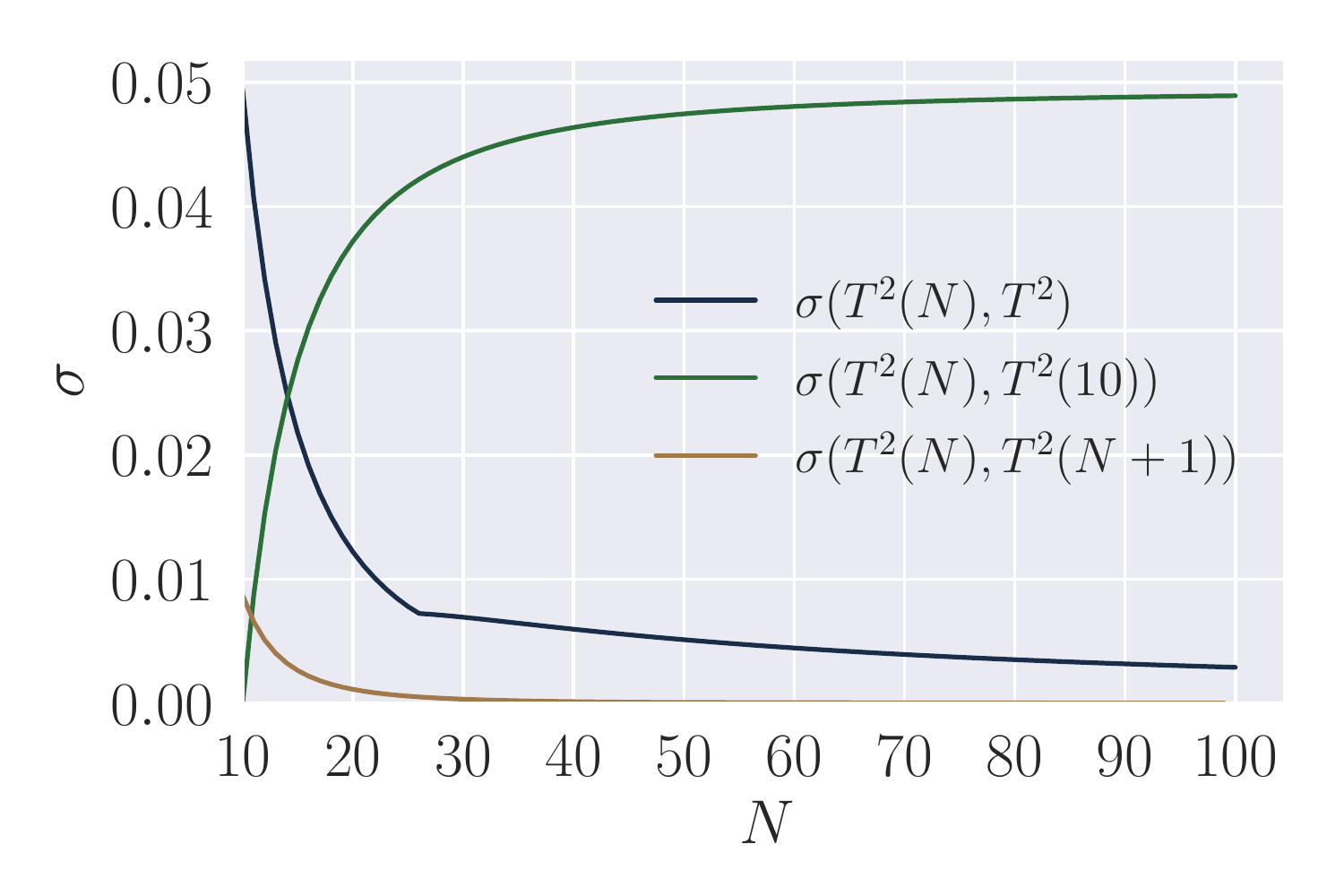}
	\caption{Spectral distances between fuzzy spaces of different matrix size $N$ and the continuum space, using $\gamma=1.5$.  The sphere is on the left and the torus on the right. The dark blue line shows that the distance of the fuzzy space to the continuum space. The green line denotes the distance of the fuzzy
space at $N=10$ to the fuzzy space at arbitrary $N$. The brown line shows the distance between the spaces at $N$ and $N+1$.}
\label{fig:checkit}
\end{figure}
Some numerical results are presented in Figure~\ref{fig:checkit}, where $\gamma=1.5$ to avoid numerical instabilities close to the singularity at $s=d/2 =1$. The left-hand plot shows the results for the fuzzy sphere and the right-hand plot for the fuzzy torus.
The plot for the torus has a noticeable kink in the line comparing the continuum torus to the fuzzy torus at $N=26$.
This feature arises because at this point the maximum of the logarithm flips from the upper end of the interval $[\gamma, \gamma+1]$ to the lower end of this interval.

\subsection{Distances between random geometries as $g_2$ varies}
To apply the distance function to random geometries, the first step is to pick a value of $\gamma$.
Since the random geometries do not have a simple dimension measure or poles in the $\zeta$ function the easiest choice is to explore the distance for a few values of $\gamma$ that would be compatible with the range of dimension found using the spectral variance.
For the exploration here $\gamma=0.5, 1.0, 1.5, 2.0$ are used, which are suitable considering the dimension values found in Figure~\ref{fig:maxVs}.
The choice of $\gamma$ influences the relative weight given to the low or high energy part of the spectrum, with lower $\gamma$ emphasising the higher energy part of the spectrum.

Calculating this on the random fuzzy spaces requires some form of averaging, either over distances or over geometries.
Calculating averages over quantities involving the zeta function is plagued by instabilities which arise from terms $\lambda_{0}^{-s}$ if the eigenvalue $\lambda_0$ with the smallest absolute value fluctuates close to $0$.
This makes it more practical to calculate the averages over the eigenvalues and then calculate the distances between these average spectra.
Using this, the distance function can be applied to random geometries, obtaining some interesting results.

A simple test is to calculate the distance between geometries of the same type at different $g_2$.
This is shown in Figure~\ref{fig:goodexamples} for type $(2,0)$, where the distance from the geometries with $g_2=-2.5$ and $g_2=-3.5$ as reference points is shown.
\begin{figure}
\subfloat[][Distance from $g_2=-2.5$]{\includegraphics[width=0.45\textwidth]{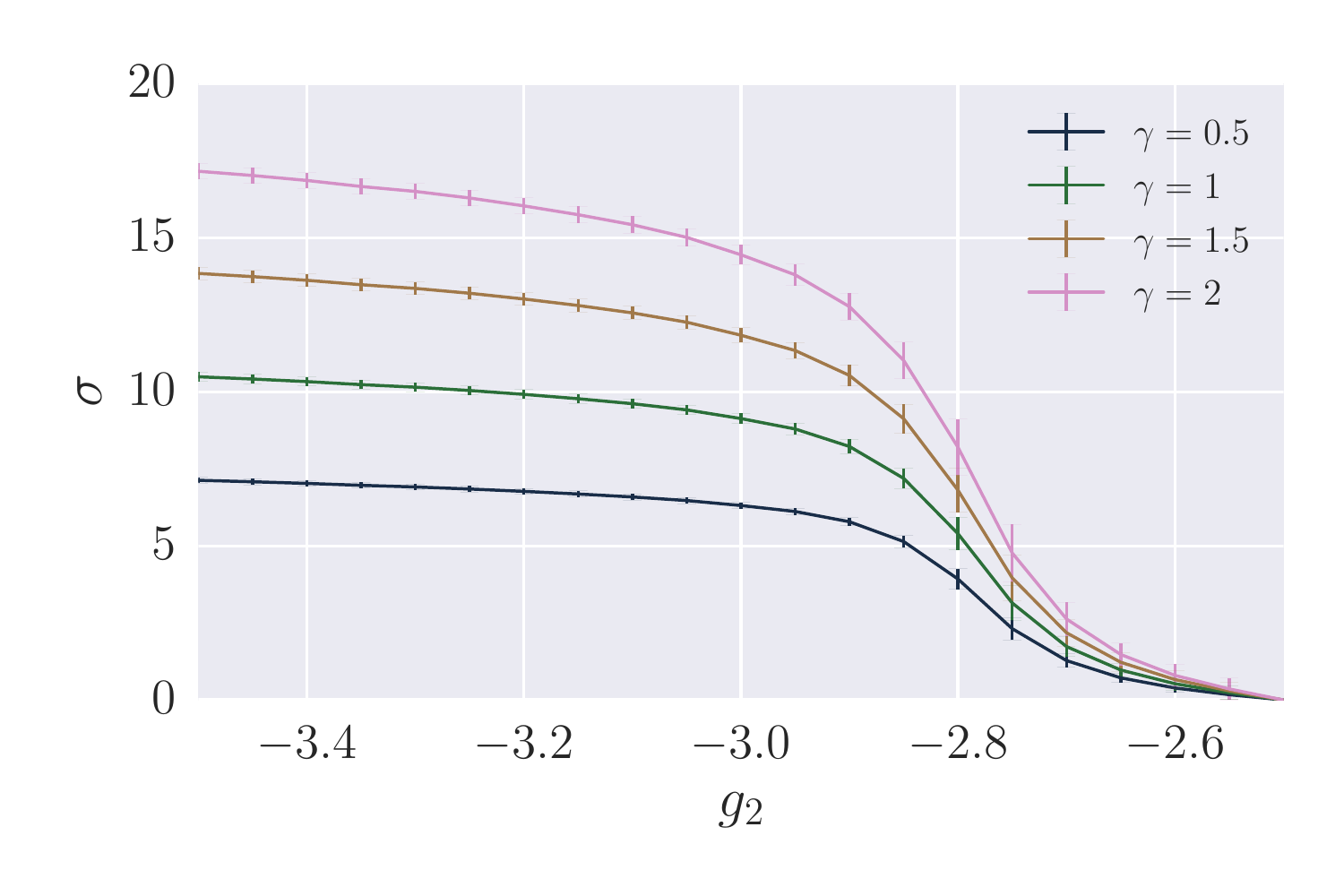}}
\subfloat[][Distance from $g_2=-3.5$]{\includegraphics[width=0.45\textwidth]{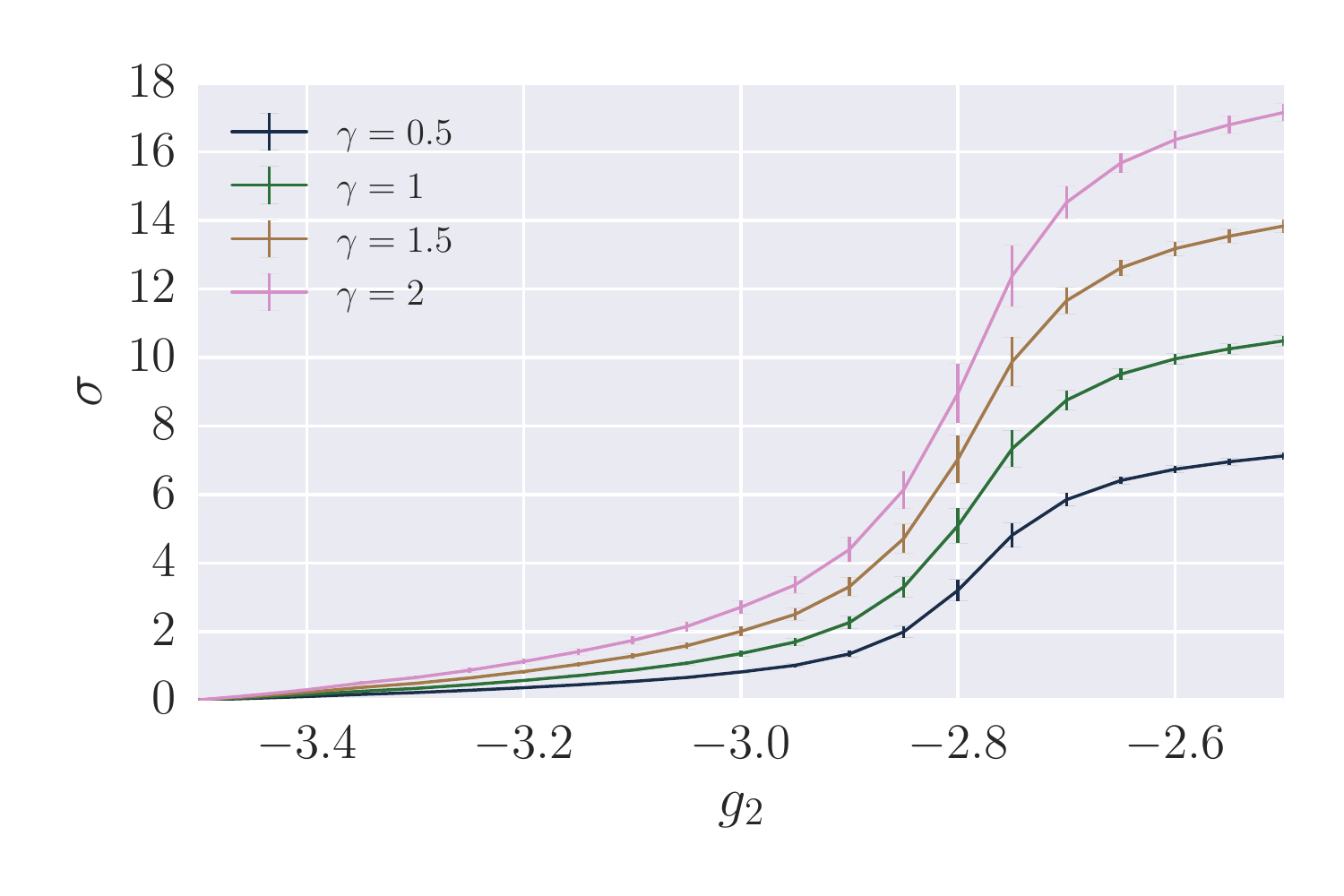}}
\caption{Distance of random geometries of type $(2,0)$ with varying $g_2$ from the geometries at $g_2=-2.5,-3.5$}\label{fig:goodexamples}
\end{figure}
The choice of $\gamma$ changes the resulting distances but the relative distances are qualitatively the same.
In particular, the geometries above the phase transition at $g_2=-2.8$ are more similar to the geometry at  $g_2=-2.5$ while those below the phase transition are more similar to the geometry at $g_2=-3.5$.
The error bars on all distance measures are calculated by propagation of uncertainty starting from the errors on the average eigenvalues.

\subsection{Measuring the distance from the fuzzy $S^2$}
The spectral distance measure can be used to investigate whether the spectra of the random geometries are close to the fuzzy sphere.
As demonstrated in~\cite{cornelissen_distances_2017}, the spectral distance measure is also sensitive to size differences between geometries.
However for the random fuzzy spaces their size is of less interest than the question of whether they resemble a fuzzy sphere of any size.
Thus to compare the geometries, the average spectra are scaled so that the maximum eigenvalues are equal.
Physically this rescaling should correspond to fixing the Planck scale for the geometries to agree, as discussed in section~\ref{subsec:EVscale}.

The resulting distances between the random geometries and the appropriate fuzzy spheres are shown in Figure~\ref{fig:rescaledDistS2lmax}.
These are encouraging, since they align with the hope that the geometries close to the phase transition are similar to the fuzzy sphere.
For all three types the minimum in the distance to the fuzzy sphere is for the geometries with $g_2$ one step away from the phase transition.
For type $(1,1)$ it is for $g>g_c$, and for types $(2,0)$ and $(1,3)$ it is for $g<g_c$.
This would hint that in the latter two cases the fuzzy sphere is approached when approaching the phase transition from the gapped phase.
\begin{figure}
\subfloat[][Type $(1,1)$ $N=10$]{\includegraphics[width=0.48\textwidth]{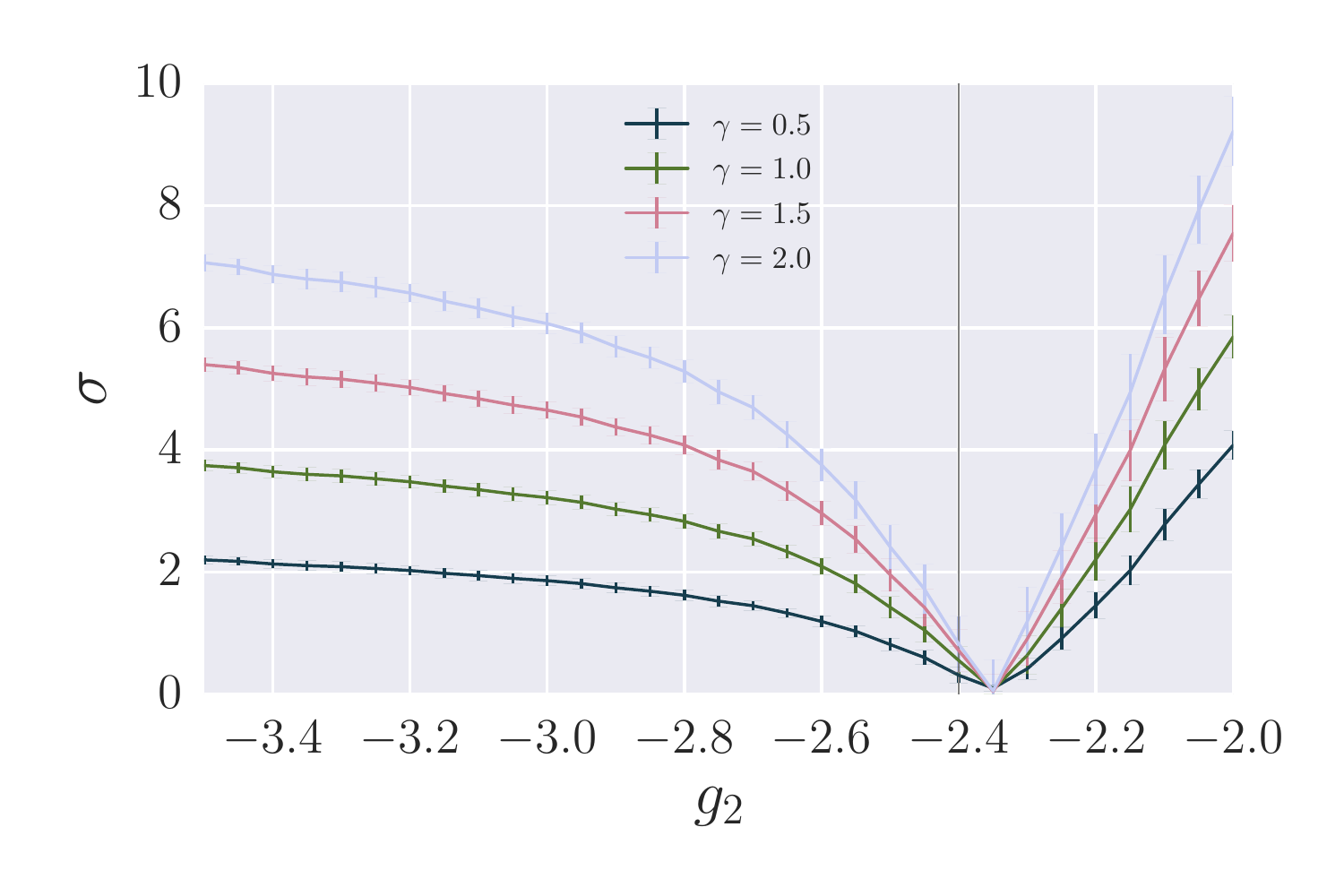}}
\subfloat[][Type $(2,0)$ $N=10$]{\includegraphics[width=0.48\textwidth]{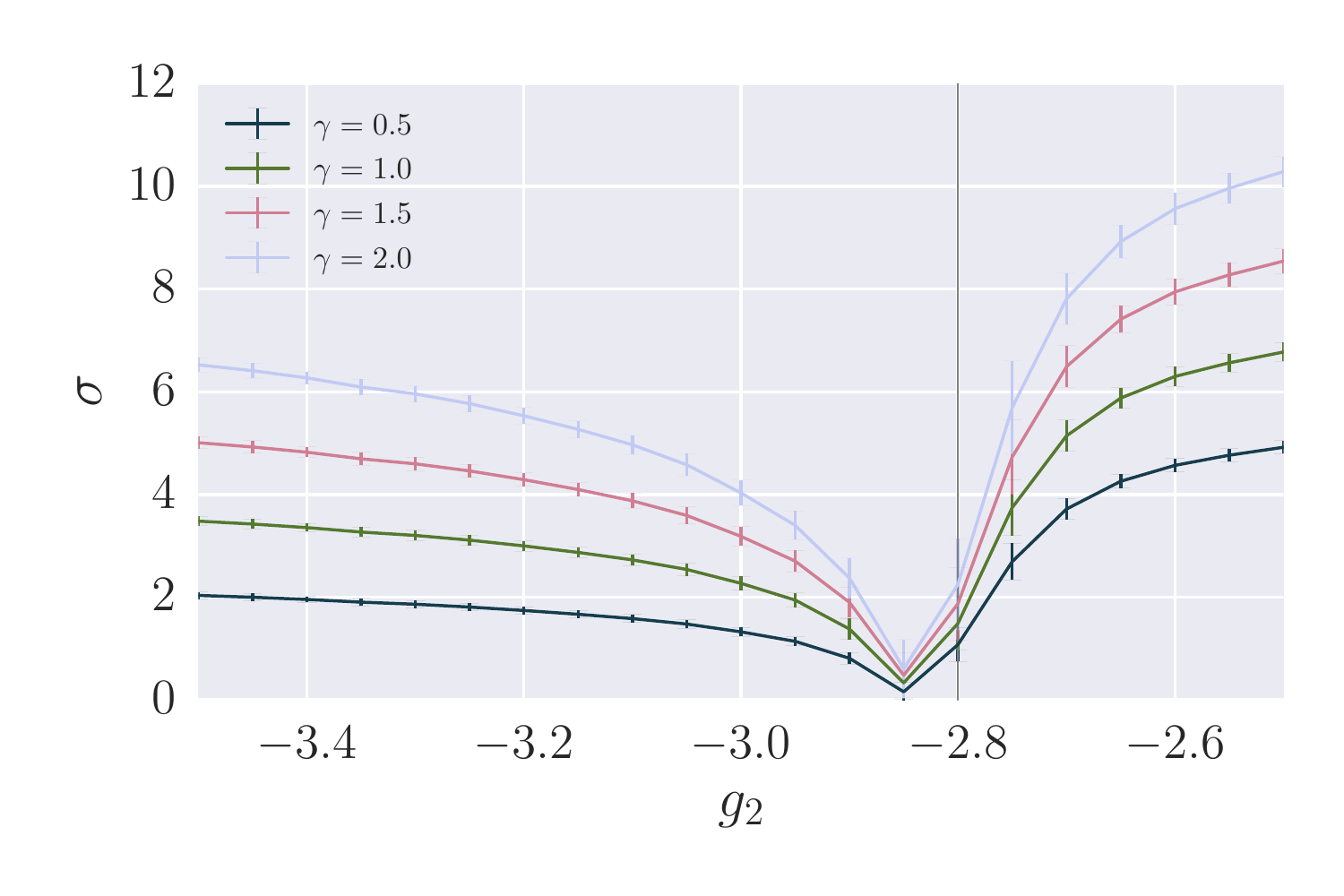}}

\subfloat[][Type $(1,3)$ $N=8$]{\includegraphics[width=0.48\textwidth]{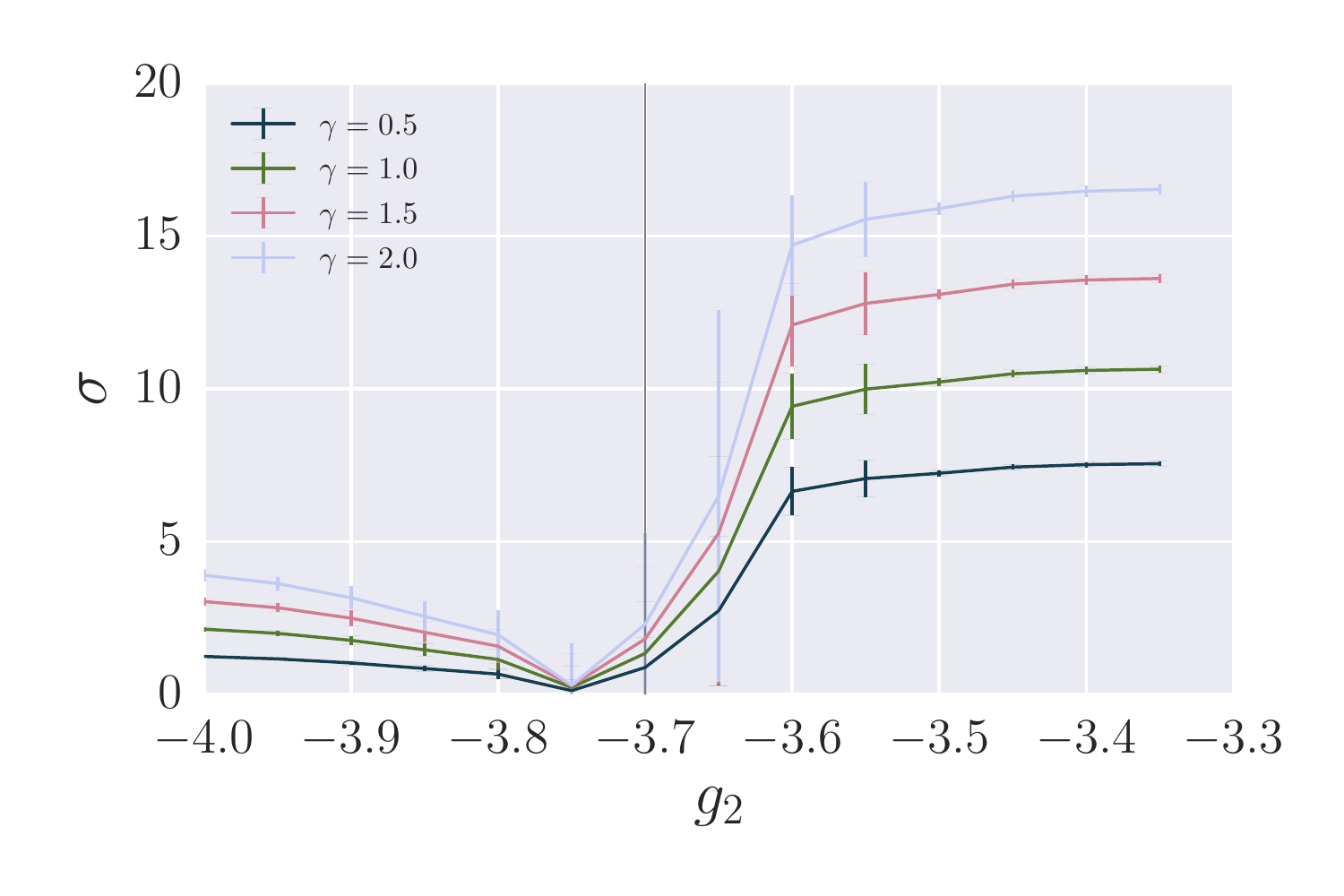}}
  \caption{Distance between the fuzzy sphere and random fuzzy geometries, with the spectra rescaled such that their maximal eigenvalues agree. For the geometries of types $(1,1)$ and $(2,0)$ only half of the multiplicity of the fuzzy sphere eigenvalues were used to compare to geometries with the same number of eigenvalues.}\label{fig:rescaledDistS2lmax}
\end{figure}
In principle it would also be interesting to rescale spaces to match their volumes.
However this would require a solution to the problems with the volume definition described in the previous section.

\subsection{Measuring the distance between type $(1,1)$ and type $(2,0)$}

One of the aims in doing this additional analysis was to understand the difference between the geometries of type $(1,1)$ and $(2,0)$ better.
For this a distance measure is very useful.
When comparing geometries of different types, in addition to choosing the rescaling, there is freedom in which $g_2$ values to compare with each other.

The simplest option is to compare the geometries at the same value of $g_2$, disregarding the fact that the geometries have different phase transition points.
The difference found in this way is dominated by whether the geometries are in the same phase or not, as seen in Figure~\ref{fig:type1120comp}(a).
The distances measured do not change much when the eigenvalues are rescaled to $\lambda_{\max}=1$. This is because the maximum eigenvalues in both geometries are very similar.

\begin{figure}
\subfloat[][At the same $g_2$]{ \includegraphics[width=0.45\textwidth]{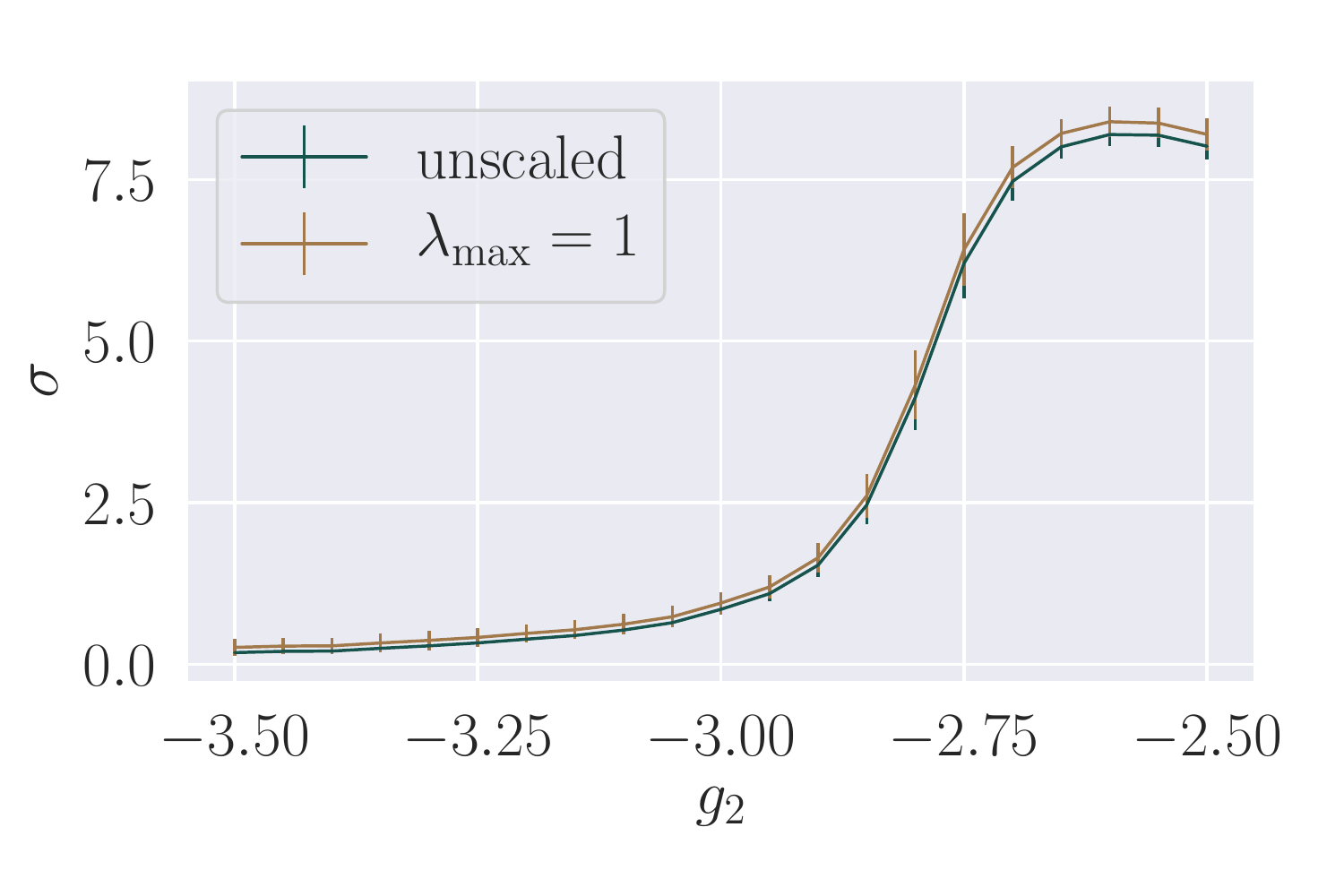}}
\subfloat[][At the same $g_2-g_c$]{ \includegraphics[width=0.45\textwidth]{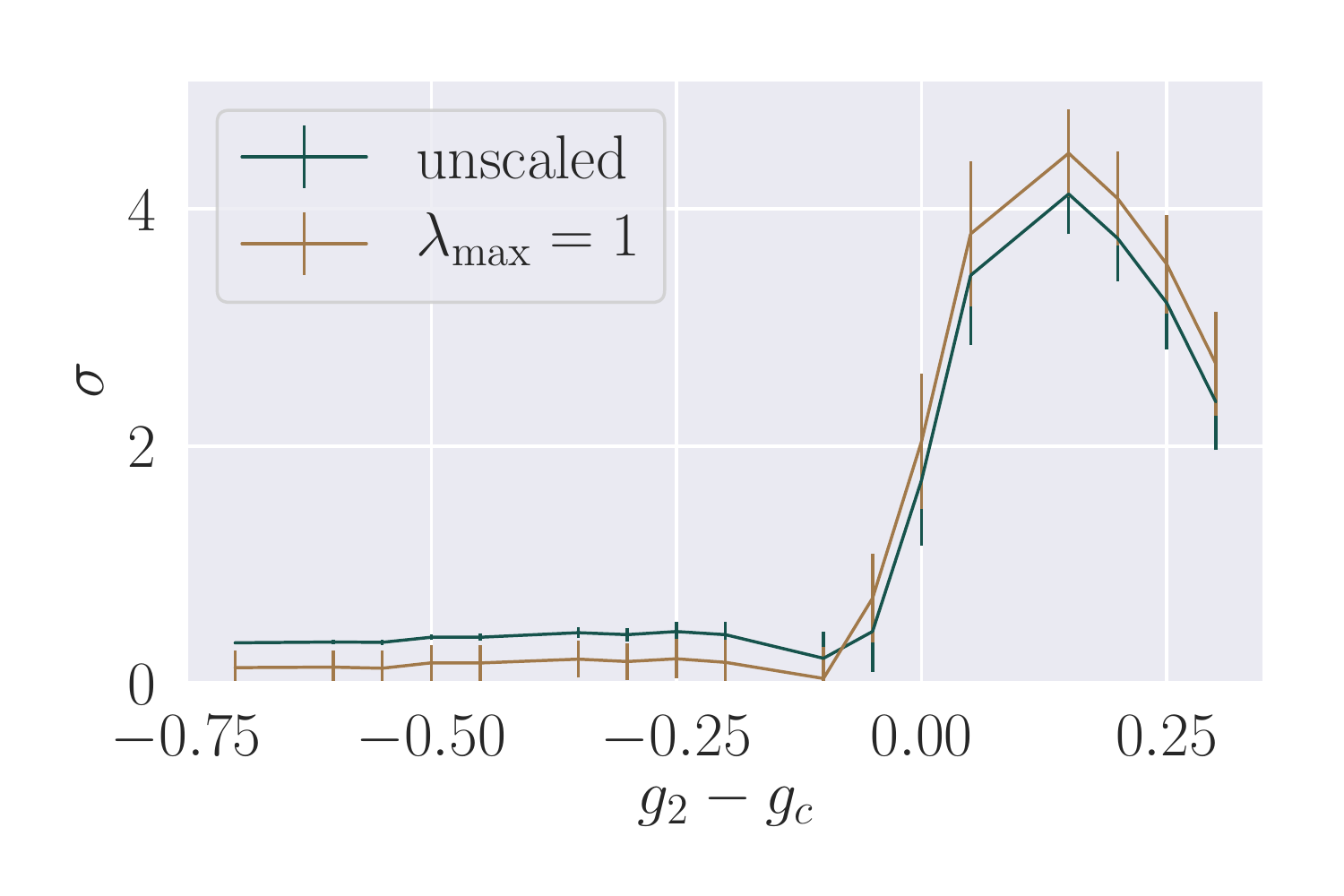}}
  \caption{The spectral distance between geometries of type $(1,1)$ and type $(2,0)$ with different rescalings. Plotted for $\gamma=1$.}\label{fig:type1120comp}
\end{figure}

Comparing at fixed $g_2$ value the difference between the geometries is largest between $g_2=-2.8, \dots ,-2.5$ which is the region in which type $(1,1)$ is already transitioned, while type $(2,0)$ is not yet.
After $g_2=-2.8$, when type $(2,0)$ transitions to the gapped phase, the distance between the geometries becomes much smaller.

The other option is to correct for this, by comparing the geometries at a constant distance from the phase transition $g_2-g_c$, shown in Figure~\ref{fig:type1120comp}(b). This requires the assumption that the determination of the phase transition for both phases is correct, which introduces an additional possible error.
The overall distance between the geometries in this measure is smaller than before.
It is also much larger when comparing the geometries in the un-gapped phase as opposed to comparing those in the gapped phase.
The geometries after the phase transition have distance close to $0$, while before the phase transition, their distance is quite large.

\section{Conclusion}\label{sec:conclusion}

This work has developed a number of tools to measure and compare spectra of the Dirac operator $D$ on fuzzy spaces. Since the spectra are finite, one does not have the luxury of using the large-eigenvalue asymptotics to measure geometric quantities. As a result, the spectral measurements depend on a choice of energy scale. This can be seen most simply by adapting Weyl's law of asymptotic eigenvalue scaling to the finite situation. Some results for the dimension of a space are plotted but, as discussed in section~\ref{sec:weyl}, it is hard to interpret the plots systematically.

The spectral dimension was originally developed to give a scale-dependent dimension measure from the eigenvalues of a scalar Laplacian. This can be adapted to the case of the operator $D^2$ but is not as useful as for the Laplacian due to the absence of a zero-mode. The new notion of spectral variance introduced here fixes this problem. It depends only on the difference in eigenvalues of $D^2$ and so is insensitive to the gap between the lowest squared eigenvalue and zero. Like the spectral dimension, for fuzzy spaces it goes to zero when the energy scale becomes infinite. This is due to the absence of very high energy modes above the cut-off energy scale, which is interpreted as the Planck scale in quantum gravity.

The spectral variance is applied to several examples of spectra and it is argued that it gives a good measure of dimension. To extract a single number for the dimension, rather than a function, it appears to be best to look at the value of the spectral variance at the point where the curve is flattest. Usually this is at the maximum, but the example of the fuzzy torus shows that this is not always the case. This exhibits an upward spike at the Planck energy scale but the much larger flat part of the curve at lower energies gives a very good agreement with the expected value 2 for the dimension.

For the random fuzzy spaces, the spectral variance shows that the geometries around the phase transition have a dimension close to the value $2$, at least at the energy scale corresponding to the maximum point on the spectral variance curve. Note that the Planck-scale spike in the spectral variance of the torus does not occur for the random geometries studied here.
Plotting the value of the maximum of the spectral variance against the coupling constant $g_2$ shows that these lines, for different values of the matrix size $N$, intersect near the phase transition value of $g_2$ for the $(2,0)$ and the $(1,3)$ geometries. This is a hint that the behaviour at the phase transition might be independent of $N$.

In principle, the volume of fuzzy spaces could be calculated using the spectral $\zeta$ function. However, the formulas defining the volume depend on the dimension of the space. Since the dimension varies with energy scale and eigenvalues at all different energy scales contribute to the volume, it is not clear what a uesful measure of volume is. For the random spaces there are additional difficulties due to the fact that the dimension varies significantly within the ensemble. This is shown in Figure~\ref{fig:vscomp}. It is an interesting question as to whether the variance of the dimension reduces for larger values of $N$ than those investigated here.

Comparing the ratios of zeta functions gives very good results for the comparison of geometries via their spectra. It is even possible to compare an infinite spectrum with a finite one. In particular, the results for the distance between the fuzzy sphere and the average eigenvalues of some random ensembles are very striking. These spectral distances become very small when the random geometries are near to their phase transitions. This give a quantitative estimate of the similarity that was previously only apparent qualitatively from plots of the eigenvalue distribution or the spectral variance. The spectral distance is also useful in comparing the random geometries at different values of the coupling constant, giving a characterisation of the different phases.

There are other methods that could be used to examine the convergence of fuzzy spaces to continuum spaces, see~\cite{latrmolire2015, rieffel2004gromov} and the references therein. These are based around comparing the spaces as quantum compact metric spaces and using a generalisation of the Gromov-Hausdorff distance. There seems to be no obvious way to implement this method numerically and moreover, there is no obvious role for the spectrum of the Dirac operator. Thus it is hard to compare their method with any of our results.

In summary, this work establishes several new tools to explore random fuzzy spaces, extending the understanding of the geometries described in~\cite{barrett_monte_2015,glaser_scaling_2016}. In future it should be possible to develop more efficient Monte Carlo codes that will allow the use of the new tools on much larger matrices. Several interesting new phenomena are flagged for further study. It would also be useful to develop ways to characterise the coordinates and symmetries of fuzzy spaces to exploit more of the information contained in the spectral triple than just the spectrum of the Dirac operator. The spectral measures developed in this paper help to identify the fuzzy spaces that are worthy of further study.

\subsection*{Acknowledgements}
The simulations and some of the analysis presented in this paper owe their thanks to the University of Nottingham High Performance Computing Facility.

PD was supported by the Engineering and Physical Sciences Research Council grant number EP/MX50810X/1.

LG has been supported by the People Programme (Marie Curie Actions) H2020 REA grant agreement n.706349 `Renormalisation Group methods for discrete Quantum Gravity' during parts of this work.

\appendix
\section{Convergence proofs}\label{sec:appendix}

This appendix shows the convergence of the spectrum of the Dirac operator for the square fuzzy torus $a=d=1$, $b=c=0$ to the spectrum of the Dirac operator of the corresponding continuum torus.

First, there is a general result that if the zeta functions converge uniformly, then the distance converges to zero. As before, $\gamma$ is greater than the real part of any pole so that the zeta functions are convergent series with positive terms.
\begin{lemma} Let $\zeta\ne0$.
        If $\zeta_n\to\zeta$ as $n\to\infty$ uniformly in $[\gamma,\gamma+c]$ for any constant $c>0$, then $\sigma(\zeta_n,\zeta)\to0$.
\end{lemma}
\begin{proof}
    Suppose $|\zeta_n(s)-\zeta(s)|<\epsilon$ for all $n>n_0$ and all $s\in[\gamma,\gamma+c]$. Then
    $$\left|\log\frac{\zeta_n}{\zeta}\right|\le\left|\frac{\zeta-\zeta_n}{\zeta}\right|\le\frac{\epsilon}{\zeta}.$$
    Hence $\sigma(\zeta_n,\zeta)\le \epsilon\,\sup_s{1/\zeta(s)}={\epsilon/\min_s\zeta(s)}$. Hence the distance $\sigma$ converges to zero.
\end{proof}

For the square fuzzy torus with $a=d=1$, $b=c=0$, the spin structure is $\Sigma_c = (1,1)$. Therefore the same spin structure is used for the continuum torus and its eigenvectors are  labelled with $(k,l)\in A_\infty=\Z^2+(1/2,1/2)$. Let $A_N\subset A_\infty$ be the subset such that $-N/2\le k,l < N/2$. This indexes the eigenvectors of the fuzzy torus of size $N$ exactly once each, ignoring the fourfold degeneracy of the squared eigenvalues for each $(k,l)$ (see equation~\eqref{eq:fuzzyTorus}).

The difference in the two zetas is
$$
\Delta\zeta= \sum\limits_{(k,l)\in A_N} ([k]^2 + [l]^2)^{-s} - \sum\limits_{(k,l)\in A_\infty} (k^2 + l^2)^{-s}.
$$
The remainder of this appendix is a proof that $\Delta\zeta\to0$ as $N\to \infty$, uniformly on $[\gamma,\gamma+1]$, for any $\gamma>1$.

To show this convergence, the comparison is split into two regions. The fuzzy zeta spectrum is similar in value to the continuum torus for small $[k],[l]$ but for large $[k],[l]$ the spectra differ quite drastically. This can be seen in Figure~\ref{fig:TorFuzzyComm}. Define $A_M\subset A_N$ to be the region with $|k|, |l|<M$ for some $M$ such that  $1\le M\le N/2$. Then different comparisons are used for $A_M$ and $A_N\setminus A_M$.
The difference in the two zetas is thus
\begin{align}
\Delta\zeta
&= \sum\limits_{A_M} ([k]^2 + [l]^2)^{-s} - (k^2+ l^2)^{-s}  + \sum \limits_{A_N\setminus A_M} ([k]^2 + [l]^2)^{-s} - \sum\limits_{A_\infty\setminus A_M} (k^2 + l^2)^{-s}
\label{eq:TorDeltaZeta}
\end{align}

These three sums are investigated in turn. For the first term,
\begin{align}
([k]^2+[l]^2)^{-s} - (k^2+ l^2)^{-s} &= \left(\frac{\sin^2 x+\sin^2 y}{\sin^2z}\right)^{-s}-\left(\frac{x^2+y^2}{ z^2}\right)^{-s}
\end{align}
with $x=\pi k/N$, $y=\pi l/N$, $z=\pi/N$. The Taylor expansion for this is
\begin{align}
\left(\frac{\sin^2 x+\sin^2 y}{\sin^2z}\right)^{-s}=\left(\frac{x^2+y^2}{ z^2}\right)^{-s}\biggl(1+s\bigl(O(x^2+y^2)+O(z^2)\bigr)\biggr)
\end{align}
providing $x$, $y$ and $z$ are in a sufficiently small region. This requires that $M/N$ be sufficiently small.

Putting this estimate into the sum gives
\begin{align}
  \sum\limits_{A_M} ([k]^2 + [l]^2)^{-s} - (k^2+ l^2)^{-s}  &\leq s \cdot \sum_{A_M}  O\left( \frac{(k^2+l^2)^{1-s}}{N^2} \right)\le s \,O(M^2/N^2).
\end{align}
The second inequality follows from the fact that $s>1$ and so the summand is greatest for  $k,l=\pm 1/2$. This shows that as long as $M/N\to 0$, the first piece of expression~\eqref{eq:TorDeltaZeta} converges uniformly to zero.

The second term in Eq~\eqref{eq:TorDeltaZeta} is examined now. Let $\sinc(x)=(\sin x)/x\le 1$. The term is
\begin{align}
    \sum\limits_{ A_N \setminus A_{M}} ([k]^2 + [l]^2)^{-s}
& \leq N^2 \left( \frac{ \sin \frac{\pi M}{N}}{\sin\frac{\pi}{N}} \right)^{-2s}
= N^2 M^{-2s}  \left( \frac{\sinc\frac{\pi M}{N} }{\sinc\frac{\pi}{N}}\right)^{-2s}
\leq  N^2 M^{-2\gamma}   \frac{\left(\sinc\frac{\pi M}{N} \right)^{-2(\gamma+1)}}
{\left(\sinc\frac{\pi}{N}\right)^{-2\gamma}}
\end{align}
Thus as long as $M$ is chosen so that $NM^{-\gamma}\to 0$ as $N\to\infty$ as well as the previous condition $M/N\to0$, the bound on the right converges to zero.
Hence this sum also converges uniformly to zero.

The last piece is the one arising from the continuum zeta. This converges to zero as long as $M\to\infty$ since it is the tail of a convergent series. The convergence is uniform since
\begin{align}
    \sum_{A_\infty \setminus A_M} \left(k^2+l^2 \right)^{-s} \leq \sum_{A_\infty \setminus A_M} (k^2 +l^2)^{-\gamma},
\end{align}
a bound independent of $s$.

All of the conditions on $M$ can be satisfied by taking $M=(N/2)^a$ for $1/\gamma<a<1$.

\bibliography{ref}
\end{document}